\documentclass[a4paper,USenglish,numberwithinsect,cleveref,autoref,thm-restate]{lipics-v2021}

\def\isarxiv{} 
\ifdefined\isarxiv
\pdfoutput=1
\hideLIPIcs
\fi

\title{A Coalgebraic Dijkstra Algorithm
} 

\author{Takahiro Sanada}{Fukui Prefectural University, Fukui, Japan \and \url{https://www.kurims.kyoto-u.ac.jp/~tsanada/} }{tsanada@fpu.ac.jp}{https://orcid.org/0000-0003-3409-6963}{Supported by JSPS KAKENHI Grant No.\ JP24K23867.}
\author{Yoàv Montacute}{National Institute of Informatics, Tokyo, Japan \and \url{https://www.yoavmontacute.com/}}{montacute@nii.ac.jp}{https://orcid.org/0000-0001-9814-7323}{Supported by JST ACT-X Grant No.\ JPMJAX24CR.}
\author{Kittiphon Phalakarn}{National Institute of Informatics, Tokyo, Japan \and \url{https://kphalakarn.github.io/}}{kphalakarn@nii.ac.jp}{https://orcid.org/0009-0006-5406-7480}{}
\author{Ichiro Hasuo}{National Institute of Informatics, Tokyo, Japan \and SOKENDAI (The Graduate University for Advanced Studies), Kanagawa, Japan \and Imiron Co., Ltd., Tokyo, Japan}{kphalakarn@nii.ac.jp}{https://orcid.org/0009-0006-5406-7480}{}

\authorrunning{T.\ Sanada, Y.\ Montacute, K.\ Phalakarn, I.\ Hasuo} 

\Copyright{Takahiro Sanada, Yoàv Montacute, Kittiphon Phalakarn, Ichiro Hasuo} 


\category{} 

\relatedversion{} 

\acknowledgements{The authors are supported by JST ASPIRE Grant No.\ JPMJAP2301.}

\nolinenumbers 

\EventEditors{John Q. Open and Joan R. Access}
\EventNoEds{2}
\EventLongTitle{42nd Conference on Very Important Topics (CVIT 2016)}
\EventShortTitle{CVIT 2016}
\EventAcronym{CVIT}
\EventYear{2016}
\EventDate{December 24--27, 2016}
\EventLocation{Little Whinging, United Kingdom}
\EventLogo{}
\SeriesVolume{42}
\ArticleNo{23}

\usepackage{amsmath, amssymb, stmaryrd, mathtools}
\usepackage{wrapfig}
\usepackage[noadjust]{cite}
\usepackage{booktabs, colortbl}
\usepackage{tikz}
\usetikzlibrary{cd,automata,positioning,decorations.pathmorphing,decorations.pathreplacing,calc}
\newcommand{\tikzmark}[1]{\tikz[overlay,remember picture] \node (#1) {};}
\newcommand*{\AddNote}[4]{%
    \begin{tikzpicture}[overlay, remember picture]
        \draw [decoration={brace,amplitude=0.5em},decorate,ultra thick,gray]
            ($(#3)!(#1.north)!($(#3)-(0,1)$)$) --  
            ($(#3)!(#2.south)!($(#3)-(0,1)$)$)
                node [align=center, text width=2.5cm, pos=0.5, anchor=west] {#4};
    \end{tikzpicture}
}%
\tikzset{every state/.style={minimum size=1pt, inner sep=1pt}}
\tikzset{
    dot/.style = {circle, fill, minimum size=#1,inner sep=0pt, outer sep=0pt},
    dot/.default = 3pt 
}
\usepackage{algorithm}
\usepackage{algpseudocode}
\newcommand{\letin}[2]{#1 \leftarrow #2}

\usepackage[disable]{todonotes} 
\newcommand{\todoil}[1]{\todo[inline,caption={}]{
#1
}}

\ifdefined\isarxiv
\newcommand{\versionswitch}[2]{#2}
\else
\newcommand{\versionswitch}[2]{#1}
\fi

\newtheorem{problem}[theorem]{Problem}

\newtheorem{notations}[theorem]{Notations}

\newcommand{\blank}{{-}}

\newcommand{\idmor}{\mathrm{id}}
\newcommand{\co}{\mathbin{\circ}}
\newcommand{\WGG}[1]{\mathcal{W}_{G}}
\newcommand{\BB}{\mathbb{B}}
\newcommand{\btt}{\mathbf{t}}
\newcommand{\bff}{\mathbf{f}}
\newcommand{\NN}{\mathbb{N}}
\newcommand{\NNinf}{\NN^{\infty}}
\newcommand{\RR}{\mathbb{R}}
\newcommand{\RRpos}{\RR_{\ge 0}}
\newcommand{\RRinf}{\RR^{\infty}}
\newcommand{\RRinfpos}{\RR^{\infty}_{\ge 0}}
\newcommand{\RRpminf}{\RR^{\pm \infty}}

\newcommand{\Bellman}[3]{\Phi^{#1, #2}_{#3}}

\newcommand{\Sets}{\mathbf{Set}}

\newcommand{\idfunc}{\mathrm{Id}}
\newcommand{\powset}{\mathcal{P}}
\newcommand{\powsetfin}{\powset_{\mathrm{fin}}}
\newcommand{\powsetfinne}{\powsetfin^{\mathrm{ne}}}
\newcommand{\distr}{\mathcal{D}}

\newcommand{\pullback}{\mathsf{pb}}
\newcommand{\weakpullback}{\mathsf{wpb}}

\newcommand{\finalweight}{\xi}

\newcommand{\gfp}{\nu}

\newcommand{\runpath}{\mathrm{RP}}
\newcommand{\runtree}{\mathrm{RT}}
\newcommand{\height}{h}

\newcommand{\predecessor}{\mathrm{Pred}}
\newcommand{\successor}{\mathrm{Succ}}

\newcommand{\order}{\mathcal{O}}
\newcommand{\sizeof}[1]{\#{#1}}

\newcommand{\supp}{\theta}

\newcommand{\tree}{\mathcal{T}}

\newcommand{\CSPP}{CSPP}

\begin{document}

\maketitle

\begin{abstract}
    The Dijkstra algorithm is a classical method for solving the shortest path problem on weighted graphs.
    There are several variations of the Dijkstra algorithm, including algorithms for the widest path problem and for two-player games.
    In this paper, we introduce the \emph{coalgebraic shortest path problem} (CSPP), a unifying framework for a broad class of optimization problems on state-transition systems.
    This framework encompasses not only the aforementioned problems but also new ones such as the shortest binary tree problem.
    We further present a \emph{coalgebraic Dijkstra algorithm} for solving the CSPP efficiently under a suitable condition.
    Our condition is necessary and sufficient for the algorithm to return correct solutions, thereby providing a precise criterion for when Dijkstra-style acceleration is possible.
    We also show that the proposed algorithm achieves asymptotic complexity comparable to that of the classical Dijkstra algorithm.
    
\end{abstract}

\section{Introduction}\label{sec:intro}
The shortest path problem (SPP) on a weighted directed graph is a classical problem.
The Bellman--Ford algorithm \cite{Bellman1958,Ford1956} solves the SPP 
 by iteratively updating the distances of vertices.
This update can be viewed as an operator $\Phi \colon L \to L$ for a lattice $L$, and the solution to the SPP is the greatest fixed point $\gfp \Phi$ of $\Phi$ \cite{Misra2001}.
The Dijkstra algorithm \cite{Dijkstra1959} is an improvement on the Bellman--Ford algorithm for weighted graphs without negative weights.
The Dijkstra algorithm restricts updates to vertices belonging to a specific set and avoids updating vertices that are expected to have already achieved the minimum value.

There are several variations of the Dijkstra algorithm.
Reachability is a special case of the SPP, thus amenable to the Dijkstra algorithm.
The widest path problem, which seeks a path with maximum capacity, is  solvable by a modified version of the Dijkstra algorithm \cite{Pollack1960}.
A class of two-player games can also be solved by a Dijkstra-style algorithm \cite{BardiLopez2016}.

It is important that, aside from these \emph{positive} variations, there are many known \emph{negative} variations---those optimization problems based on state transition systems that do \emph{not} allow the application of the Dijkstra-style acceleration. Known examples of such negative variations include the SPP with negative edges, and the reachability probability in a probabilistic system such as a Markov chain. 

An overview of such variations is given in Table~\ref{table:examples-problems}, where the last column shows the applicability of the Dijkstra-style acceleration. While many previous works establish positive results~\cite{Pollack1960,Wing1961,Mazala2002,BardiLopez2016}, we do see some negative entries. 

Moreover, the table shows that the distinction between positive and negative can be  subtle. For example (the sixth and seventh in Table~\ref{table:examples-problems}), when the SPP comes with a (multiplicative) rate, 1) Dijkstra is applicable when the rate is $\ge 1$ (thus the rate represents \emph{interest}), while 2) it is not applicable when the rate is $\le 1$ (the rate represents \emph{discount}). Another subtle difference is observed for dynamic games with a discount rate $r$ (the second last in Table~\ref{table:examples-problems}), where Dijkstra is applicable only when the stepwise reward is constant (with $\ell_{0}=L$). 

Therefore, in this paper, we are interested in the following questions.
\begin{enumerate}
    \setlength{\parskip}{0cm}
    \setlength{\itemsep}{0cm}
    \item What is the essence of the Dijkstra-style acceleration of fixed point computation?
    \item What is the condition that makes a problem amenable to the Dijkstra-style acceleration? 
    \item What is a good mathematical level of abstraction for answering these questions?
\end{enumerate}
Our answer to the last (meta-level) question is the language of \emph{category theory}---the language of \emph{coalgebras}~\cite{Rutten2000,Jacobs2016}, in particular, as a categorical model of state-space transition systems. We focus on a coalgebraic signature functor of the form $F = \BB \times \powsetfin(G\blank)$, where $\BB=\{ \btt, \bff \}$ tells if a state is a target or not,  and the finite power set functor $\powsetfin$ represents a finite branching over $G$-transitions. Besides, we use a poset $\Omega$ equipped with a \emph{modality} $\sigma\colon G\Omega\to\Omega$---describing how weights are accumulated along $G$-transitions. This is much like in predicate transformer semantics for coalgebras~\cite{Hasuo15CMCSJournVer,AguirreKK22}.  All these give us a categorical framework that accommodates various problems in Table~\ref{table:examples-problems} (\S\ref{sec:minimum-weight-gfp}). 

Within this categorical framework, we formalize the \emph{coalgebraic Dijkstra algorithm} (Algo.~\ref{alg:coalgebraic-dijkstra}), embodying our understanding of ``the essence of the Dijkstra algorithm'' and addressing the first question above. Moreover, we identify a necessary and sufficient condition for the correctness of the coalgebraic Dijkstra algorithm, formulated in the language of category theory (Thm.~\ref{thm:correctness-G}), thus addressing the second question. In the course of this venture, we find that our categorical/coalgebraic abstraction level is the right one.

It turns out that the abstract categorical framework is concrete enough, too, so that we can conduct complexity analysis and discuss algorithmic improvements. Indeed, in~\S\ref{subsec:complexity}, we 1) establish the asymptotic complexity of the coalgebraic Dijkstra algorithm, 2) show that it coincides with the classical Dijkstra algorithm for the original problem of the SPP, and 3) propose an improvement by a Fibonacci heap, for the general coalgebraic algorithm.

We summarize our contributions  as follows.
\begin{enumerate}
    \setlength{\parskip}{0cm}
    \setlength{\itemsep}{0cm}
    \item We formulate the \emph{coalgebraic shortest path problem} (CSPP, \S\ref{sec:minimum-weight-gfp}), using coalgebras of the type  $\gamma \colon X \to \BB \times \powsetfin(GX)$ and transition modalities $\sigma\colon G\Omega\to\Omega$, and formalizing the objective as the greatest fixed point of the \emph{Bellman operator} $\Bellman{G}{\sigma}{\gamma} \colon [X, \Omega] \to [X, \Omega]$ suitably defined using $\sigma$ and $\gamma$. This framework accommodates various problems (Table~\ref{table:examples-problems}), covering both path-like and tree-like settings (such as \emph{shortest binary tree} and some games). 

    \item We present the \emph{coalgebraic Dijkstra algorithm} for the CSPP (\S\ref{sec:coalg-dijkstra-algorithm}),
    and identify a necessary and sufficient condition on a transition modality $\sigma \colon G \Omega \to \Omega$ for the  algorithm to be correct. This explains the subtleties in the last ``applicability'' column of Table~\ref{table:examples-problems}. 
    \item We establish complexity of the coalgebraic Dijkstra algorithm (\S\ref{subsec:complexity}). We also present its improvement with Fibonacci heaps.

\end{enumerate}

\paragraph*{Related work}
    The Bellman--Ford algorithm \cite{Ford1956,Bellman1958} solves the SPP of a single source or multiple sources.
        The Dijkstra algorithm \cite{Dijkstra1959} is an improvement of the Bellman--Ford algorithm.
        Misra \cite{Misra2001} showed that the SPP is equivalent to finding the greatest fixed point of an operator.
        Our approach is a coalgebraic generalization of Misra's observation.
    
    There are several approaches for generalizing path problems on weighted graphs \cite{Lehmann1977,Mohri2002,Sobrinho2002,GondranMinoux2008}.
        These approaches introduce algebraic structures, such as semirings, dioids and routing algebras, in order to capture weight of paths.
        Typically, these algebraic structures have a carrier set and a binary operation on the carrier set to accumulate edge weights along a path.
        Instead of such a binary operation, we use a categorical algebra $\sigma \colon G \Omega \to \Omega$, where $\Omega$ is a carrier set and $G$ is a functor describing the ``arity'' of $\sigma$.
        Owing to this generalization, our framework can extend paths to tree structures such as game trees or  shortest trees.
    
    Coalgebras are the dual notion of algebras and are widely used as a generalization of state-transition systems \cite{Rutten2000,Jacobs2016}.
        Many algorithms for graph-like structures have been generalized to algorithms for coalgebras---for example, coalgebraic model checking \cite{KojimaCirstea2025,KoriWatanabe2025}, coalgebraic automata learning \cite{JacobsAlexandra2014,SimoneKupkeRot2019}, and the computation of bisimilarity for coalgebras \cite{Difel+2019,JacobsWissmann2023,Sanada+2024}.

\paragraph*{Organization of the paper} 
In \S\ref{sec:prelim}, we review the classical Dijkstra algorithm and shed a  lattice-theoretic view on it. In \S\ref{sec:minimum-weight-gfp}, we formulate the shortest path problem coalgebraically as a coalgebraic shortest path problem (CSPP), and
exhibit its instances. 
In \S\ref{sec:coalg-dijkstra-algorithm}, we introduce the coalgebraic Dijkstra algorithm, presenting its correctness condition (\S\ref{subsec:fin-supp}) and complexity studies (\S\ref{subsec:complexity}). 

Many proofs and details are deferred to
\ifdefined\isarxiv appendices \else appendices in \cite{CoalgDijkstraArXiv} \fi.
\todo{IN the concur version, add citation}
The general theory in \S\ref{sec:coalg-dijkstra-algorithm} may not be the most intuitive due to its abstraction; in~\versionswitch{\cite[\S{}C]{CoalgDijkstraArXiv}}{\S\ref{sec:conditions-specific-functors}}, we make a more pedagogical presentation, starting from specific classes of $G$, gradually abstracting  to general $G$. 

% ...


\paragraph*{Notations}\label{subsec:notations}
We write $\Sets$ for the category of sets and maps.
We write $\NN$ for the set of natural numbers including zero, $\NNinf = \NN \cup \{ \infty \}$, $\RR$ for the set of real numbers, $\RRinf = \RR \cup \{ \infty\}$, $\RRinfpos = [0, \infty]$, $\RRpminf = [-\infty, \infty]$ and $\BB = \{ \btt, \bff \}$.
The identity functor, the finite power set functor, on $\Sets$, are denoted by $\idfunc \colon \Sets \to \Sets$, $\powsetfin \colon \Sets \to \Sets$, respectively.
Given sets $X_0, X_1, \dots, X_n$, we denote the $i$-th projection by $\pi_i \colon X_0 \times \dots \times X_n \to X_i$.
For a set $X$, we write $X^*$ for the set of finite strings over $X$. 
The empty string is denoted by $\epsilon$.
For a natural number $n \in \NN$, the finite set $\{ 0, \dots, n - 1\}$ of $n$ elements is denoted by $[n]$.
The \emph{function space} $[X,Y]$ is the set of functions from $X$ to $Y$. 

\begin{definition}
    The \emph{non-empty finite power set functor} $\powsetfinne \colon \Sets \to \Sets$ is defined by $\powsetfinne(X) = \powsetfin(X) \setminus \{ \emptyset \}$ 
    for $X \in \Sets$.
    Its action on morphisms is given by direct images.\todo{Yoav: This should be clarified given the community.}

    The (discrete, finite-support) \emph{distribution functor} $\distr \colon \Sets \to \Sets$ is the functor defined by 
    \[\textstyle
    \distr(X)
    = \left\{ \mu \colon X \to [0, 1] \mid \sizeof{\{ x \mid \mu(x) > 0\}} < \infty \text{ and } \sum_{x \in X} \mu(x) = 1\right\}
    \]\todo{\checkmark Yoav: probably clear with cardinality notation}
    for $X \in \Sets$, where $\sizeof{S}$ denotes the cardinality of a set $S$.
    Its action on morphisms is 
    $\distr(f)(\mu)(y) = \sum_{x \in f^{-1}(y)} \mu(x)$
    for $f \colon X \to Y$, $\mu \in \distr(X)$ and $y \in Y$.
\end{definition}

\section{Preliminaries}\label{sec:prelim}\todo{Yoav: The preliminaries would benefit from a category theory introduction, in particular an exposition on functors.}

We are interested in certain greatest fixed points; their existence is guaranteed by the following well-known result.

\todo{Yoav: I don't know how comfortable the FM community is with terms like transfinite. Ichiro? Ichiro: it's fine!}  
\begin{lemma}[{Cousot--Cousot~\cite[Cor.~3.3]{CousotCousot1979}}]\label{lem:CC}
    Let $(L ,\sqsubseteq)$ be a poset, and $\Psi\colon L \to L $ be a monotone function. Let $L$ have  infimum $\bigsqcap S$ of every $S \subseteq L $. 
We define a transfinite sequence 
\begin{equation}\label{eq:CC}
 \top
 \;\sqsupseteq\; \Psi(\top)
 \;\sqsupseteq\; \Psi^{2}(\top)
 \;\sqsupseteq\;\cdots
 \;\sqsupseteq\; \Psi^{\alpha}(\top)
 \;\sqsupseteq\;\cdots,
\end{equation} 
where $\alpha$ is an arbitrary ordinal, by the following transfinite induction: 
    $\Psi^0(\top) = \top$, 
    $\Psi^{\alpha + 1}(\top) = \Psi(\Psi^{\alpha}(\top))$ and
    $\Psi^{\alpha}(\top) = \bigsqcap_{\beta < \alpha} \Psi^{\beta}(\top)$
    where $\alpha$ is a limit ordinal.
    Then, the sequence~\eqref{eq:CC} eventually stabilizes and its limit is the greatest fixed point (gfp) $\nu \Psi$ of $\Psi$. Precisely, 
there is an ordinal $\alpha$ such that $\Psi^{\alpha}(\top) = \Psi^{\alpha+1}(\top) = \cdots =\gfp \Psi$.
\end{lemma}
While the sequence \eqref{eq:CC} is transfinite in general, our coalgebraic Dijkstra algorithm is guaranteed to terminate. See Prop.~\ref{prop:termination} and \S\ref{subsec:complexity} for complexity analysis.

 The famous \emph{Kleene theorem} can be thought of as a variant of Lem.~\ref{lem:CC}, where 1) $\Psi$ is assumed to be cocontinuous (preserving suitable infimums) and 2) the stabilization is guaranteed at the ordinal $\omega$. We use the Cousot--Cousot theorem because we do not require $\Psi$'s cocontinuity in our general framework. 


\subsection{Dijkstra and Bellman--Ford via order-theoretic fixed points}
\label{sec:dijkstra-and-bellman--ford}
We review the classical Dijkstra algorithm, and the Bellman--Ford algorithm as its prototype.
We introduce necessary order-theoretic preliminaries, too.\todo{Yoav: this sentence doesn't read well. Ichiro: hopefully it's fixed}

Let $(X,E,w,T)$ be a \emph{weighted graph}, where $X$ is the \emph{state space}, $E\subseteq X\times X$ collects \emph{edges}, and $w\colon E\to \RRpos$ assigns \emph{weights} to  edges. Our notion of weighted graph additionally specifies the set $T\subseteq X$ of \emph{target states}. 
We introduce some notations and terminologies.
\begin{itemize}
 \item $\Omega=\RRinfpos=\RRpos\cup\{\infty\}$---we augment $\RRpos$ with the maximum $\infty$---is the \emph{weight domain}. Its usual order $\le$ is henceforth denoted by $\sqsubseteq$. 
 \item A function $d\colon X\to \Omega$---assigning, to each state $x$, a weight $d(x)\in \Omega$---is an \emph{$\Omega$-valuation}.
 \item The set $[X,\Omega]$ of $\Omega$-valuations inherits the order $\sqsubseteq$ from $\Omega$ in the pointwise manner ($d_{1}\sqsubseteq d_{2}$ iff $d_{1}(x)\sqsubseteq d_{2}(x)$ for each $x\in X$). 
 \item We define the \emph{Bellman operator} $\Phi\colon [X,\Omega]\to [X,\Omega]$ by 
\begin{equation}\label{eq:BellmanOriginal}
\Phi(d) (x) = \begin{cases}
    0 & \text{if $x\in T$}, \\
    \displaystyle \min_{(x,x')\in E} \bigl(w(x,x') + d(x')\bigr) & \text{if } x\not\in T.
\end{cases}
\end{equation}
\end{itemize}
The following characterization is well-known. It is crucial for our theoretical development.
\begin{proposition}[\cite{Misra2001}]
    \label{prop:shortestPathAsGFPOriginal}\todo{\checkmark citation}
    The \emph{shortest path valuation} $d_{\mathrm{SP}}\colon X\to \Omega$, which assigns to each state $x\in X$ the shortest path length $d_{\mathrm{SP}}(x)$ to one of the target states in $T$, coincides with the greatest fixed point $\nu\Phi$ of the Bellman operator~(see \eqref{eq:BellmanOriginal}).
\end{proposition}

\begin{problem}[shortest path problem (SPP)]\label{def:shortestPathProblem}
The \emph{shortest path problem} (SPP) is the problem of computing $d_{\mathrm{SP}}\colon X\to \Omega$ for a given weighted graph  $(X,E,w,T)$.
\end{problem}\todo{\checkmark Yoav: I think when talking about problems it is better to write something like: "The shortest path problem is the problem of computing..."}

The order-theoretic essence of the Bellman--Ford algorithm is the computation of $\nu\Phi$ 
(Prop.~\ref{prop:shortestPathAsGFPOriginal})
via the iterative characterization in Lem.~\ref{lem:CC}. Specifically, we let $L=[X,\Omega]$; it is easily seen to have all infimums (including $\top=\bigsqcap \emptyset$) since $\Omega=\RRinfpos$ does. Therefore, by iteratively applying the Bellman operator $\Phi$ (which is $\Psi$ in Lem.~\ref{lem:CC}), we eventually obtain $\gfp \Phi$. (It is, of course, a different problem whether the sequence~\eqref{eq:CC} stabilizes after finitely many steps. This is the case with the Bellman operator $\Phi$ for the shortest paths in~\eqref{eq:BellmanOriginal}.)

\begin{algorithm}[t]
    \caption{The (classical) Dijkstra algorithm for the SPP}\label{alg:dijkstra-original}
    \footnotesize
    \begin{algorithmic}[1]
        \Procedure{Dijkstra}{$X$, $E$, $w$, $T$}
            \State $\letin{S}{T}$ \label{line:original-dijkstra-initialize-start}
            \State $\letin{Y}{S}$
            \State $\letin{d}{\lambda x \in X. \begin{cases} 0 & (x \in S) \\ \infty & (x \not \in S) \end{cases}} \quad \in [X, \Omega]$ \label{line:original-dijkstra-initialize-end}
            \While{$S \neq X$} \tikzmark{original-dijkstra-main-top}
                \State $\letin{P}{\text{the set of predecessors of $Y$}}$ \label{line:original-dijkstra-collect-predecessors}
                \State $\letin{d}{\lambda x. \begin{cases}
                    \min_{(x,x') \in E} (w(x,x') + d(x')) & (x \in P \cap (X \setminus S)) \\
                    d(x) & (\text{otherwise})
                \end{cases}}$
                \label{line:original-dijkstra-update-d}
                \Comment{Selective Bellman update}
                \State $\letin{Y}{\left\{ y \in X \setminus S \mid d(y) = \min_{z \in X \setminus S} d(z) \right\}}$
                \label{line:original-dijkstra-collect-minimal}
                \Comment{Minimal freezing}
                \State $\letin{S}{S \cup Y}$ \label{line:original-dijkstra-update-S}
                \Comment{Update the set of frozen states}
            \EndWhile \tikzmark{original-dijkstra-main-bottom}
            \State \Return $d$
        \EndProcedure
    \end{algorithmic}
\end{algorithm}

The Dijkstra algorithm (Algo.~\ref{alg:dijkstra-original}) accelerates this Bellman--Ford iteration; it does so by 1) maintaining the set $S$ of \emph{frozen} states, and 2) applying the Bellman update only selectively.

Specifically, 
the \emph{minimal freezing} step (Line~\ref{line:original-dijkstra-collect-minimal}) is crucial in Algo.~\ref{alg:dijkstra-original}. There one selects, among the states that have not yet been frozen, states $y$ with the minimal valuation and newly freezes them.
The core argument of the correctness proof is that the valuation of these states $y$ is already optimal: $d(y) = (\gfp\Phi)(y)$ in the algorithm.

We shall  develop a coalgebraic generalization of this argument, in \S\ref{sec:coalg-dijkstra-algorithm} onward.
The classical correctness proof of the Dijkstra algorithm can be obtained as a special case of the general correctness proof there (Thm.~\ref{thm:correctness-G}).

\begin{example}\label{example:shortest-path-run}
    Let $(\{ 0, 1, 2, 3, 4, 5 \}, E, w, \{ 0 \})$ be the weighted graph in Fig.~\ref{fig:example-weighted-graph}.
    The run of Algo.~\ref{alg:dijkstra-original} for the weighted graph is summarized in Table~\ref{table:run-shortest-path}.
    The returned valuation $d \colon \{ 0, 1, 2, 3, 4, 5 \} \to \RRinfpos$ represents the shortest path from each vertex $x \in \{ 0, 1, 2, 3, 4, 5 \}$ to the target $0 \in T \subseteq X$.
\end{example}
\begin{minipage}{\textwidth}
    \begin{minipage}{0.27\textwidth}
            \centering
            \scriptsize
            \begin{tikzpicture}[node distance=0.8cm, auto, baseline=-0.5cm]
                \node[state,accepting](0) {$0$};
                \node[state] (1) [below=of 0.center, anchor=center] {$1$};
                \node[state] (2) [left=of 0.center, anchor=center] {$2$};
                \node[state] (3) [below=of 2.center, anchor=center] {$3$};
                \node[state] (4) [left=of 2.center, anchor=center] {$4$};
                \node[state] (5) [left=of 3.center, anchor=center] {$5$};
                \path[->] (1) edge node[right] {$1$} (0)
                              edge[bend left=20] node {$1$} (3);
                \path[->] (2) edge[bend left=20] node {$1$} (4)
                              edge node {$6$} (0)
                              edge node {$2$} (3);
                \path[->] (3) edge[bend left=20] node {$2$} (1);
                \path[->] (4) edge[bend left=20] node {$1$} (2);
                \path[->] (5) edge node {$1$} (3)
                              edge node {$3$} (4);
            \end{tikzpicture}
            
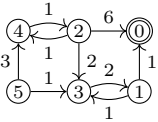
\captionof{figure}{An example of weighted graph}
            \label{fig:example-weighted-graph}
    \end{minipage}
    \begin{minipage}{0.7\textwidth}
            \centering
            \captionof{table}{The run of Algo.~\ref{alg:dijkstra-original} for the weighted graph in Fig.~\ref{fig:example-weighted-graph}.}
            \label{table:run-shortest-path}
            \scriptsize
            \begin{tabular}{ccccccccc}
                \toprule
                $n$ & $d(0)$ & $d(1)$ & $d(2)$ & $d(3)$ & $d(4)$ & $d(5)$ & $S$ & $Y$\\
                \midrule
                $0$ & $\infty$ & $\infty$ & $\infty$ & $\infty$ & $\infty$ & $\infty$ & & \\
                $1$ & $0$      & $\infty$ & $\infty$ & $\infty$ & $\infty$ & $\infty$ & $\{0\}$ & $\{0\}$ \\
                $2$ & $0$      & $1$      & $6$      & $\infty$ & $\infty$ & $\infty$ & $\{0, 1\}$ & $\{1\}$ \\
                $3$ & $0$      & $1$      & $6$      & $3$      & $\infty$ & $\infty$ & $\{0, 1, 3\}$ & $\{3\}$ \\
                $4$ & $0$      & $1$      & $5$      & $3$      & $\infty$ & $4$ & $\{0, 1, 3, 5\}$ & $\{5\}$ \\
                $5$ & $0$      & $1$      & $5$      & $3$      & $\infty$ & $4$ & $\{0, 1, 2, 3, 5\}$ & $\{2\}$ \\
                $6$ & $0$      & $1$      & $5$      & $3$      & $6$      & $4$ & $\{0, 1, 2, 3, 4, 5\}$ & $\{4\}$ \\
                \bottomrule
            \end{tabular}
    \end{minipage}
\end{minipage}

For accommodating this \emph{Dijkstra acceleration}, we present a technical adaptation of Lem.~\ref{lem:CC}. For convenience, we focus on sequences of length $\omega$.


\begin{lemma}[selective Bellman update]\label{lem:partial-Kleene-chain}
 Let $X$ be a set, $\Omega$ be a poset with the maximum $\top$, $\Phi \colon [X,\Omega] \to [X, \Omega]$ be a monotone map, and  $(P_1, P_2, \dots)$ be a sequence of subsets of $X$ (an element $x\in P_{n}$ is called an  \emph{active} state at the $n$-th iteration). 

    We define maps $d_n \colon X \to \Omega$, for each $n\in \mathbb{N}$,
    by induction:
    \[
    d_0 = \top_{[X,\Omega]}
    \quad \text{and} \quad
    d_{n + 1}(x) = \begin{cases}
        \Phi(d_n)(x) & \text{if } x \in P_{n + 1}, \\
        d_n(x) & \text{if }x\not\in P_{n + 1}.
    \end{cases}
    \]
    Note that the valuations are updated only in active states. 
    Then, we have $d_{n}\sqsupseteq d_{n + 1} \sqsupseteq\Phi( d_{n} )$ for each $n\in \mathbb{N}$. Moreover, for each $n$, we have
    $d_n \sqsupseteq \Phi^n(\top) $.
\end{lemma}
Obviously, convergence to $\gfp\Phi$---the main statement of Lem.~\ref{lem:CC}---is not guaranteed for an arbitrary choice of  $(P_1, P_2, \dots)$. The Dijkstra algorithm achieves this convergence by the careful design of active states (which is by the design of frozen states).


\section{The coalgebraic shortest path problem}\label{sec:minimum-weight-gfp}


In this section, we formalize our problem in the categorical language. It generalizes
 the SPP on weighted graphs (Prob.~\ref{def:shortestPathProblem}).
The generalization is mainly in three directions.
\begin{itemize}
 \item A \emph{(signature) functor} $F$ for coalgebras. Here coalgebras model state-based dynamics and generalize weighted graphs.  We restrict to functors of the shape $F = \WGG{G} = \BB \times \powsetfin(G\blank)$, where $G$ is a functor. Coalgebras for $\WGG{G}$ are referred to as \emph{weighted $G$-graphs}. 
 \item A \emph{weight domain} $\Omega$ that generalizes $\RRinfpos$ in Prob.~\ref{def:shortestPathProblem}. (We use the notion of \emph{pointed weight domain}; it additionally specifies so-called a \emph{final weight}.)
 \item A \emph{transition modality} $G\Omega\to \Omega$, formalized as a $G$-algebra over $\Omega$. This specifies how weights are accumulated along a $G$-path (i.e.\ repeated $G$-transitions).
\end{itemize}
The generalized problem is called the 
\emph{coalgebraic shortest path problem (CSPP)}. The original problem (Prob.~\ref{def:shortestPathProblem}) is its example, of course; other examples are given in \S\ref{subsec:CSPPExamples} (cf.\ Table~\ref{table:examples-problems}).

\subsection{Coalgebras as a generalization of weighted graphs}
Coalgebras are categorical models of state transition systems \cite{Rutten2000,Jacobs2016}.
\begin{definition}[coalgebra]\label{def:coalg}\todo{In preliminaries?}
    Let $F \colon \Sets \to \Sets$ be an endofunctor on $\Sets$.
    An \emph{$F$-coalgebra} is a morphism $\gamma \colon X \to FX$ in $\Sets$.
\end{definition}
In this paper, we restrict the signature functor $F$ to the following class.
\begin{definition}[weighted $G$-graph]\label{def:weightedGGraph}
     Let $G \colon \Sets \to \Sets$ be an endofunctor on $\Sets$. 
    The endofunctor $\WGG{G} \colon \Sets \to \Sets$ is defined by
    $\WGG{G} := \BB \times \powsetfin(G \blank)$
    where $\BB=\{ \btt, \bff \}$ and $\powsetfin$  is the finite power set functor.
    A $\WGG{G}$-coalgebra is called a \emph{weighted $G$-graph}. 
\end{definition}
In the weighted $G$-graph
\begin{math}
\gamma\colon X\to \BB \times \powsetfin(G X)
\end{math},
the first component of $\gamma(x)$ (the Boolean truth value, $\btt$ or $\bff$) tells if a state $x$ is a target or not. The second component of $\gamma(x)$ gives a finite branching (modeled by $\powsetfin$) over $G$-transitions. 
In the original SPP
(Prob.~\ref{def:shortestPathProblem}), the functor $G$ is given by $\RRpos \times \blank$ and assigns a weight from $\RRpos$ to each successor.

\subsection{Weight domains}
In the next definition, 
we pose  additional conditions  on the poset $(\Omega,\sqsubseteq)$. A \emph{final weight} is the weight given to an immediately terminating path; it generalizes the weight $0\in\RRinfpos$ for the case $x\in T$ in~\eqref{eq:BellmanOriginal}. 

\begin{definition}[(pointed) weight domain]\label{def:pointedWeightDom}
    A \emph{weight domain} is a poset $\Omega = (\Omega, \sqsubseteq)$ such that
    1) $\sqsubseteq$ is a total order,
    2) it has the least element $\bot_{\Omega}$, and 
    3) it has  the infimum of any subset of $\Omega$ (thus also the greatest element $\top_{\Omega}$).
    A \emph{pointed weight domain} is a triple $(\Omega, \sqsubseteq, \finalweight)$ where
    $(\Omega, \sqsubseteq)$ is a weight domain and
    $\finalweight \in \Omega$ is a designated \emph{final weight}.
\end{definition}


Recall that $[X,\Omega]$ denotes the function space from $X$ to $\Omega$, with its elements called $\Omega$-valuations.
The set $[X,\Omega]$ is equipped with the following pointwise order.
\begin{definition}
    Let $\Omega$ be a weight domain, and  $X$ be a set.
    For maps $d, e \in [X,\Omega]$, we write $d \sqsubseteq e$ if $d(x) \sqsubseteq e(x)$ for every $x \in X$.

    The pair $([X, \Omega], \sqsubseteq)$ is a partially ordered set with the greatest element $\top_{[X, \Omega]} = \lambda x. \top_{\Omega}$ and the least element $\bot_{[X, \Omega]} = \lambda x. \bot_{\Omega}$.
\end{definition}

\subsection{Transition modalities and Bellman operators}
In a weighted $G$-graph 
\begin{math}
    \gamma\colon X\to \BB \times \powsetfin(G X)
\end{math}, at each state $x$, we pick one $G$-transition among those finitely many which are offered.\footnote{Specifically, by $(\pi_{1} \co \gamma)(x)\in \powsetfin(G X)$---recall our indexing convention from \S\ref{subsec:notations}.} Repeating this exhibits a succession of $G$-transitions. Weights get accumulated in its course. 

A transition modality specifies how this accumulation is defined. It is formalized as a categorical algebra, together with a suitable ``continuity'' condition that we need for correctness  (Thm.~\ref{thm:correctness-G}, and 
Lem.~\ref{lem:fold-one-step} in particular).


\todo{2026-03-07 Takahiro, please revise the part after this point. 2026-03-10 Revised}


\begin{definition}[transition modality]
    \label{def:transtionModality}
    Let  $G \colon \Sets \to \Sets$ be a functor.
    A \emph{transition modality} $\sigma$ for $G$ is a $G$-algebra $\sigma \colon G\Omega \to \Omega$ such that, for any
    finite set $X$, the map
    \begin{equation}\label{eq:sigmaCircGBla}
        \sigma \co G(\blank) \colon\quad [X, \Omega] \;\longrightarrow\; [GX, \Omega],\qquad
        \bigl(\,X\xrightarrow{d}\Omega\,\bigr)\;\longmapsto\; \bigl(\,GX\xrightarrow{Gd}G\Omega\xrightarrow{\sigma} \Omega\,\bigr)
    \end{equation}
    preserves the infimum of any subset $D$ of $[X, \Omega]$, that is, 
    $ \sigma \circ G\left(\bigsqcap_{d \in D} d\right)
    = \bigsqcap_{d \in D} (\sigma \circ Gd)$.
\end{definition}


Optimization problems on coalgebras are specified by Bellman operators, which is induced by a transition modality $\sigma \colon G \Omega \to \Omega$.

\todoil{The organization of the rest of the paper is too much like the usual mathematical writing. Let's change it!}


\begin{definition}[Bellman operator]
Let $G \colon \Sets \to \Sets$ be a functor, $\Omega = (\Omega, \sqsubseteq, \xi)$ be a pointed weight domain, $\sigma \colon G \Omega \to \Omega$ be a transition modality and $\gamma \colon X \to \BB \times \powsetfin(GX)$ be a weighted $G$-graph.
    The \emph{Bellman operator} $\Bellman{G}{\sigma}{\gamma} \colon [X, \Omega] \to [X, \Omega]$ is defined as follows:
    \[
    \Bellman{G}{\sigma}{\gamma}(d)(x) = \begin{cases}
        \finalweight \sqcap \bigsqcap_{a \in A} \sigma\bigl((Gd)(a)\bigr)
        & \text{(if $\gamma(x) = (\btt, A)$)} \\
        \bigsqcap_{a \in A} \sigma\bigl((Gd)(a)\bigr)
        & \text{(if $\gamma(x) = (\bff, A)$).}
    \end{cases}
    \]
\end{definition}
\todoil{\checkmark Explain connection Bellman operator and Bellman--Ford algorithm}
The definition of the Bellman operator generalizes that for the SPP given in~\eqref{eq:BellmanOriginal}.
By iteratively applying the Bellman operator $\Bellman{G}{\sigma}{\gamma}$ to $\top \in [X, \Omega]$, we obtain the greatest fixed point $\gfp \Bellman{G}{\sigma}{\gamma}$.
This procedure can be seen as a coalgebraic generalization of the original Bellman--Ford algorithm for the SPP.



\subsection{The coalgebraic shortest path problem}


The problem we tackle is stated as follows.
In \versionswitch{\cite[\S{}C.1]{CoalgDijkstraArXiv}}{\S\ref{sec:VxX}}, we show that by taking $G = \RRpos \times \blank$ and $\sigma \colon G \RRinfpos \to \RRinfpos$ defined by $\sigma(a, b) = a + b$, we obtain the original SPP.
Note that, some examples obtained by instantiating $G$ and $\sigma$ are not literally about shortest paths; however, for the sake of terminology, we  call the general problem a \emph{coalgebraic shortest path problem}.
\todo{\checkmark Name: coalgebraic shortest path problem (with a remark). Make it self-contained. A complicated definition can be somewhere else, but we need a pointer to it}
\begin{problem}[coalgebraic shortest path problem (CSPP)]\label{prob:CSPP}
    Let $G \colon \Sets \to \Sets$ be a functor, and $\sigma \colon G\Omega \to \Omega$ be a transition modality on a pointed weight domain $\Omega = (\Omega, \sqsubseteq, \finalweight)$.
    A \emph{coalgebraic shortest path problem} (\emph{CSPP}) on $(G, \sigma)$ is the problem of computing the greatest fixed point $\gfp \Bellman{G}{\sigma}{\gamma}$ of the Bellman operator $\Bellman{G}{\sigma}{\gamma} \colon [X, \Omega] \to [X, \Omega]$ for a given weighted $G$-graph $\gamma \colon X \to \BB \times \powsetfin(G X)$.
\end{problem}

\subsection{Instances of the CSPP}\label{subsec:CSPPExamples}
\todo[inline]{2026-05-04 Takahiro, please triple-check the writing here}
 Many  graph-based optimization problems are instances of the CSPP.
See Table~\ref{table:examples-problems}. What is notable there is that subtle differences in problems can change applicability of the Dijkstra acceleration (see the last column). Such phenomena are well-known for shortest path with negative edges or longest path; we identify more instances, such as interest vs.\ discount rates.

\begin{notations}
We use $\sqcap$ and $\sqcup$ for the meet and join with respect to the order of $\Omega$, respectively, and use $\min$ and $\max$ for the minimum and maximum with respect to the usual order of $\NN$ or $\RR$, respectively.
\end{notations}

\begin{remark}
    We formulated the CSPP using the greatest fixed point.
   A dual view is possible, reversing the order and computing the 
 least fixed point.
    Our choice of the greatest fixed point is made to naturally match the usual fixed-point formulation of the SPP \cite{Misra2001}.
\end{remark}
\begin{table}[]
\caption{
 Instances of the CSPP on $(G, \sigma)$.
    The ``Dijkstra'' column indicates if the coalgebraic Dijkstra algorithm is correct there. Known Dijkstra-like algorithms are noted with references.
}
\label{table:examples-problems}
\rowcolors{1}{}{gray!10} 
\renewcommand{\arraystretch}{1.2} 
\small 
\setlength{\tabcolsep}{3pt} 

\centering 
\resizebox{0.95\textwidth}{!}{%
\begin{tabular}{@{}llllc@{}}
    \toprule
    Problem & $(\Omega, \sqsubseteq, \finalweight)$ & $G$ & $\sigma \colon G\Omega \to \Omega$ & Dijkstra applies? \\
    \midrule
    Reachability & $(\{0,\infty\}, \le, 0)$ & $\idfunc$ & $\sigma(a) = a$ & Yes \\
    Unweighted shortest path & $(\NNinf, \le, 0)$ & $\idfunc$ & $\sigma(a) = 1 + a$ & Yes \\
    Unweighted longest path
    & $(\NNinf, \ge, 0)$ & $\idfunc$ & $\sigma(a) = 1 + a$ & No \\
    SPP \cite{Dijkstra1959} & $(\RRinfpos, \le, 0)$ & $\RRpos \times \blank$ & $\sigma(a, b) = a + b$ & Yes \cite{Dijkstra1959} \\
    SPP with negative edges & $(\RR^{\pm\infty}, \le, 0)$ & $\RR \times \blank$ & $\sigma(a, b) = a + b$ & No \\
    SPP with interest rate & $(\RRinfpos, \le, 0)$ & $\RRpos \times [1, \infty) \times \blank$ & $\sigma(a, a', b) = a + a'b$ & Yes \\ 
    SPP with discount rate & $(\RRinfpos, \le, 0)$ & $\RRpos \times [0, 1] \times \blank$ & $\sigma(a, a', b) = a + a'b$ & No \\
    Widest path \cite{Pollack1960} & $(\RRinfpos, \ge, \infty)$ & $\RRpos \times \blank$ & $\sigma(a, b) = a \sqcup b = \min(a, b)$ & Yes \cite{Pollack1960} \\
    Most reliable path \cite{Wing1961} & $([0, 1], \ge, 1)$ & $[0,1] \times \blank$ & $\sigma(a, b) = ab$ & Yes \cite{Wing1961}\\
    Shortest binary tree & $(\RRinfpos, \le, 0)$ & $\RRpos \times (\blank)^2$ & $\sigma(a, b_0, b_1) = a + b_0 + b_1$ & Yes \\
    Binary reachability game & $(\{0, \infty\}, \le, 0)$ & $(\blank)^2$ & $\sigma(a_0, a_1) = a_0 \sqcup a_1 = \max(a_0, a_1)$ & Yes \\
    Reachability game \cite{Mazala2002} & $(\{0, \infty\}, \le, 0)$ & $\powsetfinne$ & $\sigma(A) = \bigsqcup A = \max A$ & Yes \cite{Mazala2002} \\
    Dynamic game \cite{BardiLopez2016} & $(\RRinfpos, \le, 0)$ & $\powsetfinne(\RRpos \times \blank)$ & $\sigma(A) = \max_{(a,b) \in A} (a + b)$ & Yes \cite{BardiLopez2016} \\
\begin{minipage}{3.8cm}
     Dynamic game with\newline discount rate \cite{BardiLopez2016} 
\end{minipage}
& $(\RRinfpos, \le, 0)$ & $\powsetfinne([\ell_0, L] \times \blank)$ & $\sigma_r(A) = \max_{(a,b) \in A} (a + rb)$ &
    \begin{minipage}{3.2cm}
        No (unless $\ell_0 = L$ or $r = 1$, see Ex.~\ref{example:incorrectness-of-bardi-lopez-condition}) \cite{BardiLopez2016}
    \end{minipage} \\
    \begin{minipage}{3.8cm}
        Max probabilistic\newline reachability \cite[\S 10.6.1]{BaierKatoen2008}
    \end{minipage}
    & $([0, 1], \ge, 1)$ & $\distr$ & $\sigma(\mu) = \sum_{a\in \Omega} \mu(a)a$ & No \\
    \bottomrule
\end{tabular}
}
\end{table}

\begin{figure}[t]
\centering
\begin{minipage}[b]{0.38\linewidth}
    \centering
    \scriptsize
    \begin{tikzpicture}[node distance=0.8cm, auto, baseline=-0.5cm]
        \node[state,accepting](0) {$0$};
        \node[state] (1) [below=of 0.center, anchor=center] {$1$};
        \node[state] (2) [left=of 0.center, anchor=center] {$2$};
        \node[state] (3) [below=of 2.center, anchor=center] {$3$};
        \node[state] (4) [left=of 2.center, anchor=center] {$4$};
        \path[->] (1) edge (0)
                      edge[bend left=20] (3);
        \path[->] (2) edge[bend left=20] (4);
        \path[->] (3) edge[bend left=20] (1);
        \path[->] (4) edge[bend left=20] (2);
    \end{tikzpicture}
    \hspace{2mm}
    $\begin{aligned}
        \gamma(0) & = (\btt, \emptyset) \\
        \gamma(1) & = (\bff, \{ 0, 3 \}) \\
        \gamma(2) & = (\bff, \{ 4 \}) \\
        \gamma(3) & = (\bff, \{ 1 \}) \\
        \gamma(4) & = (\bff, \{ 2 \})
    \end{aligned}$
\end{minipage}
\hspace{1em}
\begin{minipage}[b]{0.54\linewidth}
    \centering
    \scriptsize
    \begin{tikzpicture}[node distance=0.8cm, auto, baseline=-0.5cm]
        \node[state,accepting](0) {$0$};
        \node[state] (1) [below=of 0.center, anchor=center] {$1$};
        \node[state] (2) [left=of 0.center, anchor=center] {$2$};
        \node[state] (3) [below=of 2.center, anchor=center] {$3$};
        \node[state] (4) [left=of 2.center, anchor=center] {$4$};
        \node[state] (5) [left=of 3.center, anchor=center] {$5$};
        \path[->] (1) edge node[right] {$1$} (0)
                      edge[bend left=20] node {$1$} (3);
        \path[->] (2) edge[bend left=20] node {$1$} (4)
                      edge node {$6$} (0)
                      edge node {$2$} (3);
        \path[->] (3) edge[bend left=20] node {$2$} (1);
        \path[->] (4) edge[bend left=20] node {$1$} (2);
        \path[->] (5) edge node {$1$} (3)
                      edge node {$3$} (4);
    \end{tikzpicture}
    \hspace{2mm}
    $\begin{aligned}
        \gamma(0) & = (\btt, \emptyset) \\
        \gamma(1) & = (\bff, \{ (1, 0), (1, 3) \}) \\
        \gamma(2) & = (\bff, \{ (6, 0), (2, 3), (1, 4) \}) \\
        \gamma(3) & = (\bff, \{ (2, 1) \}) \\
        \gamma(4) & = (\bff, \{ (1, 2) \}) \\
        \gamma(5) & = (\bff, \{ (1, 3), (3, 4) \})
    \end{aligned}$
    
\end{minipage}

\caption{Examples of weighted $G$-graphs $\gamma \colon X \to \BB \times \powsetfin(X)$ and $\gamma \colon X \to \BB \times \powsetfin(\RRpos \times X)$.}
\label{fig:coalg-shortest-path}
\end{figure}

\todo{\checkmark Please care consistency with earlier sections, such as destination vs.\ target}

Here are instances of the CSPP (cf.\ Table~\ref{table:examples-problems}). Some more examples are shown in \versionswitch{\cite[\S{}B]{CoalgDijkstraArXiv}}{\S\ref{sec:moreInstancesOfTheCSPP}}.

\todoil{Let's expand this. (If too long, it can be moved to an appendix, so no worries.) First define the Bellman operator (more precisely G, $\sigma$), and say that its GFP is blahblah. Say that it induces topological sorting. In Table~1, say ``Topological sort (in sense of Ex.~\ref{example:topological-sorting})''}

\begin{example}[SPP]\label{examlpe:shortest-path}
    Let $\Omega = (\RRinfpos, \le, 0)$,
    $G \colon \Sets \to \Sets$ be a functor defined as $GX = \RRpos \times X$ and
    $\sigma \colon G \Omega \to \Omega$ be a transition modality defined as $\sigma(a,b) = a + b$.
    For a finite set $X$,
    a weighted $G$-graph $\gamma \colon X \to \BB \times \powsetfin(\RRpos \times X)$ is a directed graph such that
    (1) $X$ is a set of vertices,
    (2) a set of targets $Z = \{ x \in X \mid \pi_0(\gamma(x)) = \btt \}$ is specified, and
    (3) for $x, y\in X$ and $a \in \RRpos$, there is an edge from $x$ to $y$ with length $a$ if $(a, y) \in \pi_1(\gamma(x))$.
    An example of a coalgebra $\gamma \colon X \to \BB \times \powsetfin(\RRpos \times X)$ is shown on the right-hand side of Fig.~\ref{fig:coalg-shortest-path}, which is the coalgebraic representation of the weighted graph in Fig.~\ref{fig:example-weighted-graph}.
    Given an $\Omega$-valuation $d \in [X,\Omega]$, the Bellman update $\Bellman{G}{\sigma}{\gamma}(d)$ is given by
    \[ \Bellman{G}{\sigma}{\gamma}(d)(x) = \begin{cases}
        0 & (\text{if }\gamma(x) = (\btt, A)) \\
        \min \{ a + d(y) \mid (a, y) \in A \} & (\text{if }\gamma(x) = (\bff, A)).
    \end{cases} \]
    We can show that $\gfp \Bellman{G}{\sigma}{\gamma}(x)$ is the value of the shortest path from $x$ to one of the targets.
\end{example}

\begin{example}[unweighted longest path]\label{example:topological-sorting}
    Let $\Omega = (\NNinf, \ge, 0)$, $G = \idfunc$ and $\sigma \colon \Omega \to \Omega$ be a transition modality defined as $\sigma(a) = 1 + a$.
    For a weighted $G$-graph $\gamma \colon X \to \BB \times \powsetfin(X)$, the Bellman operator $\Bellman{G}{\sigma}{\gamma} \colon [X, \Omega] \to [X, \Omega]$ is given by
    \[
    \Bellman{G}{\sigma}{\gamma}(d)(x) = \begin{cases}
        \max(0, \max_{x \in A} (1 + d(x)) ) & (\text{if }\gamma(x) = (\btt, A)) \\
        \max_{x \in A} (1 + d(x)) & (\text{if }\gamma(x) = (\bff, A)).
    \end{cases}
    \]
    If $\gamma$ is a directed acyclic graph (DAG), the greatest fixed point $\gfp\Bellman{G}{\sigma}{\gamma}(x)$ is the length of the longest path from $x \in X$ to one of targets.
    If $\gamma$ is not a DAG, $\gfp\Bellman{G}{\sigma}{\gamma}(x) = \infty$ holds for every $x$ reachable to a vertex in a cycle.
\end{example}

\todoil{\checkmark Write in such a way that it's compared with Ex.~\ref{examlpe:shortest-path}}
\begin{example}[SPP with negative edges]\label{example:shortest-path-negative}
    We modify the SPP in Ex.~\ref{examlpe:shortest-path} by allowing negative edge weights.
    Let $\Omega = (\RRpminf, \le, 0)$ and
    $G \colon \Sets \to \Sets$ be a functor defined as $GX = \RR \times X$.
    Let  $\sigma \colon G \Omega \to \Omega$ be a transition modality defined as $\sigma(a,b) = a + b$.
    For a finite set $X$, a weighted $G$-graph $\gamma \colon X \to \BB \times \powsetfin(\RR \times X)$ is a directed graph with targets $Z \subseteq X$. Each edge has a length, which may be negative.
    For $d \in [X, \Omega]$, the Bellman operator is
    \[
    \Bellman{G}{\sigma}{\gamma}(d)(x) = \begin{cases}
        \min(0, \min \{ \sigma(a, d(y)) \mid (a, y) \in A \})  & (\text{if }\gamma(x) = (\btt, A)) \\
        \min \{ \sigma(a, d(y)) \mid (a, y) \in A \} & (\text{if }\gamma(x) = (\bff, A)).
    \end{cases}
    \]
    For $x \in X$, the value $\gfp \Bellman{G}{\sigma}{\gamma}(x)$ of the greatest fixed point is the shortest path from $x$ to one of the targets.
    For example, if $x \in X$ is reachable to a vertex in $Z$ and is contained in a negative cycle, then $\gfp\Phi(x) = -\infty$ holds.
\end{example}

\begin{example}[SPP with interest/discount rate]\label{example:shortest-path-interest-discount}
    We modify the SPP in Ex.~\ref{examlpe:shortest-path} by introducing interest/discount rates.
    Let $\Omega = (\RRinfpos, \le, 0)$, as in Ex.~\ref{examlpe:shortest-path}.
    For a functor $GX = \RRpos \times [1, \infty) \times X$, a transition modality $\sigma(a, a', b) = a + a'b$ and a weighted $G$-graph $\gamma \colon X \to \BB \times \powsetfin(\RRinfpos \times [1, \infty) \times X)$, the greatest fixed point operator of the Bellman operator is the shortest path with interest rate.
    The ``length'' of a path $x_n \xrightarrow{(a_n, a'_n)} x_{n-1} \xrightarrow{(a_{n-1}, a'_{n-1})}\cdots \xrightarrow{(a_1, a'_1)} x_0$ is defined by $a_n + a'_n(a_{n-1} + a'_{n-1}( \cdots (a_1 + a'_1 0) \dots ))$, where $a'_i$ are interest rates.
    For a functor $GX = \RRpos \times [0,1] \times X$, a transition modality $\sigma(a, a', b) = a + a'b$ and a weighted $G$-graph $\gamma \colon X \to \BB \times \powsetfin(\RRpos \times [0, 1] \times X)$, the greatest fixed point operator of the Bellman operator is the shortest path with discount rate.
    In this case, the definition of the ``length'' of a path is similar to the shortest path with interest rate, but all rates $a'_i$ satisfy $0 \le a'_i \le 1$.
\end{example}

\begin{example}[the shortest binary tree problem]\label{example:shortest-bin-tree}
    We modify the SPP (Ex.~\ref{examlpe:shortest-path}) so that each edge has two successors. In this case, the problem becomes that of finding the shortest tree.
    Let $\Omega = (\RRinfpos, \le, 0)$ as in Ex.~\ref{examlpe:shortest-path} and
    $G \colon \Sets \to \Sets$ be a functor defined as $GX = \RRpos \times X^2$. 
    Let $\sigma \colon G\Omega \to \Omega$ is a transition modality defined as $\sigma(a, b_0, b_1) = a + b_0 + b_1$.
    For a finite set $X$, a weighted $G$-graph $\gamma \colon X \to \BB \times \powsetfin(\RRpos \times X^2)$ is defined as follows:
    (1) A set of targets $Z = \{ x \in X \mid \pi_0(\gamma(x)) = \btt \}$ is specified.
    (2) For each vertex $x \in X$, there is a finite set $\{ (a^j, y^j_0, y^j_1) \}_{j \in J_x} \subseteq \RRpos \times X \times X$.
    An example of a weighted $G$-graph $\gamma \colon X \to \BB \times \powsetfin(\RRpos \times X^2)$ is shown in Fig.~\ref{fig:coalg-shortest-binary-tree}.
    \begin{figure}[t]
        \centering
        \scriptsize
        \[
        \begin{tikzpicture}[node distance=1.1cm, auto, baseline=-0.6cm]
            \node[state,accepting](0) {$0$};
            \node[state] (1) [below=of 0.center, anchor=center] {$1$};
            \node[dot, above=0.5cm of 1.center, anchor=center] (1a) {};
            \node[state] (2) [left=of 0.center, anchor=center] {$2$};
            \node[dot, right=0.5cm of 2.center, anchor=center] (2a) {};
            \node[state] (3) [below=of 2.center, anchor=center] {$3$};
            \node[dot, right=0.5cm of 3.center, anchor=center] (3a) {};
            \node[dot, above=0.5cm of 3.center, anchor=center] (3b) {};
            \node[state] (4) [left=of 2.center, anchor=center] {$4$};
            \node[dot, right=0.5cm of 4.center, anchor=center] (4a) {};
            \path (1) edge node {$1$} (1a);
            \path[->] (1a) edge[bend right=20] (0)
                           edge[bend left=20] (2);
            \path (2) edge node {$3$} (2a);
            \path[->] (2a) edge[bend left=20] (0)
                           edge[bend right=20] (0);
            \path (3) edge node[below] {$2$} (3a);
            \path[->] (3a) edge[bend right=20] (1)
                           edge[bend left=20] (2);
            \path (3) edge node[left] {$2$} (3b);
            \path[->] (3b) edge[bend left=30] (2)
                           edge[bend left=20] (4);
            \path (4) edge node {$1$} (4a);
            \path[->] (4a) edge[bend left=20] (2)
                           edge[bend left=40] (4);
        \end{tikzpicture}
        \qquad
        \begin{aligned}
            \gamma(0) & = (\btt, \emptyset) \\
            \gamma(1) & = (\bff, \{ (1, 2, 0) \}) \\
            \gamma(2) & = (\bff, \{ (3, 0, 0) \}) \\
            \gamma(3) & = (\bff, \{ (2, 4, 2), (2, 2, 1) \}) \\
            \gamma(4) & = (\bff, \{ (1, 2, 4) \})
        \end{aligned}
        \qquad
        \begin{aligned}
            \gfp \Phi(0) & = 0 \\
            \gfp \Phi(1) & = 4 \\
            \gfp \Phi(2) & = 3 \\
            \gfp \Phi(3) & = 9 \\
            \gfp \Phi(4) & = \infty
        \end{aligned}
        \]
        \caption{An example of a weighted $G$-graph $\gamma \colon X \to \BB \times \powsetfin(\RRpos \times X^2)$ for $X = \{ 0, 1, 2, 3, 4 \}$.}
        \label{fig:coalg-shortest-binary-tree}
    \end{figure}

    For an $\Omega$-valuation $d \in [X, \Omega]$, the updated valuation $\Bellman{G}{\sigma}{\gamma}(d)$ is
    \[
    \Bellman{G}{\sigma}{\gamma}(d)(x) = \begin{cases}
        0 & (\text{if }\gamma(x) = (\btt, A)) \\
        \min_{(a, y_0, y_1) \in A} (a + d(y_0) + d(y_1)) & (\text{if }\gamma(x) = (\bff, A)).
    \end{cases}
    \]
    We define the notion of \emph{run trees} of $\gamma \colon X \to \BB \times \powset(\RRpos \times X^2)$:
   it is a rooted binary tree $T$ of finite height such that
   (1) each node is labeled by a pair $(x, a)$ of $x \in X$ and $a \in \RRpos$,
   (2) for each leaf, its label $(x, a)$ satisfies $x \in Z$ and $a = 0$, and
   (3) for each node labeled by $(x, a)$ with their children labeled by $(y_i, b_i)$ ($i = 0, 1$) satisfies $(a, y_0, y_1) \in \pi_1(\gamma(x))$.
    Given a run tree $T$, we define its length $\sigma(T)$ recursively as follows:
    (1) If $T$ consists of a single node, then $\sigma(T) = 0$.
    (2) If the root of $T$ is labeled by $(x, a)$ and the subtrees of the root are $T_0$ and $T_1$, then $\sigma(T) = a + \sigma(T_0) + \sigma(T_1)$.

    The value $\gfp \Bellman{G}{\sigma}{\gamma}(x)$, the greatest fixed point, is the length of the shortest run tree $T$ whose root is labeled by $(x,a)$ for some $a\in\Omega$.
\end{example}

\begin{example}[the reachability game, the dynamic game]\label{example:reachability-game-dynamic-game}  \todo{2026-05-04 Do not refer to the appendix.}


    Let $\Omega = (\{ 0, \infty \}, \le, 0)$ and
    $G \colon \Sets \to \Sets$ be the non-empty finite power set functor $\powsetfinne$.
    Let 
    $\sigma \colon G\Omega \to \Omega$ be a transition modality defined as
    $\sigma(A) = \bigsqcup A$.
    For a finite set $X$ of vertices, a weighted $G$-graph $\gamma \colon X \to \BB \times \powsetfin(\powsetfinne(X))$ is given by the following structure:
    (1) A set $Z = \{ x \in X \mid \pi_0(\gamma(x)) = \btt \}$ of targets is specified.
    (2) For each vertex $x \in X$, there is a finite family $\{ (y^j_i)_{i \in I_j} \}_{j \in J_x}$ of sets of vertices.

    We define a \emph{reachability game} on $\gamma$ as follows: The game is played between two players---Alice and Bob.
    Starting with a pebble placed on a vertex $x \in X$, the game proceeds by iterating the following procedure.
    \begin{enumerate}
        \item If the vertex where the pebble is placed is in $Z$, Alice wins.
        \item Otherwise, Alice chooses a label $j \in J_x$.
        \item Bob chooses a label $i \in I_j$ and moves the pebble from $x$ to $y^j_i$.
    \end{enumerate}
    Starting with a pebble on a vertex $x \in X$, we ask whether Alice can always win in finitely many steps, regardless of Bob’s moves.
    The answer is $\gfp\Bellman{G}{\sigma}{\gamma}(x)$:
    \[
    \gfp \Bellman{G}{\sigma}{\gamma}(x) =
    \begin{cases}
        0 & \text{if there is a winning strategy of Alice,} \\
        \infty & \text{otherwise.}
    \end{cases}
    \]

    By replacing $\Omega$ with $(\RRinfpos, \le, 0)$ and modifying $G$ and $\sigma$ accordingly (see Table~\ref{table:examples-problems}), we obtain a game in which Alice aims to maximize her reward.
    The maximal reward achievable by Alice under the optimal strategy in this game coincides with the greatest fixed point.

    A special case of \emph{dynamic games} studied by Bardi and L{\'{o}}pez \cite{BardiLopez2016} can also be modeled in our framework.
    Let $\Omega = (\RRinfpos, \le, \finalweight)$ be a weight domain, $\ell_0$ and $L$ be positive real numbers with $\ell_0 \le L$, $G = \powsetfinne([\ell_0, L] \times X)$ and $\sigma \colon G\Omega \to \Omega$ be a transition modality defined as $\sigma_r(A) = \max_{(a,b) \in A} (a + rb)$ for some fixed discount rate $r \in (0,1]$.
    For a coalgebra $\gamma \colon X \to \BB \times \powsetfinne([\ell_0, L] \times X)$, the value $\gfp \Bellman{G}{\sigma}{\gamma}(x)$ is the maximum reward that Alice can obtain starting from $x$.
    The \CSPP{} on $(G, \sigma_r)$ is an alternating case of 
    the dynamic game \cite{BardiLopez2016}.

Whether the Dijkstra acceleration works in this setting is a subtle issue. It is discussed  later in Ex.~\ref{example:incorrectness-of-bardi-lopez-condition}.


\end{example}

%

\begin{example}[the maximum probabilistic reachability]\label{example:max-prob-reachability}
    Let $\Omega = ([ 0, 1 ], \ge, 1)$ and
    $G \colon \Sets \to \Sets$ be the distribution functor $G = \distr$.
    Let $\sigma \colon \distr\Omega \to \Omega$ be a transition modality defined as
    $\sigma(\mu) = \sum_{a \in \Omega} \mu(a) a$.
    For a finite set $X$, a weighted $G$-graph $\gamma \colon X \to \BB \times \powsetfin(\distr X)$ is the following structure:
    (1) A set $Z = \{ x \in X \mid \pi_0(\gamma(x)) = \btt \}$ is specified.
    (2) For each vertex $x \in X$, there is a finite set $\{ \mu_j \}_{j \in J_x}$ of probabilistic distributions on $X$.

    For an $\Omega$-valuation $d \in [X, \Omega]$, the updated valuation by the Bellman operator $\Bellman{G}{\sigma}{\gamma}(d)$ is
    \[
    \Bellman{G}{\sigma}{\gamma}(d)(x) = \begin{cases}
        1 & (\text{if }\gamma(x) = (\btt, A)) \\
        \max_{\mu \in A} \sum_{y \in X} \mu(y) d(y) & (\text{if }\gamma(x) = (\bff, A)).
    \end{cases}
    \]
    We define a one-player pebble game on $\gamma$.
    Starting with a pebble placed on a vertex $x \in X$, the game proceeds by iterating the following procedure:
    \begin{enumerate}
        \item If the vertex where the pebble is placed is in $Z$, the player wins.
        \item Otherwise, the player chooses a label $j \in J_x$.
        \item The pebble moves from $x$ to $y \in X$ with probability $\mu_j(y)$.
    \end{enumerate}
    The value $\gfp \Bellman{G}{\sigma}{\gamma}(x)$ is the maximum winning probability of the player.
\end{example}

\section{The coalgebraic Dijkstra algorithm for the CSPP}
\label{sec:coalg-dijkstra-algorithm}
We now introduce our coalgebraic generalization of the Dijkstra algorithm (Algo.~\ref{alg:dijkstra-original}). Its correctness relies delicately on precise problem settings (Table~\ref{table:examples-problems}); we present general axiomatics that
capture the difference. We introduce the notion of \emph{finitely supported functor} here.

As discussed in~\S\ref{sec:intro}, here we present the most general axiomatics where $G$ is an arbitrary finitely supported functor. In \versionswitch{\cite[\S{}C]{CoalgDijkstraArXiv}}{\ref{sec:conditions-specific-functors}}, for illustration, we provide a restricted theoretical description, limiting $G$ to special classes.
\todo{reference}


In our algorithm, we use the following coalgebraic notion of \emph{predecessors}.
\begin{definition}[predecessor and successor {\cite[Def.~3.4]{JacobsWissmann2023}}]
    Let $F \colon \Sets \to \Sets$ be a functor, 
  $\gamma \colon X \to FX$ be a coalgebra, and $Y\subseteq X$. 
\todo{\checkmark Let's use a uniform format: Let F be a functor...}
        An element $x \in X$ is a \emph{predecessor} of $Y$ if
        \[
        x \;\in\; X \setminus \{ z \in X \mid \gamma(z) \in F(X \setminus Y) \},
        \quad\text{that is,}\quad
   x\;\in\; \{ x \in X \mid \gamma(x) \not\in F(X \setminus Y)\}.
        \]
        We define $\predecessor(Y) := \{ x \in X \mid \text{$x$ is a predecessor of $Y$}\}$.
        For $x, y \in X$, $y$ is a \emph{successor} of $x$ if $x$ is a predecessor of $\{y\}$.
        We define $\successor(x) := \{ y \in X \mid \text{$y$ is a successor of $x$}\}$.
    \end{definition}

\begin{definition}[$\textsc{CoalgDijkstra}_{(G,\sigma)}$]
Let $\Omega = (\Omega, \sqsubseteq, \finalweight)$ be a pointed weight domain, $G \colon \Sets \to \Sets$ a functor, and $\sigma \colon G\Omega \to \Omega$ a transition modality. The  \emph{coalgebraic Dijkstra algorithm} $\textsc{CoalgDijkstra}_{(G,\sigma)}$ for the CSPP on $(G,\sigma)$ (cf.\ Prob.~\ref{prob:CSPP}) is  Algo.~\ref{alg:coalgebraic-dijkstra}.
\end{definition}


It generalizes the (classical) Dijkstra algorithm for the SPP (Algo.~\ref{alg:dijkstra-original}). 
It is 
 almost the same as Algo.~\ref{alg:dijkstra-original}, sharing
 the two key features 
(frozen states $S$ and  selective Bellman update). The difference is almost solely in Line~7, where Algo.~\ref{alg:coalgebraic-dijkstra} uses a generic Bellman operator $\Bellman{G}{\sigma}{\gamma}$---parametrized in a transition modality $\sigma$---as opposed to a specific one in Algo.~\ref{alg:dijkstra-original}.



\begin{algorithm}[t]
    \caption{The coalgebraic Dijkstra algorithm for the CSPP on $(G, \sigma \colon G\Omega \to \Omega)$.}\label{alg:coalgebraic-dijkstra}
   \small
    \begin{algorithmic}[1]
        \Procedure{CoalgDijkstra${}_{(G, \sigma)}$}{$\gamma \colon X \to \BB \times \powsetfin(GX)$} \tikzmark{right}
            \State $\letin{S}{\{ x \in X \mid \pi_0(\gamma(x)) = \btt \}}$ \tikzmark{init-top} \label{line:init-begin}
            \State $\letin{Y}{S}$
            \State $\letin{d}{\lambda x \in X. \begin{cases} \finalweight & (x \in S) \\ \top_{\Omega} & (x \not \in S) \end{cases}} \quad \in [X, \Omega]$ \tikzmark{init-bottom} \label{line:init-end}
            \While{$S \neq X$} \tikzmark{main-top}
                \State $\letin{P}{\predecessor(Y)}$ \label{line:collect-predecessors}
                \State $\letin{d}{\lambda x. \begin{cases}
                    \Bellman{G}{\sigma}{\gamma}(d)(x) & (x \in P \cap (X \setminus S)) \\
                    d(x) & (\text{otherwise})
                \end{cases}}$ \label{line:update-d}
                \Comment{Selective Bellman update}
                \State $\letin{Y}{\Bigl\{ y \in X \setminus S \;\Big|\; d(y) = \bigsqcap_{z \in X \setminus S} d(z) \Bigr\}}$ \label{line:construct-Y}
                \Comment{Minimal freezing}
                \State $\letin{S}{S \cup Y}$ \label{line:update-S}
                \Comment{Update the set of frozen states}
            \EndWhile \tikzmark{main-bottom}
            \State \Return $d$
        \EndProcedure
    \end{algorithmic}%
    \AddNote{init-top}{init-bottom}{right}{Initialization}
\end{algorithm}

\todo{\checkmark Algorithms 1 \& 2 have the same name!}

\begin{notations}
We use the following notational convention for Algo.~\ref{alg:coalgebraic-dijkstra}.  For $d$, $S$, and $Y$ at the beginning of the $n$-th iteration ($n = 1,2,\dots$) of the main loop in Algo.~\ref{alg:coalgebraic-dijkstra}, we write $d_n$, $S_n$, and $Y_n$, respectively.
We also set $d_0 = \top_{[X,\Omega]}$.

\end{notations}

\todo{\checkmark Say that this example is in parallel to Ex.~\ref{example:shortest-path-run}}
The following example illustrates how Algo.~\ref{alg:coalgebraic-dijkstra} operates on a tree-like structure; it is to be compared with a path-like example in Ex.~\ref{example:shortest-path-run}.


\todo{We can just show the table, moving the rest to the appendix.}
\begin{example}[shortest binary tree, continued from Ex.~\ref{example:shortest-bin-tree}]
    \label{example:shortest-bin-tree-run}
    Let $\gamma \colon X \to \BB \times \powsetfin(\RRpos \times X^2)$ be the weighted $G$-graph in Fig.~\ref{fig:coalg-shortest-binary-tree}.
    The execution of Algo.~\ref{alg:coalgebraic-dijkstra} on $\gamma$ is  in Table~\ref{table:run-shortest-bin-tree}.

    

\end{example}

\begin{table}[tbp]
    \centering
    \caption{The run of Algo.~\ref{alg:coalgebraic-dijkstra} for the weighted $G$-graph $\gamma$ in Fig.~\ref{fig:coalg-shortest-binary-tree}.}
    \label{table:run-shortest-bin-tree}
    \scriptsize
    \begin{tabular}{cccccccc}
        \toprule
        $n$ & $d_n(0)$ & $d_n(1)$ & $d_n(2)$ & $d_n(3)$ & $d_n(4)$  & $S_n$ & $Y_n$\\
        \midrule
        $0$ & $\infty$ & $\infty$ & $\infty$ & $\infty$ & $\infty$ & & \\
        $1$ & $0$      & $\infty$ & $\infty$ & $\infty$ & $\infty$ &  $\{0\}$ & $\{0\}$ \\
        $2$ & $0$      & $\infty$ & $3$      & $\infty$ & $\infty$ &  $\{0, 2\}$ & $\{2\}$ \\
        $3$ & $0$      & $4$      & $3$      & $\infty$ & $\infty$ & $\{0, 1\}$ & $\{1\}$ \\
        $4$ & $0$      & $4$      & $3$      & $9$      & $\infty$ & $\{0, 1, 3\}$ & $\{3\}$ \\
        $5$ & $0$      & $4$      & $3$      & $9$      & $\infty$ & $\{0, 1, 3, 4\}$ & $\{4\}$ \\
        \bottomrule
    \end{tabular}
\end{table}

It is easy to see the termination of Algo.~\ref{alg:coalgebraic-dijkstra} for finite $X$: $S$ grows strictly in every iteration.


\begin{proposition}\label{prop:termination}
 If the state space $X$ of the weighted $G$-graph is finite, 
    Algo.~\ref{alg:coalgebraic-dijkstra} terminates.
\end{proposition}

\todo{\checkmark Introduce the terminology ``valuation'' somewhere earlier.}
 Correctness requires
 delicate conditions on $(G, \sigma)$.
Nevertheless, it holds in general that the sequence  $(d_i)_i$ of valuations, constructed iteratively in Algo.~\ref{alg:coalgebraic-dijkstra}, is descending.
\begin{lemma}\label{lem:descending}
    For a finite set $X$ and a weighted $G$-graph $\gamma \colon X \to \BB \times \powsetfin(G X)$,
    we have $\top = d_0 \sqsupseteq d_1 \sqsupseteq d_2 \sqsupseteq \dots \sqsupseteq \gfp \Bellman{G}{\sigma}{\gamma}$.
\end{lemma}


\subsection{Correctness for finitely supported functors}
\label{subsec:fin-supp}
We discuss correctness of the coalgebraic Dijkstra algorithm, presenting an abstract necessary and sufficient condition (Thm.~\ref{thm:correctness-G}) that captures the delicate difference in Table~\ref{table:examples-problems}. 

\begin{definition}
 We say that $\textsc{CoalgDijkstra}_{(G,\sigma)}$ in 
Algo.~\ref{alg:coalgebraic-dijkstra} is \emph{correct} if it returns the 
solution to the CSPP on $(G,\sigma)$, that is, $d = \gfp \Bellman{G}{\sigma}{\gamma}$.
\end{definition}

We use the following technical notion in the discussions below.
Intuitively, the subset $\Omega^{\sigma} \subseteq \Omega$ collects all the elements of $\Omega$ that can possibly appear in Algo.~\ref{alg:coalgebraic-dijkstra}.

\begin{definition}\label{def:Omega-sigma}
    We define the subset $\Omega^{\sigma} \subseteq \Omega$ using the recursion
    $\Omega^{\sigma}_0 = \{ \finalweight, \top_{\Omega} \}$,
    $\Omega^{\sigma}_{n + 1} = \Omega^{\sigma}_n \cup \{ \sigma(a) \mid a \in G(\Omega^{\sigma}_{n}) \} $, for all $n \in \NN$, such that
    $\Omega^{\sigma} = \bigcup_{n = 0}^{\infty} \Omega^{\sigma}_n$. Here $\xi$ is from Def.~\ref{def:pointedWeightDom}.
\end{definition}

We introduce the following categorical axiom for a functor $G$.

\begin{definition}[finitely supported functor]\label{def:finitelySupportedFunctor}
    Let $G \colon \Sets \to \Sets$ be a functor. A \emph{finite support} for $G$ is a natural transformation  $\supp \colon G \to \powsetfin$, 
to the finite power set functor $\powsetfin$, such that 
\begin{quote}
 for every set $X$ and $t \in GX$, $\supp_X(t)$ is the smallest subset $S \subseteq X$ such that $t \in G S$. 
\end{quote}
It is obvious that a finite support $\supp$ for $G$, if such exists, is unique. 

    A functor $G$ is called \emph{finitely supported} if it admits a finite support.
    A finite support $\supp$ is called \emph{non-empty} if $\supp_X(t)$ is not empty for every set $X$ and $t \in GX$.
\end{definition}

The class of finitely supported functors is broad.
It contains the identity functor, the (non-empty) finite powerset functor $\powsetfin,\powsetfinne$, 
the distribution functor $\distr$, and constant functors (Lem.~\ref{lem:fin-supp-examples}). The class is  closed under sum, product, and composition (Lem.~\ref{lem:fin-supp-closed}).
In particular, every polynomial functor is finitely supported. Concretely, we define the following.
\begin{enumerate}
    \item $\supp^{\idfunc}_X \colon X \to \powsetfin(X)$ is defined as $\supp^{\idfunc}_X(x) = \{x\}$.
    \item $\supp^{\powsetfin}_X \colon \powsetfin(X) \to \powsetfin(X)$ is defined as $\supp^{\powsetfin}_X(A) = A$.
    \item $\supp^{\powsetfinne}_X \colon \powsetfinne(X) \to \powsetfin(X)$ is defined as $\supp^{\powsetfinne}_X(A) = A$.
    \item $\supp^{\distr}_X \colon \distr(X) \to \powsetfin(X)$ is defined as $\supp^{\distr}_X(\rho) = \{ x \in X \mid \rho(x) > 0 \}$.
    \item For a set $V$, $\supp^{V}_X \colon V \to \powsetfin(X)$ is defined as $\supp^{V}_X(v) = \emptyset$.
\end{enumerate}
\begin{lemma}\label{lem:fin-supp-examples}
    The natural transformations $\supp^{\idfunc}$, $\supp^{\powsetfin}$, $\supp^{\powsetfinne}$, $\supp^{\distr}$ and $\supp^{V}$ defined above are finite supports for $\idfunc$, $\powsetfin$, $\powsetfinne$, $\distr$ and $V$, respectively.
\end{lemma}
\begin{lemma}\label{lem:fin-supp-closed}
    If $G_0$ and $G_1$ are finitely supported functors, then so are $G_0 + G_1$, $G_0 \times G_1$ and $G_1 \circ G_0$.
\end{lemma}

One may wonder the relationship between Def.~\ref{def:finitelySupportedFunctor} and another well-established notion with a similar intuition, namely \emph{finitariness} (see e.g.~\cite{AdamekR94}). 
We have the following result.
\begin{proposition}\label{lem:finitary-taut-finitely-supp}
    If a functor $F \colon \Sets \to \Sets$ is finitary (\versionswitch{\cite[Def.~E.1]{CoalgDijkstraArXiv}}{Def.~\ref{def:finitary-functor}}) and taut (\versionswitch{\cite[Def.~E.2]{CoalgDijkstraArXiv}}{Def.~\ref{def:taut-functor}}), then $F$ is finitely supported.
\end{proposition}
The proof is based on Manes's results on finitary and taut functors \cite{Manes1998,Manes2002}.
The tautness assumption is necessary: the naturality of $\theta$ fails without it; and it is known that a  finitary functor may not be taut~\cite{Manes1998}.
The converse does not hold: there is a finitely supported functor that is not taut \versionswitch{\cite[Prop.~E.3]{CoalgDijkstraArXiv}}{Prop.~\ref{prop:fin-supp-but-not-taut}}.
Lem.~\ref{lem:fin-supp-examples} is an immediate consequence of Prop.~\ref{lem:finitary-taut-finitely-supp}.
\todo[inline]{Let's study the converse of Prop.~\ref{lem:finitary-taut-finitely-supp}. If yes, then we modify the statement of prop.~\ref{lem:finitary-taut-finitely-supp}. If not, then we write there ``While the converse of Prop.~\ref{lem:finitary-taut-finitely-supp} seems plausible, we are yet to find its proof (future work).''}



Let's move on to correctness. The core of our categorical axiomatics is the notion of \emph{expansiveness} of $G$-algebra $\sigma \colon G\Omega \to \Omega$  (as part of a transition modality, Def.~\ref{def:transtionModality}). 
Intuitively, expansiveness means that the accumulation of weights along a path does not decrease.
For example, in the SPP (Ex.~\ref{examlpe:shortest-path}), the transition modality is given by $\sigma(a,b) = a + b$, and expansiveness amounts to the condition $b \le a + b$ for all $a \in \RRpos$ and $b \in \RRinfpos$.
\begin{definition}[expansiveness]
    Let $\Omega$ be a pointed weight domain (Def.~\ref{def:pointedWeightDom}), $G \colon \Sets \to \Sets$ be a functor and $\supp$ be a finite support of $G$.
    A $G$-algebra $\sigma \colon G \Omega \to \Omega$ is called \emph{expansive} if $b \sqsubseteq \sigma(t)$ for every $t \in G\Omega^{\sigma}$ and $b \in \supp(t)$.
\end{definition}

Here is our main theorem.
To show (\ref{item:correctness-G}$\Rightarrow$\ref{item:expansive-G}), we assume that $\sigma$ is not expansive and construct a counterexample on which Algo.~\ref{alg:coalgebraic-dijkstra} is not correct (see the proof in \versionswitch{\cite[\S{}E.1]{CoalgDijkstraArXiv}}{\S\ref{subsec:pfMainThm}}).
\begin{theorem}[correctness]\label{thm:correctness-G}
    Let $G \colon \Sets \to \Sets$ be a weak-pullback-preserving, non-empty finitely supported functor.
    Let $\sigma \colon G\Omega \to \Omega$ be a transition modality on a pointed weight domain $\Omega$.
    The following statements are equivalent:
    \begin{enumerate}
        \item \label{item:expansive-G} The transition modality $\sigma$ is expansive.
        \item \label{item:correctness-G} For every weighted $G$-graph $\gamma \colon X \to \BB\times \powsetfin(GX)$, $\textsc{CoalgDijkstra}_{(G,\sigma)}$ is correct.
    \end{enumerate}
\end{theorem}

\begin{example}[the dynamic game, continued from Ex.~\ref{example:reachability-game-dynamic-game}]\label{example:incorrectness-of-bardi-lopez-condition}
  Recall Ex.~\ref{example:reachability-game-dynamic-game}, where
    the weight domain is $(\RRinfpos, \le, \finalweight)$, the functor $G$ is $G X = \powsetfinne([\ell_0, L] \times X)$, and the transition modality $\sigma_r \colon \powsetfinne([\ell_0, L] \times \Omega) \to \Omega$ is  $\sigma_r(A) = \min_{(a,b) \in A} (a + rb)$ for a fixed discount rate $r \in (0, 1]$.
    The \CSPP{} on $(G, \sigma_r)$ is a special case\footnote{In our setting, two players make moves alternately, an (inessential) restriction not present in \cite{BardiLopez2016}. } of the dynamic game studied by Bardi and L{\'o}pez \cite{BardiLopez2016}.
    The paper has the following condition \cite[Condition~1]{BardiLopez2016} for their Dijkstra acceleration to work when $r < 1$:
    \begin{equation}\label{eq:bardi-lopez-condition-correct}
        L \sum_{j=0}^{n-1} r^j + r^n \finalweight \;\le\; \frac{\ell_0}{1 - r} 
    \quad\text{    for all $n < N$,}
    \end{equation}
where $N$ is the number of steps in which their Dijkstra-style algorithm terminates. This condition uses  universal quantification over $n$, which in general makes its verification harder. 

We note that, in the earlier version \cite{BardiLopez2015} of the same paper, the following condition is given:
\begin{equation}\label{eq:bardi-lopez-condition-incorrect}
    L + r \finalweight \;\le\; \frac{\ell_0}{1 - r}.
\end{equation}
    This  condition \eqref{eq:bardi-lopez-condition-incorrect} in the early version~\cite{BardiLopez2015} has apparently been corrected, into \eqref{eq:bardi-lopez-condition-correct} in the published version~\cite{BardiLopez2016}. The condition \eqref{eq:bardi-lopez-condition-correct} is stricter, and harder to check (due to universal quantification). 


    This example demonstrates that conditions for the Dijkstra acceleration to work can be subtle, and that our necessary and sufficient condition---categorically formulated in Thm.~\ref{thm:correctness-G}---is useful.
    It is easy to see that the earlier condition \eqref{eq:bardi-lopez-condition-incorrect} does not imply the expansiveness of $\sigma$, meaning our theory can raise a red flag to the conjectured condition \eqref{eq:bardi-lopez-condition-incorrect}.

It can be seen that the corrected condition \eqref{eq:bardi-lopez-condition-correct}---to be precise, its uniform version which requires \eqref{eq:bardi-lopez-condition-correct} for any $n\in \NN$---is essentially equivalent to our condition of expansiveness. 
See~\versionswitch{\cite[\S{}E.2]{CoalgDijkstraArXiv}}{\ref{appendix:example:incorrectness-of-bardi-lopez-condition}}.
There, we observe that the uniform version of \eqref{eq:bardi-lopez-condition-correct} implies $\ell_{0}=L$, forcing all the moves to have the same stepwise reward. This severely restricts the expressivity of the studied game-like formalism.
\end{example}

\subsection{Complexity}\label{subsec:complexity}
We fix a pointed weight domain $\Omega$, a functor $G \colon \Sets \to \Sets$, and a transition modality $\sigma \colon G\Omega \to \Omega$.
Let $X$ be a finite set and $\gamma \colon X \to \BB \times \powsetfin(GX)$ be a weighted $G$-graph.
We assume that it takes $\order(T)$ time to compute $\Bellman{G}{\sigma}{\gamma}(d)(x)$ for $d \colon X\to \Omega$, $x \in X$.
We define $V := \sizeof{X}$ and $E := \sum_{x \in X} \sizeof{\successor(x)}$.
Furthermore, assume that for each $y \in X$, the set of predecessors $\predecessor(y)$ is precomputed.
We have the following complexity result.
\begin{proposition}\label{prop:complexity}
    The time complexity of Algo.~\ref{alg:coalgebraic-dijkstra} is $\order(TE + V^2)$.
\end{proposition}

It is well known that the complexity of the classical Dijkstra algorithm can be improved by the use of a Fibonacci heap \cite{FredmanTarjan1987},  a data structure implementing a priority queue.
For Algo.~\ref{alg:coalgebraic-dijkstra}, we can use a Fibonacci heap to improve the complexity of the computation of $Y$.

Furthermore, for a more precise complexity result, we replace Line~\ref{line:update-d} in Algo.~\ref{alg:coalgebraic-dijkstra} with
\[
\letin{d}{\lambda x . \begin{cases}
    d(x) \sqcap \bigsqcap_{\substack{a \in \pi_1(\gamma(x)) \\ \supp^G_X(a) \cap Y \ne \emptyset}} \sigma(Gd(a)) & \text{if $x \in P \cap (X \setminus S)$}
    \\
    d(x) & \text{otherwise}.
\end{cases}}
\]
Intuitively, it is unnecessary to recompute the entire $\Bellman{G}{\sigma}{\gamma}(d)(x)$ each time; only the portion related to the set $Y$ needs to be updated.

Now,
Algo.~\ref{alg:coalgebraic-dijkstra-priority-queue} is an improvement of Algo.~\ref{alg:coalgebraic-dijkstra} using a Fibonacci heap and the above replacement.
\begin{algorithm}[t]
    \caption{The coalgebraic Dijkstra algorithm with a Fibonacci heap.}\label{alg:coalgebraic-dijkstra-priority-queue} \footnotesize
    \begin{algorithmic}[1]
        \Procedure{CoalgDijkstra${}_{(G, \sigma)}^{\mathrm{F}}$}{$\gamma \colon X \to \BB \times \powsetfin(GX)$}
            \State $\letin{S}{\{ x \in X \mid \pi_0(\gamma(x)) = \btt \}}$
            \State $\letin{Y}{S}$
            \State $\letin{d}{\lambda x \in X. \begin{cases} \finalweight & (x \in S) \\ \top_{\Omega} & (x \not \in S) \end{cases}} \quad \in [X, \Omega]$
            \State Let $Q$ be the empty Fibonacci heap.
            \While{$S \neq X$}
                \State $\letin{P}{\predecessor(Y)}$
                \State $\letin{d}{\lambda x . \begin{cases}
                    d(x) \sqcap \bigsqcap_{\substack{a \in \pi_1(\gamma(x)) \\ \supp^G_X(a) \cap Y \ne \emptyset}} \sigma(Gd(a)) & \text{if $x \in P \cap (X \setminus S)$}
                    \\
                    d(x) & \text{otherwise.}
                    \end{cases}}$
                \Comment{Selective Bellman update}
                \For{$x \in P \cap (X \setminus S)$} \label{Line:update-insert}
                    \If{ the key $x$ is in $Q$}
                        \State Update the value of the key $x$ to $d(x)$ in $Q$
                    \Else
                        \State Insert $(d(x), x)$ to $Q$
                    \EndIf
                \EndFor
                \State Get $\{ (a, y_0), \dots, (a, y_m) \}$ from $Q$ such that $a$ is the minimum in $Q$ \label{Line:get-Q}
                \State Delete $(a,y_0), \dots, (a, y_m)$ from $Q$ \label{Line:delete-Q}
                \State $\letin{Y}{\{ y_0, \dots, y_m \}}$
                    \Comment{Minimal freezing}
                \State $\letin{S}{S \cup Y}$
                    \Comment{Update the set of frozen states}
            \EndWhile
            \State \Return $d$
        \EndProcedure
    \end{algorithmic}
\end{algorithm}
For a refined complexity analysis, we assume it takes $\order(T')$ time to compute $\sigma((Gd)(a))$, where $d \colon X\to \Omega$, $x \in X$ and $a \in GX$.

\begin{theorem}\label{thm:complexity-fibonacci}
    When using the Fibonacci heap as a priority queue $Q$,
    the time complexity of Algo.~\ref{alg:coalgebraic-dijkstra-priority-queue} is $\order(T'E + V \log V)$.
\end{theorem}
This complexity coincides with the classical complexity result for the (classical) Dijkstra (namely $\order(E + V\log V)$---note that  $T' = \order(1)$ for the SPP). 


\section{Conclusion and Future Work}
We presented a categorical generalization of the Dijkstra-style acceleration in various graph-based optimization problems. Whether Dijkstra is correct relies delicately on problem settings (Table~\ref{table:examples-problems}); we presented a categorical necessary and sufficient condition, formulated on the notions of finitely supported functor and expansive transition modality. Our general complexity analysis generalizes the classical one, too.

We used $\Sets$ as the space of collections of states.
One direction for future work would be to determine whether it is possible to generalize $\Sets$ to arbitrary categories $\mathbb{C}$, in order to obtain the coalgebraic Dijkstra algorithm over topological spaces or nominal sets.

Further, the background of the definition of the Bellman operator $\Bellman{G}{\sigma}{\gamma}$ lies in the theory of fibrations.
In this paper, we did not explicitly use the notion of fibrations.
Another direction for future work is to examine whether the fibrational viewpoint provides useful applications or generalizations.

\bibliography{coalg-dijkstra}
\ifdefined\isarxiv
\appendix

\section{Proofs for \S\ref{sec:prelim}}

\begin{proof}[Proof Sketch for Lem.~\ref{lem:CC}] That the sequence~\eqref{eq:CC} is descending is proved by transfinite induction. If it does not stabilize, then the sequence defines an injective embedding of the proper class $\mathbf{Ord}$ of all ordinals in $L $, which cannot exist since $L $ is a small set. The limit is obviously a fixed point of $\Psi$. Given any fixed point $x$ (with $x=\Psi(x)$), we can show by induction that $x\sqsubseteq \Psi^{\alpha}(\top)$ for any ordinal $\alpha$; this proves that the limit is $\nu \Psi$.
\end{proof}

\begin{proof}[Proof for Lem.~\ref{lem:partial-Kleene-chain}]
    We prove by induction on $n$.

    If $n = 0$, we have
    \[
    \begin{aligned}
        d_1(x) & = \Phi(d_0)(x) = \Phi(\top)(x) & \text{for $x \in P_1$,}
        \\
        d_1(x) & = d_0(x) = \top \sqsupseteq \Phi(\top)(x) & \text{for $x \not\in P_1$.}
    \end{aligned}
    \]
    Hence, we have $\Phi(\top) \sqsubseteq d_1$ and $\Phi(d_0) \sqsubseteq d_1$.
    Since we have $d_1(x) = \Phi(\top)(x) \sqsubseteq \top = d_0(x)$ for $x \in P_1$ and $d_1(x) = d_0(x)$ for $x \not\in P_1$, $d_1 \sqsubseteq d_0$ holds.

    For $n + 1$, we have
    \[
    \begin{aligned}
        \Phi^{n + 1}(\top)
        & \sqsubseteq \Phi(d_n)
        & \text{by the induction hypothesis $\Phi^n(\top) \sqsubseteq d_n$}
        \\
        & \sqsubseteq d_{n + 1}
        & \text{by the induction hypothesis $\Phi(d_n) \sqsubseteq d_{n+1}$.}
    \end{aligned}
    \]
    We have
    \begin{equation}\label{eq:Phi-d-le-d-n+1}
        \begin{aligned}
            \Phi(d_{n + 1})
            & \sqsubseteq \Phi(d_n)
            & \text{by the induction hypothesis $d_{n + 1} \sqsubseteq d_n$ and the monotonicity of $\Phi$}
            \\
            & \sqsubseteq d_{n + 1}
            & \text{by the induction hypothesis $\Phi(d_n) \sqsubseteq d_{n + 1}$.}
        \end{aligned}
    \end{equation}
    Hence, for $x \in P_{n + 2}$, we have
    \[
    \begin{aligned}
        d_{n + 2}(x)
        & = \Phi(d_{n + 1})(x)
        \\
        & \sqsubseteq d_{n + 1}(x)
        & \text{by \eqref{eq:Phi-d-le-d-n+1}.}
    \end{aligned}
    \]
    For $x \not\in P_{n + 2}$, we have $d_{n + 2}(x) = d_{n + 1}(x)$.
    Therefore, we have $d_{n + 2} \sqsubseteq d_{n + 1}$.
    Furthermore, we have
    \[
        \begin{aligned}
            d_{n + 2}(x) & = \Phi(d_{n + 1})(x)
            & \text{for $x \in P_{n + 2}$}
            \\
            d_{n + 2}(x) & = d_{n + 1}(x)
            \sqsupseteq \Phi(d_{n + 1})(x)
            & \text{for $x \not\in P_{n + 2}$, by \eqref{eq:Phi-d-le-d-n+1}.}
        \end{aligned}
    \]
    Thus, we have $\Phi(d_{n+1}) \sqsubseteq d_{n+2}$.
\end{proof}

\section{More instances of the CSPP}\label{sec:moreInstancesOfTheCSPP}

\begin{example}[reachability]
    Let $\Omega = (\{ 0, \infty \}, \le, 0)$,
    $G \colon \Sets \to \Sets$ be the identity functor $G = \idfunc$ and
    $\sigma \colon \Omega \to \Omega$ be a transition modality defined as $\sigma(a) = a$.
    For a finite set $X$, a weighted $G$-graph $\gamma \colon X \to \BB \times \powsetfin(X)$ is a directed graph such that
    (1) $X$ is a set of vertices,
    (2) a set of targets $Z = \{ x \in X \mid \pi_0(\gamma(x)) = \btt \}$ is specified, and
    (3) for $x,y \in X$, there is an edge from $x$ to $y$ if $y \in \pi_1(\gamma(x))$.
    The Bellman operator $\Bellman{G}{\sigma}{\gamma} \colon [X, \Omega] \to [X, \Omega]$ is 
    \[
    \Bellman{G}{\sigma}{\gamma}(d)(x) = \begin{cases}
        0 & (\text{if }\gamma(x) = (\btt, A)) \\
        \bigsqcap_{x \in A} d(x) & (\text{if }\gamma(x) = (\bff, A)).
    \end{cases}
    \]
    An example of a weighted $G$-graph $\gamma \colon X \to \BB \times \powsetfin(X)$ is shown in Fig.~\ref{fig:coalg-shortest-path}.
    The greatest fixed point $\gfp \Bellman{G}{\sigma}{\gamma} \colon X \to \Omega$ represents the reachability to one of the targets:
    $\gfp \Bellman{G}{\sigma}{\gamma}(x) = 0$ if there is a path from $x$ to $z \in Z$, and
    $\gfp \Bellman{G}{\sigma}{\gamma}(x) = \infty$ otherwise.
\end{example}

\begin{example}[widest path]\label{example:widest-path}
    Let $\Omega = (\RRinfpos, \ge, \infty)$ and
    $G \colon \Sets \to \Sets$ be a functor defined as $GX = \RRpos \times X$.
    Let $\sigma \colon G \Omega \to \Omega$ be a transition modality defined as $\sigma(a,b) = a \sqcup b = \min(a,b)$.
    Note that the order $(\ge)$ is opposite to the ordinary order $(\le)$.
    Hence, $x \sqcup y$ is the minimum value of $x$ and $y$.
    For a finite set $X$, a weighted $G$-graph $\gamma \colon X \to \BB \times \powsetfin(\RRpos \times X)$ is a directed graph such that a subset $Z \subseteq X$ of targets is specified, and each edge has its width $a \in \RRpos$.
    For $d \in [X, \Omega]$, the Bellman operator $\Bellman{G}{\sigma}{\gamma}$ is
    \[
    \Bellman{G}{\sigma}{\gamma}(d)(x) = \begin{cases}
        \infty & (\text{if }\gamma(x) = (\btt, A)) \\
        \max\{ \min(a, d(y)) \mid (a, y) \in A \} & (\text{if }\gamma(x) = (\bff, A)).
    \end{cases}
    \]
    For the path
    $ x_n \xrightarrow{a_n} x_{n-1} \xrightarrow{a_{n-1}} \cdots \xrightarrow{a_1} x_0 \in Z $
    in the graph $\gamma \colon X \to \BB \times \powsetfin(\RRpos \times X)$,
    we define the width of the path as $\min_{0 < i \le n} a_i $.
    We can show that $\gfp \Bellman{G}{\sigma}{\gamma}(x)$ is the value of the widest path from $x$ to one of the targets.
\end{example}

\begin{example}[the binary reachability game]\label{example:bin-reachability-game}
    Let $\Omega = (\{ 0, \infty \}, \le, 0)$ and
    $G \colon \Sets \to \Sets$ be a functor defined as $GX = X \times X$.
    Let 
    $\sigma \colon G\Omega \to \Omega$ be a transition modality defined as
    $\sigma(a,b) = a \sqcup b$.
    For a finite set $X$ of vertices, a weighted $G$-graph $\gamma \colon X \to \BB \times \powsetfin(X \times X)$ is given by the following structure:
    (1) A set $Z = \{ x \in X \mid \pi_0(\gamma(x)) = \btt \}$ of targets is specified.
    (2) For each vertex $x \in X$, there is a finite set $\{ (y^j_0, y^j_1) \}_{j \in J_x}$ of pairs of vertices.

    The Bellman operator $\Bellman{G}{\sigma}{\gamma}$ is
    \[
    \Bellman{G}{\sigma}{\gamma}(d)(x) = \begin{cases}
        0 & (\text{if }\gamma(x) = (\btt, A)) \\
        \bigsqcap_{(y_0, y_1) \in A} (d(y_0) \sqcup d(y_1)) & (\text{if }\gamma(x) = (\bff, A)).
    \end{cases}
    \]
    We define a \emph{pebble-game} on $\gamma$ as follows: The game is played between two players---Alice and Bob.
    Starting with a pebble placed on a vertex $x \in X$, the game proceeds by iterating the following procedure.
    \begin{enumerate}
        \item If the vertex where the pebble is placed is in $Z$, Alice wins.
        \item Otherwise, Alice chooses a label $j \in J_x$.
        \item Bob moves the pebble from $x$ to $y^j_0$ or $y^j_1$.
    \end{enumerate}
    Starting with a pebble on a vertex $x \in X$, we ask whether Alice can always win in finitely many steps, regardless of Bob’s moves.
    The answer is $\gfp\Bellman{G}{\sigma}{\gamma}(x)$:
    \[
    \gfp \Bellman{G}{\sigma}{\gamma}(x) =
    \begin{cases}
        0 & \text{if there is a winning strategy of Alice,} \\
        \infty & \text{otherwise.}
    \end{cases}
    \]
\end{example}

\section{Conditions on transition modalities for specific functors}\label{sec:conditions-specific-functors}

\subsection{The condition of transition modalities for $GX = V \times X$}\label{sec:VxX}
Let $G \colon \Sets \to \Sets$ be a functor defined by $GX = V \times X$ for some set $V$.

This section presents a special case of \S\ref{sec:VxXt} and~\ref{subsec:fin-supp}.
The definitions, propositions, and proofs given here are obtained by taking $t=1$ in \S\ref{sec:VxXt}.
Even so, this special case already covers important examples such as the SPP and widest-path problems.
Moreover, the main ideas can be presented more simply in this setting.
We therefore begin with this case.

The following lemma gives a characterization of transition modalities $\sigma$ for $G = V \times \blank$.
\begin{lemma}\label{lem:transition-modality-OmegaxX}
    A $G$-algebra $\sigma \colon V \times \Omega \to \Omega$ is a transition modality (cf.\ Def.~\ref{def:transtionModality}) if and only if
    $\sigma(a, \bigsqcap B) = \bigsqcap_{b \in B} \sigma(a,b)$, for each $a \in V$ and $B \subseteq \Omega$.
\end{lemma}

For a transition modality $\sigma \colon V \times \Omega \to \Omega$, the set $\Omega^{\sigma}$ in Def.~\ref{def:Omega-sigma} is given by
$\Omega^{\sigma}_0 = \{ \finalweight , \top_{\Omega}\}$,
$\Omega^{\sigma}_{n+1} = \Omega^{\sigma}_n \cup \{ \sigma(a, b) \in \Omega \mid a \in V,\ b \in \Omega^{\sigma}_n\}$ for $n \in \NN$, and $\Omega^{\sigma} = \bigcup_{n=0}^{\infty} \Omega^{\sigma}_n$.

To prove the correctness of Algo.~\ref{alg:coalgebraic-dijkstra}, we introduce the notion of \emph{run paths} in weighted $G$-graphs and their \emph{$\sigma$-values}.
A run path $p$ and its $\sigma$-value $\sigma(p)$ generalize, respectively, a path in a weighted graph and its length.
Lem.~\ref{lem:runpath-Phi} relates run paths to Bellman operators.
\begin{definition}[run path and $\sigma$-value]
    For a weighted $G$-graph $\gamma \colon X \to \BB \times \powsetfin(V \times X)$,
    a \emph{run path} of $\gamma$ is a finite sequence $((x_i, a_i))_{i = 0}^{M}$ such that
    \begin{enumerate}
        \item $(x_i, a_i) \in X \times (V + \{ \star \})$ for $i \in [M + 1]$,
        \item $\pi_0(\gamma(x_0)) = \btt$ and $a_0 = \star$, and
        \item $(a_{i+1}, x_{i}) \in \pi_1(\gamma(x_{i + 1}))$ for $i \in [M]$.
    \end{enumerate}
    We write $\runpath_n(\gamma,x)$ for the set of run paths $((x_i, a_i))_{i=0}^{M}$ such that $x_M = x$ and $M \le n$.
    For a run path $p = ((x_i, a_i))_{i = 0}^{M}$, we define the \emph{$\sigma$-value} $\sigma(p)$ of $p$ recursively as
    $\sigma((x_0, a_0)) = \finalweight$ (here $\xi$ is from Def.~\ref{def:pointedWeightDom}) and
    $\sigma((x_{i}, a_{i})_{i=0}^{M + 1}) = \sigma(a_{M+1}, \sigma((x_{i}, a_{i})_{i=0}^{M}))$.
\end{definition}

\begin{lemma}\label{lem:runpath-Phi}
Let $\sigma \colon V \times \Omega \to \Omega$ be a transition modality,
$X$ be a finite set,
and $\gamma \colon X \to \BB \times \powsetfin(V \times X)$ be a coalgebra.
For each $n \in \NN$ and $x \in X$, 
$ (\Bellman{G}{\sigma}{\gamma})^{n + 1}(\top_{[X,\Omega]})(x) = \bigsqcap_{p \in \runpath_n(\gamma, x)} \sigma(p)$ holds.
Furthermore, $\gfp \Bellman{G}{\sigma}{\gamma}(x) = \bigsqcap_{n=0}^{\infty} \bigsqcap_{p \in \runpath_n(\gamma, x)} \sigma(p)$ holds.
\end{lemma}

We now obtain a necessary and sufficient condition on $\sigma$ under which the Algo.~\ref{alg:coalgebraic-dijkstra} returns the greatest fixed point for every weighted $(V \times \blank)$-graph.
\begin{proposition}\label{theorem:VxX}
    Let $\sigma \colon V \times \Omega \to \Omega$ be a transition modality.
    The following statements are equivalent:
    \begin{enumerate}
        \item\label{item:expansion} $b \sqsubseteq \sigma(a, b)$ for each $b \in \Omega^{\sigma}$ and $a \in V$.
        \item\label{item:dijkstra-crrect-OmegaxX} For every weighted $G$-graph $\gamma \colon X \to \BB \times \powsetfin(V \times X)$ on a finite set $X$, $\textsc{CoalgDijkstra}_{(G,\sigma)}$ is correct.
    \end{enumerate}
\end{proposition}
\begin{proof}[Proof sketch]
The proof of (\ref{item:expansion} $\implies$ \ref{item:dijkstra-crrect-OmegaxX}) in Prop.~\ref{theorem:VxX} follows the same basic idea as the standard correctness proof of the classical Dijkstra algorithm.
The main part of the argument is to prove, by induction on $n$, that
$d_{n}(y) \sqsubseteq \gfp\Bellman{G}{\sigma}{\gamma}(y)$ for every $y \in Y_{n}$.

Suppose that there exists $y \in Y_{n+1}$ such that $d_{n+1}(y) \sqsupset \gfp\Bellman{G}{\sigma}{\gamma}(y)$.
By Lemma~\ref{lem:runpath-Phi}, there exists a run path $p \in \runpath_n(\gamma, y)$ such that $\sigma(p) \sqsubset d_{n+1}(y)$.
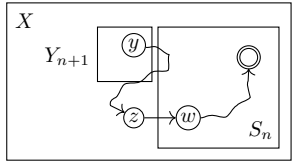
\[
    \begin{tikzpicture}[scale=0.8, transform shape]
        \draw (-1.5, 3.4) rectangle (3.2, 0.8);
        \draw (0, 3) rectangle (0.9, 2.1);
        \draw (1.0, 3) rectangle (3, 1);
        \node at (-1.2, 3.1) {$X$};
        \node at (-0.5, 2.5) {$Y_{n+1}$};
        \node at (2.7, 1.3) {$S_n$};
        \node[state] (y) at (0.6, 2.7) {$y$};
        \node[state] (z) at (0.6, 1.5) {$z$};
        \node[state] (z') at (1.5, 1.5) {$w$};
        \node[state, accepting] (target) at (2.5, 2.5) {\phantom{$a$}};
        \path[->] (z) edge (z');
        \path[->, looseness=1] (z') edge[out=0, in=270, decorate, segment length=4mm, decoration={snake, amplitude=0.3mm}] (target);
        \path[->, looseness=3] (y) edge[out=-10, in=150, decorate, segment length=4mm, decoration={snake, amplitude=0.3mm}] (z);
    \end{tikzpicture}
\]

From $p$, we extract a ``subpath'' $q \in \runpath_i(\gamma, z)$ $(i \le n)$ starting at a vertex $z \in X \setminus S_n$ such that every vertex of $q$ except $z$ belongs to $S_n$.
By the condition~\ref{item:expansion}, we have $\sigma(q) \sqsubseteq \sigma(p)$.

By the definition of $Y_{n+1}$ (minimality) and $z \not\in S_n$, we have $d_{n+1}(y) \sqsubseteq d_{n+1}(z) \sqsubseteq d_{i+1}(z)$.

Let $w$ be the successor $w$ of $z$ in $q$ and 
$a$ be the value of the edge from $z$ to $w$ in $q$.
Combining the above inequalities, we have
\[\begin{aligned}
d_{n + 1}(y)
& \sqsubseteq d_{i + 1}(z)
\sqsubseteq \Bellman{G}{\sigma}{\gamma}(d_i)(z) \\
& \sqsubseteq \sigma(a, d_i(w))
\quad \text{by the definition of $\Bellman{G}{\sigma}{\gamma}(d_i)$}
\\
& \sqsubseteq \sigma(a, \gfp\Bellman{G}{\sigma}{\gamma}(w))
\quad \text{by the induction hypothesis for $i < n + 1$}
\\
& \sqsubseteq \sigma(q)
\sqsubseteq \sigma(p) \sqsubset d_{n+1}(y).
\end{aligned}
\]
This is a contradiction, and hence we have $d_{n+1}(y) \sqsubseteq \gfp\Bellman{G}{\sigma}{\gamma}(y)$ for every $y \in Y_{n+1}$.
\end{proof}


\begin{example}[SPP, continued from Ex.~\ref{example:shortest-path-run}]
    We have $\Omega^{\sigma} = \RRinfpos = \Omega$ since 
    $\Omega^{\sigma}_0 = \{ 0, \infty \}$ and
    $\Omega^{\sigma}_1 = \Omega^{\sigma}_0 \cup \{ a + b \mid a \in \RRpos, b \in \Omega^{\sigma}_0 \} = \RRinfpos$.
    For every $a \in \RRpos$ and $b \in \Omega^{\sigma} = \RRinfpos$,
    we have $b \sqsubseteq a + b = \sigma(a,b)$.
    Hence, the condition~\ref{item:expansion} of Prop.~\ref{theorem:VxX} holds. Therefore, Algo.~\ref{alg:coalgebraic-dijkstra} returns the solution to the CSPP on $(G,\sigma)$, namely the greatest fixed point $\gfp \Bellman{G}{\sigma}{\gamma}$, which gives the length of a shortest path from each vertex to a target.
\end{example}
\begin{example}[SPP with negative edges, continued from Ex.~\ref{example:shortest-path-negative}]
    We have $\Omega^{\sigma} = \RRinf$ since 
    $\Omega^{\sigma}_0 = \{ 0, \infty \}$ and
    $\Omega^{\sigma}_1 = \Omega^{\sigma}_0 \cup \{ \sigma(a, b) \mid a \in \RR, b \in \Omega^{\sigma}_0 \} = \RRinf$.
    For $a = -1 \in \RR$ and $b = 1 \in \Omega^{\sigma}$, we have $b = 1 \sqsupset 0 = \sigma(-1, 1) = \sigma(a, b)$.
    Hence, the condition \ref{item:expansion} of Prop.~\ref{theorem:VxX} does not hold.
    Consequently, there exists a weighted $G$-graph $\gamma \colon X \to \BB \times \powsetfin(\RR \times X)$ for which Algo.~\ref{alg:coalgebraic-dijkstra} does not return the greatest fixed point $\gfp \Bellman{G}{\sigma}{\gamma}$.

    The proof of Prop.~\ref{theorem:VxX} yields such a counterexample. Starting from
    $1 \sqsupset 0 = \sigma(-1,1) = \sigma(-1,\sigma(1,0))$,
    we define a finite set $X$ and a weighted $G$-graph $\gamma \colon X \to \BB \times \powsetfin(\RR \times X)$ by
    \[
        X = \{ 0, 1 \}, \qquad
        \gamma(0) = (\btt, \emptyset), \qquad
        \gamma(1) = (\bff, \{ (1,0), (-1,1) \}), \qquad
        \begin{tikzpicture}[node distance=1cm, auto]
            \node[state,accepting](0) {$0$};
            \node[state] (1) [left=0.4cm of 0] {$1$};
            \path[->] (1) edge node {\tiny $1$} (0);
            \path[->] (1) edge [loop left] node[left] {\tiny $-1$} ();
        \end{tikzpicture}.
    \]
    For the valuation $d \colon X \to \Omega$ returned by $\textsc{CoalgDijkstra}_{(G,\sigma)}(\gamma)$, we have $d(0) = 0$ and $d(1) = 1$.
    However, the greatest fixed point satisfies $\gfp \Bellman{G}{\sigma}{\gamma}(0) = 0$ and $\gfp \Bellman{G}{\sigma}{\gamma}(1) = -\infty$.
\end{example}
\begin{example}[SPP with interest/discount rate, continued from Ex.~\ref{example:shortest-path-interest-discount}]
    For both the SPP with interest rate and discount rate, we have $\Omega^{\sigma} = \RRinfpos$.
    For the interest rate, we have $b \le a + a' b = \sigma(a, a', b)$ for each $a \in \RRpos$, $a' \in [1, \infty)$ and $b \in \Omega^{\sigma}$.
    Hence, Algo.~\ref{alg:coalgebraic-dijkstra} computes the greatest fixed point $\gfp \Bellman{G}{\sigma}{\gamma}$.
    For the discount rate, we have $1 > 0 + 0 \times 1 = \sigma(0, 0, 1)$.
    Hence, there is a weighted $G$-graph for which Algo.~\ref{alg:coalgebraic-dijkstra} fails.
\end{example}

\subsection{Transition modalities condition for $GX = V \times X^t$}\label{sec:VxXt}
Let $V$ be a set, $t \in \NN$, and $G \colon \Sets \to \Sets$ be the functor defined by $GX = V \times X^t$.
In this setting, a weighted $G$-graph $\gamma \colon X \to \BB \times \powsetfin(V \times X^t)$ can be viewed as a transition system in which each state $x$ has $t$-ary transitions to states $x_0, \dots, x_{t-1}$, each equipped with a weight $a \in V$. Accordingly, the CSPP on $(V \times (\blank)^t, \sigma)$ can be viewed as the problem of computing a shortest tree with $t$-ary branching.

The following lemma gives a characterization of transition modalities $\sigma$ for $G = V \times (\blank)^t$.
\begin{lemma}\label{lem:transition-modality-VxXt}
    A $G$-algebra $\sigma \colon V \times \Omega^t \to \Omega$ is a transition modality if and only if
    $ \sigma(a, \bigsqcap B_0, \dots, \bigsqcap B_{t-1}) = \bigsqcap_{b_0 \in B_0} \dots \bigsqcap_{b_{t-1} \in B_{t-1}} \sigma(a, b_0, \dots, b_{t-1}) $ holds for each $a \in V$ and $B_0, \dots, B_{t-1} \subseteq \Omega$.
\end{lemma}

Given a transition modality $\sigma \colon V \times \Omega^t \to \Omega$, the set $\Omega^{\sigma}$ in Def.~\ref{def:Omega-sigma} is given by
$\Omega^{\sigma}_0 = \{ \finalweight, \top_{\Omega} \}$,
$\Omega^{\sigma}_{n + 1} = \Omega^{\sigma}_n \cup \{ \sigma(a, b_0, \dots, b_{t-1}) \mid a \in V \text{ and } b_0, \dots, b_{t-1} \in \Omega^{\sigma}_n \}$,
and
$\Omega^{\sigma} = \bigcup_{n = 0}^{\infty} \Omega^{\sigma}_n$.

The notion of run paths and their $\sigma$-values, introduced in \S\ref{sec:VxX}, are generalized to the notion of \emph{run trees} and their \emph{$\sigma$-values}, as follows.
\begin{definition}[run tree and $\sigma$-value]
    Given a weighted $G$-graph $\gamma \colon X \to \BB \times \powsetfin(V \times X^t)$, a \emph{run tree} of $\gamma$ is a pair $T = (T, f)$, consisting of
    \begin{itemize}
        \item a prefix closed subset $T$ of the set $\{ 0, 1, \dots, t-1\}^*$ of strings, and
        \item a function $f \colon T \to X \times (V + \{ \star \})$ such that
        \begin{itemize}
            \item if $\tau \in T$ is a leaf, then $f(\tau) = (x_{\tau}, \star)$ and $\pi_0(\gamma(x_{\tau})) = \btt$, and
            \item if $\tau \in T$ is an internal node, then
            $ (a_{\tau}, x_{\tau 0}, \dots, x_{\tau (t-1)}) \in \pi_1(\gamma(x_{\tau})) $
            holds where $f(\tau) = (x_\tau, a_{\tau})$ and $x_{\tau j} = \pi_0(f(\tau j))$, for each $j \in \{ 0, \dots, t-1\}$.
        \end{itemize}
    \end{itemize}
    The \emph{height} $\height(T)$ of a run tree $T = (T,f)$ is the maximum length of strings in $T$.
    We write $\runtree_m(\gamma,x)$ for the set of run trees $T$
    such that $\pi_0(f(\epsilon)) = x$ and $\height(T) \le m$.
    We define the \emph{$\sigma$-value} of a run tree $T$ recursively by 
    $\sigma(\{\epsilon\}, f) = \finalweight$ and
    $\sigma(T, f) = \sigma(\pi_1(f(\epsilon)), \sigma(T_0, f_0), \dots, \sigma(T_{t-1}, f_{t-1}))$ where $T_i = \{ \tau \mid i\tau \in T\}$ and $f_i(\tau) = f(i\tau)$ for all $\tau \in T_i$.
\end{definition}

The following lemma relates the Bellman operator $\Bellman{G}{\sigma}{\gamma}$ to $\sigma$-values of run trees.
\begin{lemma}\label{lem:runtree-Phi}
    Let $\sigma \colon V \times \Omega^t \to \Omega$ be a transition modality,
    $X$ be a finite set, and
    $\gamma \colon X \to \BB \times \powsetfin(V \times X^t)$ be a coalgebra.
    For each $n \in \NN$ and $x \in X$, 
    $(\Bellman{G}{\sigma}{\gamma})^{n + 1}(\top_{[X,\Omega]})(x)
    =
    \bigsqcap_{T \in \runtree_n(\gamma,x)} \sigma(T)$ holds.
    Furthermore, $\gfp \Bellman{G}{\sigma}{\gamma}(x) = \bigsqcap_{n=0}^{\infty} \bigsqcap_{T \in \runtree_n(\gamma,x)} \sigma(T)$ holds.
\end{lemma}

The necessary and sufficient condition on $\sigma$ for the correctness of Algo.~\ref{alg:coalgebraic-dijkstra} is given as follows.
\begin{proposition}\label{theorem:VxXt}
    For a transition modality $\sigma \colon V \times \Omega^t \to \Omega$, the following statements are equivalent:
    \begin{enumerate}
        \item\label{item:expansion-VxXt} $b_j \sqsubseteq \sigma(a, b)$, for each $a \in V$, $b = (b_0, \dots, b_{t-1}) \in (\Omega^{\sigma})^t$ and $j \in \{ 0, \dots, t-1\}$.
        \item\label{item:dijkstra-crrect-VxXt} For every weighted $G$-graph $\gamma \colon X \to \BB \times \powsetfin(V \times X^t)$ on a finite set $X$, $\textsc{CoalgDijkstra}_{(G,\sigma)}$ is correct.
    \end{enumerate}
\end{proposition}
%
\begin{proof}[Proof sketch]
The proof of (\ref{item:expansion-VxXt} $\implies$ \ref{item:dijkstra-crrect-VxXt}) in Prop.~\ref{theorem:VxXt} is the tree analog of the proof of Prop.~\ref{theorem:VxX}.
The main part of the argument is to prove, by induction on $n$, that
$d_{n}(y) \sqsubseteq \gfp\Bellman{G}{\sigma}{\gamma}(y)$ for every $y \in Y_{n}$ 

Suppose that there exists $y \in Y_{n+1}$ such that
$d_{n+1}(y) \sqsupset \gfp\Bellman{G}{\sigma}{\gamma}(y)$.
By Lemma~\ref{lem:runtree-Phi} there exists a run tree $T \in \runtree_n(\gamma, y)$ such that $\sigma(T) \sqsubset d_{n+1}(y)$.
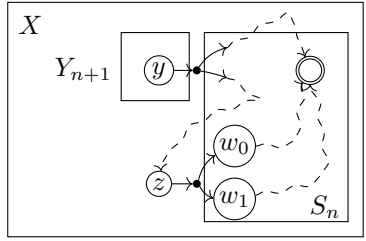
\[
\begin{tikzpicture}{scale=0.8, transform shape}
    \draw (-1.5, 3.4) rectangle (3.2, 0.3);
    \draw (0, 3) rectangle (0.9, 2.1);
    \draw (1.1, 3) rectangle (3, 0.5);
    \node at (-1.2, 3.1) {$X$};
    \node at (-0.5, 2.5) {$Y_{n+1}$};
    \node at (2.7, 0.7) {$S_n$};
    \node[state] (y) at (0.5, 2.5) {$y$};
    \node[dot, right=0.5cm of y.center, anchor=center] (ya) {};
    \node[state] (z) at (0.5, 1) {$z$};
    \node[dot, right=0.5cm of z.center, anchor=center] (za) {};
    \node[state] (z0) at (1.5, 1.5) {$w_0$};
    \node[state] (z1) at (1.5, 0.8) {$w_1$};
    \node[state, accepting] (target) at (2.5, 2.5) {\phantom{$a$}};
    \path[->] (y) edge (ya);
    \path[->] (ya) edge[bend left=20]  (1.4, 2.8)
                   edge[bend left=5] (1.4, 2.4);
    \path[->, dashed, looseness=2] (1.4, 2.8) edge[out=30, in=110, decorate, segment length=4mm, decoration={snake, amplitude=0.3mm}] (target);
    \path[->, dashed, looseness=2] (1.4, 2.4) edge[out=-30, in=80, decorate, segment length=4mm, decoration={snake, amplitude=0.3mm}] (z);
    \path[->] (z) edge (za);
    \path[->, looseness=1] (za) edge[bend left=20] (z0)
                                edge[bend right=20] (z1);
    \path[->, dashed, looseness=2] (z0) edge[out=0, in=-120, decorate, segment length=4mm, decoration={snake, amplitude=0.3mm}] (target);
    \path[->, dashed, looseness=2] (z1) edge[out=10, in=-90, decorate, segment length=4mm, decoration={snake, amplitude=0.3mm}] (target);
\end{tikzpicture}
\]

From $T$, we extract a ``subtree'' $T' \in \runtree_i(\gamma, z)$ $(i \le n)$ rooted at a vertex $z \in X \setminus S_n$ such that every vertex of $T'$ except $z$ belongs to $S_n$.
By the condition~\ref{item:expansion-VxXt}, we have $\sigma(T') \sqsubseteq \sigma(T)$.
By an argument analogous to the proof of Prop.~\ref{theorem:VxX}, we obtain
\[\begin{aligned}
  d_{n+1}(y) 
  & \;\sqsubseteq\; d_{i+1}(z)
  \sqsubseteq \Bellman{G}{\sigma}{\gamma}(d_i)(z)
  \sqsubseteq \sigma(a, d_i(w_0), \dots, d_i(w_{t-1}))
  \\
  & \sqsubseteq \sigma(a, \gfp\Bellman{G}{\sigma}{\gamma}(w_0), \dots, \gfp\Bellman{G}{\sigma}{\gamma}(w_{t-1}))
  \quad \text{by the induction hypothesis for $i < n + 1$}
  \\
  &\sqsubseteq \sigma(T')
  \sqsubseteq \sigma(T) \sqsubset d_{n+1}(y),
\end{aligned}\]
where $w_0, \dots, w_{t-1}$ are the successors of $z$ in $T'$ and $a$ is the value of the edge from $z$ to its successors in $T'$.
This is a contradiction, and hence we have $d_{n+1}(y) \sqsubseteq \gfp\Bellman{G}{\sigma}{\gamma}(y)$ for every $y \in Y_{n+1}$.
\end{proof}

\begin{example}[shortest binary tree, continued from Ex.~\ref{example:shortest-bin-tree}]
    We have $\Omega^{\sigma} = \RRinfpos$.
    For every $a \in \RRpos$ and $b_0, b_1 \in \Omega^{\sigma}$, we have
    $b_i \sqsubseteq a + b_0 + b_1 = \sigma(a, b_0, b_1)$ for $i \in \{0,1\}$.
    Hence, the condition~\ref{item:expansion-VxXt} holds, and Algo.~\ref{alg:coalgebraic-dijkstra} returns $\gfp \Bellman{G}{\sigma}{\gamma}$.
\end{example}

\begin{example}[binary reachability game, continued from Ex.~\ref{example:bin-reachability-game}]
    We have $\Omega^{\sigma} = \{ 0, \infty \}$ and $b_i \le b_0 \sqcup b_1 = \sigma(b_0,b_1)$ for all $b_0, b_1 \in \Omega^{\sigma}$ and $i = 0, 1$.
    Hence, the condition~\ref{item:expansion-VxXt} holds, and Algo.~\ref{alg:coalgebraic-dijkstra} returns $\gfp \Bellman{G}{\sigma}{\gamma}$.
\end{example}

\subsection{The condition of transition modalities for $\distr$ and $\powsetfinne(V \times \blank)$}\label{sec:DP}
By an argument similar to the case $GX = V \times X^t$, we obtain the following results for the functors $GX = \distr X$ and $GX = \powsetfinne(V \times X)$.
Instead of providing individual proofs for these propositions, we generalize them and prove the unified version (Thm.~\ref{thm:correctness-G}).
\begin{proposition}\label{theorem:D}
    Let $G \colon \Sets \to \Sets$ be the distribution functor $G = \distr$ and $\sigma \colon \distr \Omega \to \Omega$ be a transition modality.
    The following statements are equivalent:
    \begin{enumerate}
        \item\label{item:expansion-D} $a \sqsubseteq \sigma(\mu)$ for each $\mu \in \distr\Omega^{\sigma}$ and $a \in \Omega^{\sigma}$ with $\mu(a) > 0$.
        \item\label{item:dijkstra-crrect-D} For every weighted $G$-graph $\gamma \colon X \to \BB \times \powsetfin(\distr X)$ on a finite set $X$, $\textsc{CoalgDijkstra}_{(G,\sigma)}$ is correct.
    \end{enumerate}
\end{proposition}

\begin{proposition}\label{theorem:P}
    Let $G \colon \Sets \to \Sets$ be a functor defined by $GX = \powsetfinne(V \times X)$ for some set $V$ and $\sigma \colon \powsetfinne(V \times \Omega) \to \Omega$ a transition modality.
    The following statements are equivalent:
    \begin{enumerate}
        \item\label{item:expansion-P} $b \sqsubseteq \sigma(B)$, for every $B \in \powsetfinne(V \times \Omega^{\sigma})$ and $b \in \Omega^{\sigma}$ such that $(a, b) \in B$, for some $a \in V$.
        \item\label{item:dijkstra-crrect-P} For every weighted $G$-graph $\gamma \colon X \to \BB \times \powsetfin(\powsetfinne(V \times X))$ on a finite set $X$, $\textsc{CoalgDijkstra}_{(G,\sigma)}$ is correct.
    \end{enumerate}
\end{proposition}

\begin{example}[maximum probabilistic reachability, continued from Ex.~\ref{example:max-prob-reachability}.]
Algo.~\ref{alg:coalgebraic-dijkstra}  cannot solve the maximum probabilistic reachability problem.
We show this by constructing a distribution $\mu \in \distr \Omega^{\sigma}$ and an element $a \in \Omega^{\sigma}$ such that $\mu(a) > 0$ and $a \sqsupset \sigma(\mu)$, which is the negation of the condition~\ref{item:expansion-D} in Prop.~\ref{theorem:D}.
We have $\Omega^{\sigma} = [0,1]$.
Let $a = 0$ and let $\mu \in \distr [0,1]$ be the distribution defined by $\mu(0) = 1/2$ and $\mu(1) = 1/2$.
Then $\mu(a) = 1/2 > 0$ and $\sigma(\mu) = 0 \times \mu(0) + 1 \times \mu(1) = 1/2 \sqsubset 0 = a$.
From this inequality, we define a weighted $G$-graph $\gamma \colon \{0, 1\} \to \BB \times \powsetfin(\distr\{0, 1\})$ by $\gamma(0) = (\btt, \emptyset)$ and $\gamma(1) = (\bff, {\lambda x \in \{0, 1\}. 1/2})$:
\begin{tikzpicture}[node distance=1.5cm, auto, baseline=0cm]
\node  {$0$};
\node[state] (1) [left=of 0.center, anchor=center] {$1$};
\node[dot, right=0.7cm of 1.center, anchor=center] (1a) {};
\path (1) edge (1a);
\path[->] (1a) edge[bend right=40] node[above] {\tiny $1/2$} (1);
\path[->] (1a) edge[bend left=30] node[above] {\tiny $1/2$} (0);
\end{tikzpicture}.
For this weighted $G$-graph $\gamma$, Algo.~\ref{alg:coalgebraic-dijkstra} returns $d \colon \{0, 1\} \to [0,1]$ with $d(0) = 1$ and $d(1) = 1/2$.
However, the greatest fixed point $\gfp \Bellman{G}{\sigma}{\gamma}$ satisfies $\gfp\Bellman{G}{\sigma}{\gamma}(0) = 1$ and $\gfp \Phi(1) = 1$.
\end{example}
\begin{example}[the reachability game, continued from Ex.~\ref{example:reachability-game-dynamic-game}]
    We have $\Omega^{\sigma} = \{0, \infty \}$ and $b \le \max_{b' \in B}(b') = \sigma(B)$ for every $B \in \powsetfinne(\Omega^{\sigma})$ and $b \in \Omega^{\sigma}$ such that $b \in B$.
    Hence, the condition~\ref{item:expansion-P} holds, and Algo.~\ref{alg:coalgebraic-dijkstra} returns the greatest fixed point $\gfp \Bellman{G}{\sigma}{\gamma}$.
\end{example}


\section{Proofs for \S\ref{sec:coalg-dijkstra-algorithm} and for \S\ref{sec:conditions-specific-functors}}

\begin{proof}[Proof for Prop.~\ref{prop:termination}]
    In each iteration of the main loop, $Y$ is not the empty set.
    Therefore, the size of $S$ strictly increases in each iteration.
    Since $X$ is a finite set and $S \subseteq X$, $S$ reaches $X$ in a finite number of iterations.
\end{proof}

\begin{proof}[Proof for Lem.~\ref{lem:transition-modality-OmegaxX}]
    Suppose that $\sigma$ is a transition modality.
    By the definition of transition modalities (Definition~\ref{def:transtionModality}), we have
    \[
    \sigma\left(a, \left(\bigsqcap_{d \in D} d\right)(x)\right)
    =
    \bigsqcap_{d \in D} \sigma(a, d(x))
    \]
    for each $a \in V$, finite set $X$ and $D \subseteq [X, \Omega]$.
    Given arbitrary $a \in V$ and $B \subseteq \Omega$,
    we define $\tilde{B} = \{ \tilde{b} \in [\{ \star \}, \Omega] \mid b \in B \}$ where $\tilde{b}(\star) = b$ and we have
    \[
    \sigma\left(a, \bigsqcap_{b \in B} b\right)
    =
    \sigma\left(a, \left(\bigsqcap_{\tilde{b} \in \tilde{B}} \tilde{b}\right)(\star)\right)
    =
    \bigsqcap_{\tilde{b} \in \tilde{B}} \sigma(a, \tilde{b}(\star))
    =
    \bigsqcap_{b \in B} \sigma(a, b).
    \]

    Conversely, suppose that $\sigma(a, \bigsqcap B) = \bigsqcap_{b \in B} \sigma(a,b)$ for each $a \in V$ and $B \subseteq \Omega$ holds for any $a \in V$ and $B \subseteq \Omega$.
    For any finite set $X$, $D \subseteq [X, \Omega]$, $x \in X$ and $a \in V$,
    we define $B = \{ d(x) \mid d \in D \}$ and we have
    \[
    \begin{aligned}
        \sigma\left(a, \left(\bigsqcap_{d \in D} d\right)(x)\right)
        & = \sigma\left(a, \bigsqcap_{d \in D} d(x)\right)
        \\
        & = \sigma\left(a, \bigsqcap_{b \in B} b\right)
        \\
        & = \bigsqcap_{b \in B} \sigma(a, b)
        \\
        & = \bigsqcap_{d \in D} \sigma(a, d(x)).
    \end{aligned}
    \]
    Hence, $\sigma$ is a transition modality.
\end{proof}

\begin{proof}[Proof for Lem.~\ref{lem:runpath-Phi}]
    We prove the equation by induction on $n$.
    We write $\Phi$ for $\Bellman{G}{\sigma}{\gamma}$.
    
    If $n = 0$, we have
    \[
    \begin{aligned}
        & \Phi^1(\top_{[X, \Omega]})(x) \\
        & = \begin{cases}
            \finalweight \sqcap \bigsqcap_{(a, y) \in A} \sigma(a, \top_{[X,\Omega]}(y)) & (\gamma(x) = (\btt, A)) \\
            \bigsqcap_{(a, y) \in A} \sigma(a, \top_{[X,\Omega]}(y)) & (\gamma(x) = (\bff, A))
        \end{cases}
        \\
        & = \begin{cases}
            \finalweight \sqcap \bigsqcap_{(a, y) \in A} \top_{\Omega} & (\gamma(x) = (\btt, A)) \\
            \bigsqcap_{(a, y) \in A} \top_{\Omega} & (\gamma(x) = (\bff, A))
        \end{cases}
        & \text{by $\top = \sigma(a, \top)$}
        \\
        & = \begin{cases}
            \finalweight & (\gamma(x) = (\btt, A)) \\
            \top_{\Omega} & (\gamma(x) = (\bff, A))
        \end{cases}
        \\
        & = \bigsqcap_{p \in \runpath_0(\gamma,x)} \sigma(p).
    \end{aligned}
    \]

    For $n + 1$, we have
    \[
    \begin{aligned}
        & \Phi^{n + 2}(\top_{[X, \Omega]})(x)
        \\
        & = \begin{cases}
            \finalweight \sqcap \bigsqcap_{(a, y) \in A} \sigma(a, \Phi^{n + 1}(\top)(y))& (\gamma(x) = (\btt, A)) \\
            \bigsqcap_{(a, y) \in A} \sigma(a, \Phi^{n + 1}(\top)(y)) & (\gamma(x) = (\bff, A))
        \end{cases}
        \\
        & = \begin{cases}
            \finalweight \sqcap \bigsqcap_{(a, y) \in A} \sigma(a, \bigsqcap_{p \in \runpath_{n}(\gamma,y)} \sigma(p)) & (\gamma(x) = (\btt, A)) \\
            \bigsqcap_{(a, y) \in A} \sigma(a, \bigsqcap_{p \in \runpath_{n}(\gamma,y)} \sigma(p)) & (\gamma(x) = (\bff, A))
        \end{cases}
        & \text{by the induction hypothesis}
        \\
        & = \begin{cases}
            \finalweight \sqcap \bigsqcap_{(a, y) \in A} \bigsqcap_{p \in \runpath_{n}(\gamma,y)} \sigma(a, \sigma(p)) & (\gamma(x) = (\btt, A)) \\
            \bigsqcap_{(a, y) \in A} \bigsqcap_{p \in \runpath_{n}(\gamma,y)} \sigma(a, \sigma(p)) & (\gamma(x) = (\bff, A))
        \end{cases}
        & \text{\begin{minipage}{3.5cm}since $\sigma$ is a transition modality and Lem.~\ref{lem:transition-modality-OmegaxX}\end{minipage}}
        \\
        & = \begin{cases}
            \finalweight \sqcap \bigsqcap_{p \in \runpath_{n+1}(\gamma,x)} \sigma(p) & (\gamma(x) = (\btt, A)) \\
            \bigsqcap_{p \in \runpath_{n+1}(\gamma,x)} \sigma(p) & (\gamma(x) = (\bff, A))
        \end{cases}
        \\
        & = \bigsqcap_{p \in \runpath_{n+1}(\gamma,x)} \sigma(p).
    \end{aligned}
    \]
\end{proof}

\begin{proof}[Proof for Prop.~\ref{theorem:VxX}]
    (\ref{item:expansion} $\implies$ \ref{item:dijkstra-crrect-OmegaxX})
    By induction on $n$ and the assumption that $b \sqsubseteq \sigma(a, b)$ for every $a \in V$ and $b \in \Omega^{\sigma}$, we can show $\finalweight \sqsubseteq b$ for every $b \in \Omega^{\sigma}_n$.

    By induction on $n$, we prove that $d_n(x) = \bigsqcap_{m=0}^{\infty} \bigsqcap_{p \in \runpath_m(\gamma,x)}\sigma(p)$ holds for any $x \in S_n$ and $n = 1, 2, \dots$.

    \textbf{Base case ($n = 1$).}
    We have $S_1 = Y_1 = \{ x \in X \mid \pi_0(\gamma(x)) = \btt \}$
    and $d_1 = \lambda x . \begin{cases}
        \finalweight & (x \in S_1) \\
        \top_{\Omega} & (x \not \in S_1)
    \end{cases}$.
    We have
    \[ d_1(x)
    = \finalweight = \Phi(\top)(x)
    = \bigsqcap_{p \in \runpath_0(\gamma,y)} \sigma(p)
    = \bigsqcap_{m = 0}^{\infty} \bigsqcap_{p \in \runpath_m(\gamma,y)} \sigma(p)
    \]
    for $x \in S_1$ since $\sigma(p) \in \Omega^{\sigma}$ and $\finalweight \sqsubseteq \sigma(p)$.

    \textbf{Inductive case.} We have
    $d_{n+1}(x) = \begin{cases}
        d_n(x) & (x \in S_n) \\
        \Phi(d_n)(x) & (x \in P_{n+1} \setminus S_n) \\
        d_n(x) & (x \in X \setminus (P_{n+1} \cup S_n)
    \end{cases}$ and
    $Y_{n+1} = \{ y \in X \setminus S_n \mid d_{n + 1}(y) = \bigsqcap_{z \in X \setminus S_n} d_{n + 1}(z) \}$
    where $P_{n+1} = \{ x \in X \mid \exists (a, y) \in \pi_1(\gamma(x)). \ y \in Y_n\}$.
    By Lem.~\ref{lem:runpath-Phi}, we have $\Phi^{n + 1}(\top)(y) \sqsupseteq \bigsqcap_{m = 0}^{\infty} \bigsqcap_{p \in \runpath_m(\gamma,y)} \sigma(p)$.
    By Lem.~\ref{lem:partial-Kleene-chain}, we have
    \[ d_{n + 1}(y) \sqsupseteq \Phi^{n + 1} (\top) (y) \sqsupseteq \bigsqcap_{m = 0}^{\infty} \bigsqcap_{p \in \runpath_m(\gamma,y)} \sigma(p). \]
    
    It suffices to show $d_{n + 1}(y) \sqsubseteq \bigsqcap_{m = 0}^{\infty} \bigsqcap_{p \in \runpath_m(\gamma,y)} \sigma(p)$ for $y \in Y_{n + 1}$.
    Suppose that $d_{n + 1}(y) \sqsupset \bigsqcap_{m = 0}^{\infty} \bigsqcap_{p \in \runpath_m(\gamma,y)} \sigma(p)$ holds for some $y \in Y_{n + 1}$.
    There exist a natural number $m$ and a run path $p = ((x_l, a_l))_{l = 0}^m$ such that
    $x_m = y$ and 
    \begin{equation}\label{eq:sigma-p-<-dy}
        \sigma(p) \sqsubset d_{n + 1}(y).
    \end{equation} 
    We take the minimum index $k \in \{ 1, \dots, n \}$ such that $x_k \in X \setminus S_n$.
    Since $p$ is a run path, $(a_k, x_{k - 1}) \in \pi_1(\gamma(x_k))$.
    We take the minimum index $i \in \{1, \dots, n \}$ such that $x_{k - 1} \in S_i$.
    This $i$ exists because $x_{k-1} \in S_n$.
    By the minimality of $i$, we have $x_{k - 1} \in Y_i$.
    We have $d_{i + 1}(x_k) = \Phi(d_i)(x_k)$ because $x_k \in P_{i + 1} = \{ x \in X \mid \exists (a, z) \in \pi_1(\gamma(x)).\ z \in Y_i \}$.
    We have
    \[
    \begin{aligned}
        & d_{n + 1}(y)
        \\
        & = \bigsqcap_{z \in X \setminus S_n} d_{n + 1}(z)
        & \text{by $y \in Y_{n + 1}$}
        \\
        & \sqsubseteq d_{n + 1}(x_k)
        & \text{by $x_k \in X \setminus S_n$.}
        \\
        & \sqsubseteq d_{i + 1}(x_k)
        & \text{by $d_{n + 1} \sqsubseteq d_{i + 1}$}
        \\
        & = \Phi(d_i)(x_{k})
        \\
        & = \bigsqcap_{(a,z) \in \pi_1(\gamma(x_k))} \sigma(a, d_i(z))
        \\
        & \sqsubseteq \sigma(a_k, d_i(x_{k-1}))
        & \text{by $(a_k, x_{k-1}) \in \pi_1(\gamma(x_k))$}
        \\
        & \sqsubseteq \sigma\left(a_k,
        \bigsqcap_{m = 0}^{\infty} \bigsqcap_{q \in \runpath_m(\gamma,x_{k-1})} \sigma(q)
        \right)
        & \text{by $x_{k-1} \in S_i$ and the induction hypothesis}
        \\
        & \sqsubseteq \sigma\left(a_k, \sigma((x_l, a_l)_{l = 0}^{k-1}) \right)
        \\
        & = \sigma((x_l, a_l)_{l = 0}^{k})
        \\
        & \sqsubseteq \sigma((x_l, a_l)_{l = 0}^{k+1})
        & \text{by \ref{item:expansion}}
        \\
        & \sqsubseteq \dots \sqsubseteq \sigma((x_l, a_l)_{l = 0}^{m}) = \sigma(p)
        & \text{by \ref{item:expansion}}
        \\
        & \sqsubset d_{n + 1}(y)
        & \text{by \eqref{eq:sigma-p-<-dy}.}
    \end{aligned}
    \]
    This is a contradiction.
    Thus, we have $d_{n + 1}(y) \sqsubseteq \bigsqcap_{m=0}^{\infty}\bigsqcap_{p \in \runpath_m(\gamma,y)} \sigma(p)$.

    By the above argument, Algo.~\ref{alg:coalgebraic-dijkstra} returns $d \colon X \to \Omega$ such that
    \[ d(x) = \bigsqcap_{m=0}^{\infty}\bigsqcap_{p \in \runpath_m(\gamma,x)} \sigma(p) \]
    for $x \in S = X$.
    The map $d$ is the greatest fixed point $\gfp \Phi$ because
    \begin{equation*}
    \begin{aligned}
        & \Phi\left( \bigsqcap_{m=0}^{\infty}\bigsqcap_{p \in \runpath_m(\gamma,\blank)} \sigma(p) \right)(x) \\
        & = \begin{cases}
            \finalweight \sqcap \bigsqcap_{(a, y) \in A} \sigma(a, \bigsqcap_{m=0}^{\infty}\bigsqcap_{p \in \runpath_m(\gamma,y)} \sigma(p))& (\gamma(x) = (\btt, A)) \\
            \bigsqcap_{(a, y) \in A} \sigma(a, \bigsqcap_{m=0}^{\infty}\bigsqcap_{p \in \runpath_m(\gamma,y)} \sigma(p)) & (\gamma(x) = (\bff, A))
        \end{cases}
        \\
        & = \begin{cases}
            \finalweight \sqcap \bigsqcap_{m=0}^{\infty}\bigsqcap_{p \in \runpath_m(\gamma,y)}\bigsqcap_{(a, y) \in A} \sigma(a, \sigma(p)) & (\gamma(x) = (\btt, A)) \\
            \bigsqcap_{m=0}^{\infty}\bigsqcap_{p \in \runpath_m(\gamma,y)}\bigsqcap_{(a, y) \in A} \sigma(a, \sigma(p)) & (\gamma(x) = (\bff, A))
        \end{cases}
        \\
        & = \bigsqcap_{m=0}^{\infty}\bigsqcap_{p \in \runpath_m(\gamma,x)} \sigma(p).
    \end{aligned}
    \end{equation*}

    (\ref{item:dijkstra-crrect-OmegaxX} $\implies$ \ref{item:expansion})
    Suppose that $b \sqsupset \sigma(a, b)$ holds for some $b \in \Omega^{\sigma}$ and $a \in V$.
    Since $b \in \Omega^{\sigma}_m$ for some $m$, there is a sequence $a_1, \dots, a_n \in V$ for some $n \le m$ such that
    \[
    b = \sigma(a_n, \sigma(a_{n-1}, \dots \sigma(a_1, \zeta) \dots)),
    \quad
    \zeta \in \Omega^{\sigma}_0 = \{ \finalweight, \top_\Omega \}
    \]
    and $\sigma(a_i, \dots \sigma(a_1, \zeta) \dots) \in \Omega^{\sigma}_i$ for each $i = \{0, 1, \dots, n\}$.
    
    If $\zeta = \top_\Omega$, we have $\sigma(a_1, \top) \sqsubset \top$ by the definition of $\Omega^{\sigma}_1$.
    We define a weighted $(V \times (\blank))$-graph $\gamma \colon \{ \star \} \to \BB \times \powsetfin(V \times \{ \star \})$ by $\gamma(\star) = (\bff, \{ (a_1, \star) \})$.
    Algo.~\ref{alg:coalgebraic-dijkstra} returns $\top_{[X,\Omega]}$.
    However, $\top_{[X,\Omega]}$ is not the greatest fixed point because
    we have $\gfp\Phi(\star) \sqsubseteq  \Phi(\top)(\star) = \sigma(a_1, \top_{\Omega}) \sqsubset \top_{\Omega} = \top_{[X,\Omega]}(\star)$.
    
    If $\zeta = \finalweight$, we define $X = \{ 0, 1, \dots, n\}$ and a weighted $(V\times(\blank))$-graph $\gamma \colon X \to \BB \times \powsetfin(V \times X)$ as follows:
    \begin{equation*}
        \begin{aligned}
            \gamma(0) & = (\btt, \emptyset) \\
            \gamma(i + 1) & = (\bff, \{ (a_{i + 1}, i) \}) \quad \text{for $i = 0, 1, \dots, n-2$} \\
            \gamma(n) & = (\bff, \{ (a_n, n - 1), (a, n) \}).
        \end{aligned}
    \end{equation*}
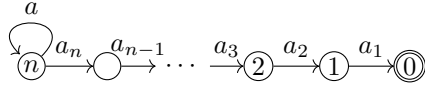
    \[
    \begin{tikzpicture}[node distance=1cm, auto, baseline=-1cm]
        \node[state,accepting](0) {$0$};
        \node[state] (1) [left=of 0.center, anchor=center] {$1$};
        \node[state] (2) [left=of 1.center, anchor=center] {$2$};
        \node (dots) [left=of 2.center, anchor=center] {$\cdots$};
        \node[state] (n_1) [left=of dots.center, anchor=center] {\phantom{$n$}};
        \node[state] (n) [left=of n_1.center, anchor=center] {$n$};
        \path[->] (1) edge node {$a_1$} (0);
        \path[->] (2) edge node {$a_2$} (1);
        \path[->] (dots) edge node {$a_3$} (2);
        \path[->] (n_1) edge node {$a_{n-1}$} (dots);
        \path[->] (n) edge node {$a_{n}$} (n_1);
        \path[->] (n) edge[loop] node[above] {$a$} ();
    \end{tikzpicture}
    \]
    Algo.~\ref{alg:coalgebraic-dijkstra} returns $d \colon X \to \Omega$ such that
    $d(i) = \sigma(a_i, \sigma(a_{i-1}, \dots, \sigma(a_1, \finalweight) \dots))$ for $i = 0, 1, \dots, n$.
    However, $d$ is not the greatest fixed point $\gfp \Phi$ because
    \[
    \begin{aligned}
        \Phi(d)(n)
        & = \bigsqcap_{(a',y) \in \pi_1(\gamma(n))} \sigma(a',d(y)) \\
        & = \sigma(a_n, d(n-1)) \sqcap \sigma(a, d(n)) \\
        & = b \sqcap \sigma(a, b) \\
        & \sqsubset b \\
        & = d(n).
    \end{aligned}
    \]
\end{proof}

\begin{proof}[Proof for Lem.~\ref{lem:transition-modality-VxXt}]
    Suppose that $\sigma$ is a transition modality.
    We define $X = \{ 0, \dots, t-1\}$
    and maps $d_{b_0, \dots, b_{t-1}} \colon X \to \Omega$ by $d_{b_0, \dots, b_{t-1}} (j) = b_j$ for $b_0, \dots, b_{t-1} \in \Omega$.
    For any $a \in V$ and $B_0, \dots, B_{t-1} \subseteq \Omega$,
    we define $D = \{ d_{b_0, \dots, b_{t-1}} \mid b_j \in B_j \text{ for each } j \in \{ 0, \dots, t-1 \}\}$.
    We have $(\bigsqcap D)(j) = \bigsqcap B_j$ and
    \[
    \begin{aligned}
        \sigma(a, \bigsqcap B_0, \dots, \bigsqcap B_{t-1})
        & = \sigma(a, (\bigsqcap D)(0), \dots, (\bigsqcap D)(t-1)) \\
        & = \sigma(G(\bigsqcap D)(a, 0, \dots, t-1)) \\
        & = \bigsqcap_{d \in D} \sigma(Gd(a, 0, \dots, t-1))\\
        & = \bigsqcap_{b_0 \in B_0} \dots \bigsqcap_{b_{t-1} \in B_{t-1}} \sigma(Gd_{b_0, \dots, b_{t-1}}(a, 0, \dots, t-1)) \\
        & = \bigsqcap_{b_0 \in B_0} \dots \bigsqcap_{b_{t-1} \in B_{t-1}} \sigma(a, b_0, \dots, b_{t-1}) \\
    \end{aligned}
    \]

    Conversely, suppose that
    \[ \sigma(a, \bigsqcap B_0, \dots, \bigsqcap B_{t-1}) = \bigsqcap_{b_0 \in B_0} \dots \bigsqcap_{b_{t-1} \in B_{t-1}} \sigma(a, b_0, \dots, b_{t-1}) \]
    holds for any $a \in V$ and $B_0, \dots, B_{t-1} \subseteq \Omega$.
    For any finite set $X$, a subset $D \subseteq [X, \Omega]$ and $(a, x_0, \dots, x_{t-1}) \in V \times X^t$, we have
    \[
    \begin{aligned}
        \sigma(G(\bigsqcap D)(a, x_0, \dots, x_{t-1}))
        & = \sigma(a, \bigsqcap_{d \in D} d(x_0), \dots, \bigsqcap_{d \in D} d(x_{t-1})) \\
        & = \bigsqcap_{d \in D} \sigma(a, d(x_0), \dots, d(x_{t-1})) \\
        & = \bigsqcap_{d \in D} \sigma(Gd(a, x_0, \dots, x_{t-1})) \\
        & = \left( \bigsqcap_{d \in D} \sigma \circ Gd \right) (a, x_0, \dots, x_{t-1}).
    \end{aligned}
    \]
    Hence, we have $\sigma \circ G(\bigsqcap D) =\bigsqcap_{d \in D} \sigma \circ Gd $ and $\sigma$ is a transition modality.
\end{proof}

\begin{proof}[Proof for Prop.~\ref{theorem:VxXt}]
    (\ref{item:expansion-VxXt} $\implies$ \ref{item:dijkstra-crrect-VxXt})
    By induction on $n$ and \ref{item:dijkstra-crrect-VxXt}, we can show $\finalweight \sqsubseteq b$ for every $b \in \Omega^{\sigma}_n$.

    By induction on $n$, we prove that
    \[ d_n(x) = \bigsqcap_{m=0}^{\infty} \bigsqcap_{T \in \runtree_m(\gamma,x)} \sigma(T)\]
    holds for every $x \in S_n$ and $n = 1, 2, \dots$.

    \textbf{Base case ($n = 1$).}
    We have
    \[ d_1(x) = \finalweight
    = \bigsqcap_{T \in \runtree_0(\gamma,x)} \sigma(T)
    = \bigsqcap_{m=0}^{\infty} \bigsqcap_{T \in \runtree_m(\gamma,x)} \sigma(T)\]
    for every $x \in S_1$ since $\finalweight \sqsubseteq \sigma(T)$ holds for every run tree $T$.
    
    \textbf{Inductive step.}
    For $n$, we have
    $d_{n+1}(x) = \begin{cases}
        d_n(x) & (x \in S_n) \\
        \Phi(d_n)(x) & (x \in P_{n+1} \setminus S_n) \\
        d_n(x) & (x \in X \setminus (P_{n+1} \cup S_n))
    \end{cases}$ and
    $Y_{n+1} = \{ y \in X \setminus S_n \mid d_{n + 1}(y) = \bigsqcap_{z \in X \setminus S_n} d_{n + 1}(z) \}$
    where $P_{n+1} = \{ x \in X \mid \exists (a, y) \in \pi_1(\gamma(x)). \ y \in Y_n\}$.
    By Lem.~\ref{lem:runtree-Phi}, we have $\Phi^{n + 1}(\top)(y) \sqsupseteq \bigsqcap_{m = 0}^{\infty} \bigsqcap_{T \in \runtree_m(\gamma,y)} \sigma(T)$ for $y \in X$.
    By Lem.~\ref{lem:partial-Kleene-chain}, for $y \in X$, we have
    \[ d_{n + 1}(y) \sqsupseteq \Phi^{n + 1} (\top) (y) \sqsupseteq \bigsqcap_{m = 0}^{\infty} \bigsqcap_{T \in \runtree_m(\gamma,y)} \sigma(T). \]
    
    It suffices to show $d_{n + 1}(y) \sqsubseteq \bigsqcap_{m = 0}^{\infty} \bigsqcap_{T \in \runtree_m(\gamma,y)} \sigma(T)$ for $y \in Y_{n + 1}$.
    Suppose that $d_{n + 1}(y) \sqsupset \bigsqcap_{m = 0}^{\infty} \bigsqcap_{T \in \runtree_m(\gamma,y)} \sigma(T)$ holds for some $y \in Y_{n + 1}$.
    There exist a natural number $m$ and a run tree $(T, f) \in \runtree_m(\gamma,y)$ such that
    \begin{equation}
        \sigma(T,f) \sqsubset d_{n+1}(y).
    \end{equation}
    Take an internal node $\kappa \in T$ such that
    $\pi_0(f(\kappa)) \in X \setminus S_n$ and
    $\pi_0(f(\kappa \tau)) \in S_n$ for any string $\tau$ such that $\kappa \tau \in T$.
    For each $j \in \{0, \dots, t-1\}$, we take the minimum index $i_j \in \{ 1, \dots, n\}$ such that $x_{\kappa j} \in S_{i_j}$.
    These $i_j$ exist because $x_{\kappa j} \in S_n$.
    For each $j$, by the minimality of $i_j$, we have $x_{\kappa j} \in Y_{i_j}$.
    For each $j$, we have $d_{i_j + 1}(x_{\kappa}) = \Phi(d_{i_j})(x_k)$ because $x_{\kappa} \in P_{i_j + 1} = \{ x \in X \mid \exists (a, (z_{j'})_{j'=0}^{t-1}) \in \pi_1(\gamma(x)).\ \exists j'. \ z_{j'} \in Y_{i_j}\}$.
    We define $i = \max_{j \in \{ 0, \dots, t-1\}} i_j$.
    We have $i_j \le i$ and $x_{\kappa j} \in S_{i_j} \subseteq S_i$ for each $j$.
    Hence, by the induction hypothesis, for each $j$,
    $d_i(x_{\kappa j}) = \bigsqcap_{m=0}^{\infty} \bigsqcap_{T' \in \runtree_m(\gamma, x_{\kappa j})} \sigma(T')$ holds.
    We have
    \begin{equation*}
    \begin{aligned}
        & d_{n + 1}(y)
        \\
        & = \bigsqcap_{z \in X \setminus S_n} d_{n + 1}(z)
        & \text{by $y \in Y_{n + 1}$}
        \\
        & \sqsubseteq d_{n + 1}(x_{\kappa})
        & \text{by $x_{\kappa} \in X \setminus S_n$}
        \\
        & \sqsubseteq d_{i + 1}(x_{\kappa})
        & \text{by $d_{n + 1} \sqsubseteq d_{i + 1}$}
        \\
        & = \Phi(d_i)(x_k)
        \\
        & = \bigsqcap_{\left(a, (z_j)_{j=0}^{t-1}\right) \in \pi_1(\gamma(x_{\kappa}))} \sigma(a, (d_i(z_j))_{j=0}^{t-1})
        \\
        & \sqsubseteq \sigma(a_{\kappa}, (d_i(x_{\kappa j}))_{j=0}^{t-1})
        & \text{by $(a_{\kappa}, (x_{\kappa j})_{j=0}^{t-1}) \in \pi_1(\gamma(x_{\kappa}))$}
        \\
        & = \sigma\left(a_{\kappa}, \left(\bigsqcap_{m=0}^{\infty} \bigsqcap_{T' \in \runtree_m(\gamma, x_{\kappa j})} \sigma(T')\right)_{j=0}^{t-1}\right)
        \\
        & \sqsubseteq \sigma(a_{\kappa}, (\sigma(T_{\kappa j}))_{j=0}^{t-1})
        & \text{by Lem.~\ref{lem:transition-modality-VxXt}}
        \\
        & = \sigma(T_{\kappa})
        \\
        & \sqsubseteq \sigma(T)
        & \text{by applying \ref{item:expansion-VxXt} recursively}
        \\
        & \sqsubset d_{n+1}(y).
    \end{aligned}
    \end{equation*}
    This is a contradiction.
    Thus, we have $d_{n + 1}(y) \sqsubseteq \bigsqcap_{m=0}^{\infty}\bigsqcap_{T \in \runtree_m(\gamma,y)} \sigma(T)$ for any $y \in Y_{n + 1}$.

    By the above argument, Algo.~\ref{alg:coalgebraic-dijkstra} returns $d \colon X \to \Omega$ such that
    \[ d(x) = \bigsqcap_{m=0}^{\infty}\bigsqcap_{T \in \runtree_m(\gamma,x)} \sigma(T) \]
    for $x \in S = X$.
    We can show that the map $d$ is the greatest fixed point $\gfp \Phi$.

    (\ref{item:dijkstra-crrect-VxXt} $\implies$ \ref{item:expansion-VxXt})
    Suppose that
    $b_j \sqsupset \sigma(a, b)$ holds for some $b = (b_0, \dots, b_{t-1}) \in (\Omega^{\sigma})^t$, $a \in V$ and $j \in \{ 0, \dots, t -1\}$.

    For a pair $(X,g)$ of a prefix closed subset $X \subseteq \{ 0, \dots, t-1\}^*$ and a map $g \colon X \to V + \Omega^{\sigma}_0$ such that
    \begin{itemize}
        \item for each leaf $\tau \in X$, $g(\tau) \in \Omega^{\sigma}_0 = \{ \finalweight, \top_{\Omega}\}$ and
        \item for each internal node $\tau \in X$, $g(\tau) \in V$ and $\tau i \in T$ for $i \in \{ 0, \dots, t - 1 \}$,
    \end{itemize}
    we define $\sigma(X,g) \in \Omega^{\sigma}$ by the following recursion:
    \begin{itemize}
        \item $\sigma(\{ \epsilon\}, g) = g(\epsilon) \in \Omega^{\sigma}_0$.
        \item $\sigma(X,g) = \sigma(g(\epsilon), \sigma(X_0, g_0), \dots, \sigma(X_{t-1}, g_{t-1}))$ where $X_{i} = \{ \tau \mid i \tau \in T \}$ and $g_{i}(\tau) = g(i\tau)$ for $i \in \{ 0, \dots, t-1\}$.
    \end{itemize}
    By the definition of $\Omega^{\sigma}$, there is a pair $(X,g)$ such that
    \begin{itemize}
        \item $(X,g)$ satisfies the above condition,
        \item $\sigma(X,g) = \sigma(a, b)$,
        \item $\sigma(X_i, g_i) = b_i$ for $i \in \{0, \dots, t-1\}$.
    \end{itemize}

    If there is a leaf $\tau i \in X$ such that $g(\tau i) = \top$,
    we have $\top \sqsupset \sigma(X_{\tau}, g_{\tau}) = \sigma(g(\tau), g(\tau 0), \dots, g(\tau (t-1)))$ by the definition of $\Omega^{\sigma}$.
    Hence, we replace $(X, g)$ by $(X_{\tau}, g_{\tau})$.
    
    We define $X' = X \setminus \{ \epsilon \}$ and a weighted graph $\gamma \colon X'  \to \BB \times \powsetfin(V \times X'^{t})$ as follows:
    \begin{itemize}
        \item For each leaf $\tau \in X'$, $\gamma(\tau) = (\btt, \emptyset)$.
        \item For each internal node $\tau \in X'$,
        \[ \gamma(\tau) = \begin{cases}
            (\bff, \{ (g(j), j 0, \dots, j (t-1) ), (a, 0, \dots, t-1)\}) & (\tau = j)\\
            (\bff, \{ (g(\tau), \tau 0, \dots, \tau (t-1) )\}) & \text{otherwise.}
        \end{cases} \]
    \end{itemize}
    
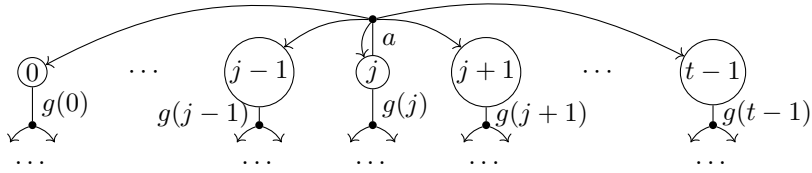
\[
    \begin{tikzpicture}[node distance=1.5cm, auto]
        \node[state, minimum size=11pt] (0) {$0$};
        \node[dot] (0a) [below= 0.7cm of 0.center, anchor=center] {};
        \node (0l) [below left = 0.5cm of 0a.center, anchor=center] {};
        \node (0r) [below right= 0.5cm of 0a.center, anchor=center] {};
        \node[below= 0.5cm of 0a.center, anchor=center] {$\dots$};
        \node (dotsa) [right=of 0.center, anchor=center] {$\cdots$};
        \node[state] (jm1) [right=of dotsa.center, anchor=center] {$j-1$};
        \node[dot] (jm1a) [below= 0.7cm of jm1.center, anchor=center] {};
        \node (jm1l) [below left = 0.5cm of jm1a.center, anchor=center] {};
        \node (jm1r) [below right= 0.5cm of jm1a.center, anchor=center] {};
        \node[below= 0.5cm of jm1a.center, anchor=center] {$\dots$};
        \node[state] (j) [right=of jm1.center, anchor=center] {$j$};
        \node[dot] (ja) [below= 0.7cm of j.center, anchor=center] {};
        \node[dot] (jb) [above= 0.7cm of j.center, anchor=center] {};
        \node (jl) [below left = 0.5cm of ja.center, anchor=center] {};
        \node (jr) [below right= 0.5cm of ja.center, anchor=center] {};
        \node[below= 0.5cm of ja.center, anchor=center] {$\dots$};
        \node[state] (jp1) [right=of j.center, anchor=center] {$j + 1$};
        \node[dot] (jp1a) [below= 0.7cm of jp1.center, anchor=center] {};
        \node (jp1l) [below left = 0.5cm of jp1a.center, anchor=center] {};
        \node (jp1r) [below right= 0.5cm of jp1a.center, anchor=center] {};
        \node[below= 0.5cm of jp1a.center, anchor=center] {$\dots$};
        \node (dotsb) [right=of jp1.center, anchor=center] {$\cdots$};
        \node[state] (tm1) [right=of dotsb.center, anchor=center] {$t-1$};
        \node[dot] (tm1a) [below= 0.7cm of tm1.center, anchor=center] {};
        \node (tm1l) [below left = 0.5cm of tm1a.center, anchor=center] {};
        \node (tm1r) [below right= 0.5cm of tm1a.center, anchor=center] {};
        \node[below= 0.5cm of tm1a.center, anchor=center] {$\dots$};
        \path (0)   edge node[right] {$g(0)$}   (0a);
        \path (jm1) edge node[left]  {$g(j-1)$} (jm1a);
        \path (j)   edge node[right] {$g(j)$}   (ja);
        \path (jp1) edge node[right] {$g(j+1)$} (jp1a);
        \path (tm1) edge node[right] {$g(t-1)$} (tm1a);
        \path[->] (0a)   edge[bend right=20] (0l);
        \path[->] (0a)   edge[bend left=20]  (0r);
        \path[->] (jm1a) edge[bend right=20] (jm1l);
        \path[->] (jm1a) edge[bend left=20]  (jm1r);
        \path[->] (ja)   edge[bend right=20] (jl);
        \path[->] (ja)   edge[bend left=20]  (jr);
        \path[->] (jp1a) edge[bend right=20] (jp1l);
        \path[->] (jp1a) edge[bend left=20]  (jp1r);
        \path[->] (tm1a) edge[bend right=20] (tm1l);
        \path[->] (tm1a) edge[bend left=20]  (tm1r);
        \path (j) edge node[right] {$a$} (jb);
        \path[->] (jb)
            edge[bend right=20] (0)
            edge[bend right=20] (jm1)
            edge[bend right=30] (j)
            edge[bend left=20] (jp1)
            edge[bend left=20] (tm1);
    \end{tikzpicture}
    \]
    
    Algo.~\ref{alg:coalgebraic-dijkstra} returns $d \colon X' \to \Omega$ such that
    $d(\tau) = \sigma(X_\tau, g_\tau)$ for $\tau \in X'$.
    However, $d$ is not the greatest fixed point $\gfp \Phi$ because
    \[
    \begin{aligned}
        \Phi(d)(j)
        & = \bigsqcap_{(a',\tau_0, \dots, \tau_{t-1}) \in \pi_1(\gamma(j))} \sigma(a',d(\tau_0), \dots, d(\tau_{t-1})) \\
        & = \sigma(g(j), d(j0), \dots, d(j(t-1))) \sqcap \sigma(a, d(0), \dots, d(t-1)) \\
        & = b_j \sqcap \sigma(a, b_0, \dots, b_{t-1}) \\
        & \sqsubset b_j \\
        & = d(j).
    \end{aligned}
    \]
\end{proof}

\section{Proofs for \S\ref{subsec:fin-supp}}
\newcommand{\catc}{\mathbb{C}}
\newcommand{\catd}{\mathbb{D}}
The following classes of functors are well-studied in the literature; see e.g.~\cite{AdamekR94,Manes2002}. \todo{Add refs}
\begin{definition}[finitary functor \cite{AdamekR94}]\label{def:finitary-functor}
    Let $\catc$ and $\catd$ be categories.
    A functor $F \colon \catc \to \catd$ is a \emph{finitary functor} if it preserves filtered colimits.
\end{definition}
\begin{definition}[taut functor \cite{Manes2002}]\label{def:taut-functor}
    Let $\catc$ and $\catd$ be categories with pullbacks.
    A functor $F \colon \catc \to \catd$ is a \emph{taut functor} if it preserves pullbacks along monomorphisms.
    That is, for any pullback square along a monomorphism $m \colon B \to D$ in $\catc$ (left below),
    its image by $F$ is also a pullback square in $\catd$ (right below):
    \[
        \begin{tikzcd}
            A
            \ar[rd, phantom, "\pullback"]
            \ar[r]
            \ar[d, >->]
            &
            B
            \ar[d, >->, "m"]
            \\
            C
            \ar[r]
            &
            D
        \end{tikzcd}
        \quad \text{in $\catc$}
        \qquad
        \implies
        \qquad
        \begin{tikzcd}
            FA
            \ar[rd, phantom, "\pullback"]
            \ar[r]
            \ar[d]
            &
            FB
            \ar[d, "Fm"]
            \\
            FC
            \ar[r]
            &
            FD
        \end{tikzcd}
        \quad \text{in $\catd$}.
    \]
\end{definition}
\begin{proof}[Proof for Lem.~\ref{lem:finitary-taut-finitely-supp}]
    The proof is based on the argument by Manes \cite{Manes1998,Manes2002}.
    Let $F \colon \Sets \to \Sets$ be a finitary and taut functor.
    The existence of a map $\theta_X \colon GX \to \powsetfin(X)$ follows from the fact that $F$ is finitary.
    The naturality of $\theta$ follows from the fact that $F$ is taut.
\end{proof}
\begin{proposition}
\label{prop:fin-supp-but-not-taut}
    There is a finitely supported but not taut functor $G \colon \Sets \to \Sets$.
\end{proposition}
\begin{proof}
    We define a functor $G \colon \Sets \to \Sets$ by
    $G\emptyset = \{ 0, 1 \}$ and $GX = \{ 0 \}$ for every non-empty set $X$.
    For morphisms, we define
    \begin{itemize}
        \item $G(\emptyset \to \emptyset) = \idmor_{\{ 0, 1 \}} \colon \{ 0, 1 \} \to \{ 0, 1 \}$,
        \item $G(\emptyset \to X) \colon \{ 0, 1 \} \to \{ 0 \}$ is the unique map for every non-empty set $X$,
        \item $G(f \colon X \to Y) = \idmor_{\{ 0 \}} \colon \{ 0 \} \to \{ 0 \}$ for every morphism $f \colon X \to Y$ between non-empty sets.
    \end{itemize}
    We define a map $\theta_X \colon GX \to \powsetfin(X)$ by $\theta_X(0) = \emptyset$ for every set $X$.
    Then, $\theta$ is a natural transformation $G \to \powsetfin$ and a finite support for $G$.

    We show that $G$ is not taut.
    The following square is a pullback square along a monomorphism in $\Sets$:
    \[
    \begin{tikzcd}
        \emptyset
        \ar[rd, phantom, "\pullback"]
        \ar[r, >->]
        \ar[d, >->]
        &
        \emptyset
        \ar[d, >->]
        \\
        \emptyset
        \ar[r, >->]
        &
        \{ 0 \}.
    \end{tikzcd}
    \]
    However, its image by $G$ is not a pullback square in $\Sets$:
    \[
    \begin{tikzcd}
        \{ 0, 1 \}
        \ar[r, "\idmor"]
        \ar[d, "\idmor"]
        &
        \{ 0, 1 \}
        \ar[d]
        \\
        \{ 0, 1 \}
        \ar[r]
        &
        \{ 0 \}.
    \end{tikzcd}
    \]
    Therefore, $G$ is not taut.
\end{proof}
\begin{proof}[Proof for Lem.~\ref{lem:fin-supp-closed}]
    Let $\supp^{G_0} \colon G_0 \to \powsetfin$ and $\supp^{G_1} \colon G_1 \to \powsetfin$ be finite supports for $G_0$ and $G_1$, respectively.
    \begin{itemize}
        \item The finite support $\supp^{G_0 + G_1} \colon G_0 + G_1 \to \powsetfin$ of $G_0 + G_1$ is defined as
        \[ \supp^{G_0 + G_1}_X(z) = \begin{cases}
            \supp^{G_0}_X(x) & (z = \iota_0(x),\ x \in G_0 X) \\
            \supp^{G_1}_X(y) & (z = \iota_1(y),\ y \in G_1 X).
        \end{cases} \]
        \item The support $\supp^{G_0 \times G_1} \colon G_0 \times G_1 \to \powsetfin$ of $G_0 \times G_1$ is defined as
        \[ \supp^{G_0 \times G_1}_X(x,y) = \supp^{G_0}_X(x) \cup \supp^{G_1}_X(y). \]
        \item The support $\supp^{G_1 \circ G_0} \colon G_1 \circ G_0 \to \powsetfin$ of $G_1 \circ G_0$ is defined as
        \[ \supp^{G_1 \circ G_0}_X(t)
        = \mu(\powsetfin(\supp^{G_0}_X)(\supp^{G_1}_{G_0 X}(t)))
        = \bigcup_{x \in \supp^{G_1}_{G_0 X}(t)} \supp^{G_0}_X(x). \]
    \end{itemize}
    We can show that these are indeed finite supports by a straightforward calculation.
\end{proof}

To prove Thm.~\ref{thm:correctness-G}, we introduce some definitions and lemmas.
Let $\Omega$ be a pointed weight domain and $G \colon \Sets \to \Sets$ be a functor with a non-empty finite support $\supp \colon G \to \powsetfin$.

\begin{lemma}\label{lem:omega-min}
    If a $G$-algebra $\sigma \colon G\Omega \to \Omega$ is expansive,
    then $\finalweight \sqsubseteq b$ holds for any $b \in \Omega^\sigma$.
\end{lemma}
\begin{proof}
    Let $b \in \Omega^{\sigma}_n$.
    We prove the statement by induction on $n$.

    \textbf{Base case ($n = 0$).}
    In this case, $b = \finalweight$ or $b = \top_{\Omega}$ hold.
    Hence, we have $\finalweight \sqsubseteq b$.

    \textbf{Inductive step.}
    We have $b \in \Omega^{\sigma}_{n + 1} = \Omega^{\sigma}_n \cup \{ \sigma(t) \mid t \in G\Omega^{\sigma}_n \}$.
    If $b \in \Omega^{\sigma}_n$, then the statement holds by the induction hypothesis.
    We assume $b = \sigma(t)$ for some $t \in G \Omega^{\sigma}$.
    The set $\{ b_0, \dots, b_{k-1} \} = \supp(t) \subseteq \Omega^{\sigma}_n$ is not empty since $\supp$ is a non-empty support.
    By the induction hypothesis, we have $\finalweight \sqsubseteq b_0$.
    We have $b_0 \sqsubseteq \sigma(t) = b$ since $\sigma$ is expansive.
    Hence, we have $\finalweight \sqsubseteq b_0 \sqsubseteq b$.
\end{proof}

\begin{lemma}\label{lem:pred-supp}
    Let $F \colon \Sets \to \Sets$ be a finitely supported functor and $\gamma \colon X \to FX$ be a coalgebra.
    If $F \colon \Sets \to \Sets$ preserves weak pullbacks, then for any subset $Y \subseteq X$ it holds that
    \[
        \predecessor(Y) = \{ x \in X \mid Y \cap \supp^F_X(\gamma(x)) \ne \emptyset\}.
    \]
\end{lemma}
\begin{proof}
    Assume $x \in \predecessor(Y)$.
    We have $\gamma(x) \not\in F(X \setminus Y)$.
    If $\supp(\gamma(x)) \subseteq X \setminus Y$, then
    we have $\gamma(x) \in F(X \setminus Y)$ since
    \[
    \begin{tikzcd}[row sep=small]
        F(\supp(\gamma(x)))
        \ar[r, "F\iota"]
        &
        F(X \setminus Y)
        \\
        \gamma(x)
        \ar[r, mapsto]
        &
        \gamma(x).
    \end{tikzcd}
    \]
    This is a contradiction.
    Hence, we have $\supp(\gamma(x)) \not \subseteq X \setminus Y$, which means $Y \cap \supp(\gamma(x)) \ne \emptyset$.

    Conversely, assume that $x \in X$ satisfies $Y \cap \supp(\gamma(x)) \ne \emptyset$.
    The following diagram is a pullback in $\Sets$.
    \[
    \begin{tikzcd}
        \supp(\gamma(x)) \setminus Y
        \ar[r, hook]
        \ar[d, hook]
        \ar[rd, phantom, "\pullback"]
        &
        \supp(\gamma(x))
        \ar[d, hook]
        \\
        X \setminus Y
        \ar[r, hook]
        &
        X
    \end{tikzcd}
    \]
    Thus, the following diagram is a weak pullback in $\Sets$ since $F$ preserves pullbacks.
    \[
    \begin{tikzcd}
        F(\supp(\gamma(x)) \setminus Y)
        \ar[r, hook]
        \ar[d, hook]
        \ar[rd, phantom, "\weakpullback"]
        &
        F(\supp(\gamma(x)))
        \ar[d, hook]
        \\
        F(X \setminus Y)
        \ar[r, hook]
        &
        F(X)
    \end{tikzcd}
    \]
    If $\gamma(x) \in F(X \setminus Y)$, then we have $\gamma(x) \in F(\supp(x) \setminus Y)$ because $F(\supp(\gamma(x)) \setminus Y)$ is a weak pullback:
    \[
    \begin{tikzcd}[column sep=0.5cm]
        \{ \star \}
        \ar[rrd, bend left=10]
        \ar[rdd, bend right=20]
        \ar[rd, dashed]
        &
        &
        \\
        &
        F(\supp(\gamma(x)) \setminus Y)
        \ar[r, hook]
        \ar[d, hook]
        \ar[rd, phantom, "\weakpullback"]
        &
        F(\supp(\gamma(x)))
        \ar[d, hook]
        \\
        &
        F(X \setminus Y)
        \ar[r, hook]
        &
        F(X)
    \end{tikzcd}
    \begin{tikzcd}[column sep=0.5cm]
        \star
        \ar[rrd, mapsto, bend left=10]
        \ar[rdd, mapsto, bend right=20]
        \ar[rd, mapsto, dashed]
        &
        &
        \\
        &
       \gamma(x)
        \ar[r, mapsto]
        \ar[d, mapsto]
        \ar[rd, phantom, "\weakpullback"]
        &
        \gamma(x)
        \ar[d, mapsto]
        \\
        &
        \gamma(x)
        \ar[r, mapsto]
        &
        \gamma(x)
    \end{tikzcd}
    \]
    We have $\supp(\gamma(x)) \setminus Y \subsetneq \supp(\gamma(x))$ and $\gamma(x) \in F(\supp(\gamma(x)) \setminus Y)$.
    This contradicts to the minimality of $\supp(\gamma(x))$.
    Hence, we have $\gamma(x) \not\in F(X \setminus Y)$, which means $x \in \predecessor(Y)$.
\end{proof}

\begin{definition}[construction tree]
    Let $\sigma \colon G \Omega \to \Omega$ be a $G$-algebra on $\Omega$.
    A \emph{construction tree} of $b \in \Omega^{\sigma}$ is a finite prefix-closed set $\tree \in \NN^*$ together with a map $\ell \colon \tree \to G\Omega^{\sigma} + \Omega^{\sigma}_0$ that satisfies the following conditions:
    \begin{itemize}
        \item For $\tau \in \NN^*$ and $m \le n$, if $\tau n \in \tree$ holds, then $\tau m \in \tree$ holds.
        \item For leaf $\tau \in \tree$, $\ell(\tau) \in \Omega^{\sigma}_0$ holds.
        \item For each internal node $\tau$, the following condition holds:
        
        $\ell(\tau) \in G\Omega^\sigma$ and
        $\{ b_{\tau n} \mid \tau n \in \tree \} = \supp_{\Omega^\sigma}(\ell(\tau))$, where $b_{\tau n} = \begin{cases}
                \ell(\tau n) & \text{if }\ell(\tau n) \in \Omega^{\sigma}_0, \\
                \sigma(\ell(\tau n)) & \text{otherwise}.
            \end{cases}$
        \item $\ell(\epsilon) = b$ or $\sigma(\ell(\epsilon)) = b$ hold, where $\epsilon$ is the empty sequence.
    \end{itemize}
\end{definition}

\begin{lemma}
    Let $\sigma \colon G \Omega \to \Omega$ be a $G$-algebra on $\Omega$.
    For any $b \in \Omega^\sigma$, there exists a construction tree $\tree$ of $b$.
\end{lemma}
\begin{proof}
    By the definition of $\Omega^{\sigma}$, we have $b \in \Omega^{\sigma}_n$ for some $n \in \NN$.
    We prove the statement by induction on $n$.

    \textbf{Base case ($n = 0$).}
    The tree defined as $\tree = (\{ \epsilon \}, \lambda \tau. b)$ is a construction tree of $b$.

    \textbf{Inductive step.}
    Suppose $b \in \Omega^{\sigma}_{n + 1} = \Omega^{\sigma}_n \cup \{ \sigma(t) \mid t \in G\Omega^{\sigma}_n \}$.
    If $b \in \Omega^{\sigma}_n$, then the statement holds immediately from the induction hypothesis.
    Hence, we assume that $b = \sigma(t)$ holds for some $t \in G\Omega^{\sigma}_n$.
    Let $\{ b_0, \dots, b_{k-1}\} = \supp(t) \subseteq \Omega^{\sigma}_n$.
    By the induction hypothesis, there is a construction tree $(\tree_j, \ell_j)$ of $b_j$ for each $j \in \{0, \dots, k-1\}$.
    We define $\tree = \{ \epsilon \} \cup {\bigcup_{j = 0}^{k-1} \{ j \tau \mid j \tau \in \tree_j \}}$ and $\ell \colon \tree \to G\Omega^{\sigma} + \Omega^{\sigma}_0$ by
    $\ell(\epsilon) = t$ and $\ell(j \tau) = \ell_j(\tau)$.
    Then, $(\tree, \ell)$ is a construction tree of $b$.
\end{proof}

\begin{lemma}\label{lem:expansive-compare-constr-tree-subtree}
    If a $G$-algebra $\sigma \colon G \Omega \to \Omega$ is expansive,
    then $b_j \sqsubseteq \sigma(\ell(\epsilon))$ holds for every construction tree $\tree = (\tree, \ell)$ and $j \in \tree$, where
    \[ b_j = \begin{cases}
        \ell(j) & \text{if $\ell(j) \in \Omega^{\sigma}_0$}, \\
        \sigma(\ell(j)) & \text{otherwise}.
    \end{cases} \]
\end{lemma}
\begin{proof}
    By the definition of construction trees, we have $b_j \in \supp(\ell(\epsilon))$ for each $j$.
    Hence, we have $b_j \sqsubseteq \sigma(\ell(\epsilon))$ by expansiveness of $\sigma$.
\end{proof}

\begin{definition}[run tree]\label{def:run-tree-generally}
    Let $\gamma \colon X \to \BB \times \powsetfin(GX)$ be a weighted $G$-graph.
    A \emph{run tree} of $\gamma$ is a pair $(\tree, \ell)$ of a prefix-closed set $\tree \subseteq \NN^*$ and a labelling function $\ell \colon \tree \to X \times (GX + \{ \star \})$ that satisfies the following conditions for each $(x_\tau, t_{\tau}) = \ell(\tau)$ with $\tau \in \tree$:
    \begin{itemize}
        \item For $\tau \in \NN^*$ and $m \le n$, if $\tau n \in \tree$ holds, then $\tau m \in \tree$ holds.
        \item If $\tau \in \tree$ is a leaf, then $t_{\tau} = \star$ and $\gamma(x_\tau) = (\btt, \emptyset)$ hold.
        \item If $\tau \in \tree$ is not a leaf, then $t_{\tau} \in GX$ and $\gamma(x_\tau) = (\bff, A)$ for some $A \subseteq GX$ such that $t_{\tau} \in A$ hold.
        Furthermore, $\{ x_{\tau n} \mid \tau n \in \tree \} = \supp_X(t_{\tau})$ holds.
    \end{itemize}
    The \emph{height} of a run tree $(\tree, \ell)$ is defined by $\height(\tree, \ell) = \max \{ |\tau| \mid \tau \in \tree \}$.
    We write $\runtree_m(\gamma, x)$ for the set of run trees $\tree = (\tree, \ell)$ of $\gamma$ such that $\height(\tree) \le m$ and $x_{\epsilon} = x$.
\end{definition}

\begin{definition}[$\sigma$-value]\label{def:sigma-value-of-run-tree}
    Let $\gamma \colon X \to \BB \times \powsetfin(GX)$ be a weighted $G$-graph,
    $\sigma \colon G\Omega \to \Omega$ be a $G$-algebra on $\Omega$, and $x \in X$.
    The $\sigma$-value $\sigma(\tree) \in \Omega^\sigma$ of a run tree $\tree \in \runtree_m(\gamma, x)$ is defined recursively as follows:
    \begin{itemize}
        \item If $\tree = (\{ \epsilon \}, \ell)$ is a leaf, then $\sigma(\tree) = \finalweight$.
        \item If $\tree$ is not a single leaf, then $\sigma(\tree) = \sigma(G\phi(t_{\epsilon}))$ where $\phi \colon \supp_X(t_{\epsilon}) \to \Omega^{\sigma}$ is defined as follows:
        Let $x_j = \pi_0(\ell(j))$ for each child $j$ of $\epsilon$ and $\tree_j = (\tree_j, \ell_j)$ be the subtree of $\tree$ rooted at $j$.
        We define $\phi(x_j) = \sigma(\tree_j)$ for each child $j$ of $\tau$.
        \[
        \begin{tikzcd}[row sep=small]
            G\supp(t_{\epsilon}) = G\{ x_0, \dots, x_{k-1}\}
            \ar[r, "G\phi"]
            &
            G \Omega^\sigma
            \ar[r, "\sigma"]
            &
            \Omega^\sigma
            \\
            t_{\epsilon}
            \ar[rr, mapsto]
            &
            &
            \sigma(G\phi(t_{\epsilon}))
            = \sigma(\tree)
        \end{tikzcd}
        \]
    \end{itemize}
\end{definition}

\begin{lemma}\label{lem:expansive-run-tree}
    Let $\sigma \colon G\Omega \to \Omega$ be an expansive transition modality,
    $X$ be a finite set,
    $\gamma \colon X \to \BB \times \powsetfin(GX)$ be a coalgebra.
    For every run tree $\tree = (\tree, \ell) \in \runtree_m(\gamma, x)$ and any child $j \in \tree$ of $\epsilon$,
    we have $\sigma(\tree_{j}) \sqsubseteq \sigma(\tree)$.
\end{lemma}
\begin{proof}
    Let $(x_\tau, t_{\tau}) = \ell(\tau)$ for each $\tau \in \tree$.
    We define $\ell' \colon \tree \to G \Omega^{\sigma} + \Omega^{\sigma}_0$ by
    \[\ell'(\tau) = \begin{cases}
        \finalweight & \text{if $\tau$ is a leaf,} \\
        G\phi_{\tau}(t_{\tau}) & \text{if $\tau$ is an internal node}
    \end{cases}\]
    where $\phi_\tau \colon \{ x_{\tau 0}, \dots, x_{\tau (k -1)}\} \to \Omega^{\sigma}$ is defined by $\phi_{\tau}(x_{\tau i}) = \sigma(\tree_{\tau j})$ for each child $\tau j$ of $\tau$.
    For each internal node $\tau$, the following diagram commutes, since $\supp$ is a natural transformation:
    \begin{equation}\label{eq:nat-supp}
        \begin{tikzcd}
            G \supp(t_{\tau})
            \ar[r, "G\phi_{\tau}"]
            \ar[d, "\supp_{\supp(t_{\tau})}"']
            &
            G \Omega^{\sigma}
            \ar[d, "\supp_{\Omega^{\sigma}}"]
            \\
            \powsetfin(\supp(t_{\tau}))
            \ar[r, "\powsetfin(\phi_{\tau})"']
            &
            \powsetfin(\Omega^{\sigma}).
        \end{tikzcd}
    \end{equation}
    Then, $(\tree, \ell')$ is a construction tree because
    \begin{itemize}
        \item for each leaf $\tau \in \tree$, $\ell'(\tau) = \finalweight \in O$ holds, and
        \item for each internal node $\tau$, we have
        \[\begin{aligned}
            \supp(\ell'(\tau))
            & = \supp(G\phi_{\tau}(t_{\tau}))
            \\
            & = \powsetfin(\phi_{\tau})(\supp(t_{\tau}))
            & \text{by \eqref{eq:nat-supp}}
            \\
            & = \powsetfin(\phi_{\tau})(\{ x_{\tau j} \mid \tau j \in \tree \})
            & \text{by Definition~\ref{def:run-tree-generally}}
            \\
            & = \{ \phi_{\tau}(x_{\tau j}) \mid \tau j \in \tree \}
            \\
            & = \{ \sigma(\tree_{\tau j}) \mid \tau j \in \tree \}
            \\
            & = \{ b_{\tau j} \mid \tau j \in \tree \}
        \end{aligned}\]
        where $b_{\tau j} = \begin{cases}
            \ell'(\tau j) = \finalweight& \text{if $\tau j$ is a leaf,} \\
            \sigma(\ell'(\tau j)) = \sigma(G\phi(t_{\tau j}))& \text{otherwise}.
        \end{cases}$
    \end{itemize}
    Therefore, by the expansiveness of $\sigma$ and Lem.~\ref{lem:expansive-compare-constr-tree-subtree}, we have
    $\sigma(\tree_{j})
    = b_{j}
    \sqsubseteq \sigma(\ell'(\epsilon))
    = \sigma(\tree)$.
\end{proof}

\begin{lemma}\label{lem:fold-one-step}
    Let $\sigma \colon G\Omega \to \Omega$ be a transition modality,
    $X$ be a finite set,
    $\gamma \colon X \to \BB \times \powsetfin(GX)$ be a weighted $G$-graph.
    For each $x \in X$ and $n \in \NN$, 
    \[
        \bigsqcap_{t \in \pi_1(\gamma(x))} \sigma\left(G\left(\bigsqcap_{\tree \in \runtree_n(\gamma, \blank)} \sigma(\tree) \right)(t)\right)
        =
        \bigsqcap_{\tree \in \runtree_{n + 1}(x)} \sigma(\tree)
    \]
\end{lemma}
\begin{proof}
    We have
    \[
    \begin{aligned}
        & \bigsqcap_{\tree \in \runtree_{n + 1}(x)} \sigma(\tree)
        \\
        & = \bigsqcap_{t \in \pi_1(\gamma(x))}
        \bigsqcap_{\tree_0 \in \runtree_{n}(\gamma, x^t_0)}
        \cdots
        \bigsqcap_{\tree_{k-1} \in \runtree_{n}(\gamma, x^t_{k_t-1})}
        \sigma(G\phi^{t}_{\tree_0, \dots, \tree_{k_t-1}}(t))
        \\
        & = \bigsqcap_{t \in \pi_1(\gamma(x))}
        \sigma\left(G\left(
            \bigsqcap_{\tree_0 \in \runtree_{n}(\gamma, x^t_0)}
            \cdots
            \bigsqcap_{\tree_{k-1} \in \runtree_{n}(\gamma, x^t_{k_t-1})}
            \phi^t_{\tree_0, \dots, \tree_{k_t-1}}
        \right)(t)\right)
        &
        \text{by Def.~\ref{def:transtionModality}}
        \\
        & = \bigsqcap_{t \in \pi_1(\gamma(x))} \sigma\left(G\left(\bigsqcap_{\tree \in \runtree_n(\gamma, \blank)} \sigma(\tree) \right)(t)\right)
    \end{aligned}
    \]
    where
    $\supp(t) = \{ x^t_0, \dots, x^t_{k_t - 1}\}$ and
    $\phi^t_{\tree_0, \dots, \tree_{k-1}} \colon \supp(t) \to \Omega^\sigma$ is defined by $\phi^t_{\tree_0, \dots, \tree_{k-1}}(x_j) = \sigma(\tree_j)$.
\end{proof}

\begin{lemma}\label{lem:Phi-run-tree}
    Let $\sigma \colon G \Omega \to \Omega$ be an expansive transition modality,
    $X$ be a finite set,
    $\gamma \colon X \to \BB \times \powsetfin(GX)$ be a coalgebra.
    For any $n \in \NN$ and $x \in X$,
    \[
    (\Bellman{G}{\sigma}{\gamma})^{n + 1}(\top_{[X, \Omega]})(x)
    =
    \bigsqcap_{\tree \in \runtree_n(\gamma, x)} \sigma(\tree).
    \]
\end{lemma}
\begin{proof}
    We prove the statement by induction on $n$.
    
    \textbf{Base case ($n = 0$).}
    For any $x \in X$, we have
    \[
        (\Bellman{G}{\sigma}{\gamma})^{1}(\top_{[X, \Omega]})(x)
        = \finalweight
        = \sigma(\{ \epsilon \}, \ell)
        = \bigsqcap_{\tree \in \runtree_0(\gamma, x)} \sigma(\tree).
    \]
    
    \textbf{Inductive step.}
    Suppose that the statement holds for $n$.
    For any $x \in X$, if $\pi_0(\gamma(x)) = \btt$ holds, then we have
    \[
        (\Bellman{G}{\sigma}{\gamma})^{n + 1 + 1}(\top_{[X, \Omega]})(x)
        = \finalweight
        = \bigsqcap_{\tree \in \runtree_{n + 1}(\gamma, x)} \sigma(\tree).
    \]
    Suppose that $\pi_0(\gamma(x)) = \bff$ holds.
    We have
    \[
    \begin{aligned}
        & (\Bellman{G}{\sigma}{\gamma})^{n + 1 + 1}(\top_{[X, \Omega]})(x)
        \\
        & = \bigsqcap_{t \in \pi_1(\gamma(x))} \sigma\left(G(\Bellman{G}{\sigma}{\gamma})^{n + 1}(\top_{[X, \Omega]})(t)\right)
        & \text{by definition of $\Bellman{G}{\sigma}{\gamma}$}
        \\
        & = \bigsqcap_{t \in \pi_1(\gamma(x))} \sigma\left(G\left( \bigsqcap_{\tree \in \runtree_n(\gamma, \blank)} \sigma(\tree)\right)(t)\right)
        & \text{by the induction hypothesis}
        \\
        & = \bigsqcap_{\tree \in \runtree_{n+1}(\gamma, x)} \sigma(\tree).
        & \text{Lem.~\ref{lem:fold-one-step}}
    \end{aligned}
    \]
\end{proof}

\begin{lemma}\label{lem:gfp-run-tree}
    Let $\sigma \colon G \Omega \to \Omega$ be an expansive transition modality,
    $X$ be a finite set,
    $\gamma \colon X \to \BB \times \powsetfin(GX)$ be a coalgebra.
    Then, the greatest fixed point $\gfp \Bellman{G}{\sigma}{\gamma}$ is given by
    \[
    \gfp \Bellman{G}{\sigma}{\gamma}(x)
    = \bigsqcap_{m = 0}^{\infty} \bigsqcap_{\tree \in \runtree_m(\gamma, x)} \sigma(\tree).
    \]
\end{lemma}
\begin{proof}
    By Lem.~\ref{lem:Phi-run-tree}, we have
    \[
        \bigsqcap_{n = 0}^{\infty} \bigsqcap_{\tree \in \runtree_n(\gamma, x)} \sigma(\tree)
        = \bigsqcap_{n = 0}^{\infty} (\Bellman{G}{\sigma}{\gamma})^{n}(\top_{[X, \Omega]})(x)
        \sqsupseteq \gfp \Bellman{G}{\sigma}{\gamma}(x).
    \]
    
    Thus, it suffices to show that
    \[ 
        \bigsqcap_{n = 0}^{\infty} \bigsqcap_{\tree \in \runtree_n(\gamma, \blank)} \sigma(\tree)
        \sqsubseteq \gfp \Bellman{G}{\sigma}{\gamma}.
    \]
    By coinduction principle, it suffices to show that
    \begin{equation}\label{eq:coinduction}
        \bigsqcap_{n = 0}^{\infty} \bigsqcap_{\tree \in \runtree_n(\gamma, \blank)} \sigma(\tree)
        \sqsubseteq
        \Bellman{G}{\sigma}{\gamma}\left(\bigsqcap_{n = 0}^{\infty} \bigsqcap_{\tree \in \runtree_n(\gamma, \blank)} \sigma(\tree)\right).
    \end{equation}
    For any $x \in X$, if $\pi_0(\gamma(x)) = \btt$ holds, then the \eqref{eq:coinduction} holds.
    Suppose that $\pi_0(\gamma(x)) = \bff$ holds.
    We have
    \[
    \begin{aligned}
        & \Bellman{G}{\sigma}{\gamma}\left(\bigsqcap_{n = 0}^{\infty} \bigsqcap_{\tree \in \runtree_n(\gamma, \blank)} \sigma(\tree)\right)(x)
        \\
        & = \bigsqcap_{t \in \pi_1(\gamma(x))} \sigma\left(G\left(\bigsqcap_{n = 0}^{\infty} \bigsqcap_{\tree \in \runtree_n(\gamma, \blank)} \sigma(\tree)\right)(t)\right)
        \\
        & = \bigsqcap_{n=0}^{\infty} \bigsqcap_{t \in \pi_1(\gamma(x))} \sigma\left(G\left( \bigsqcap_{\tree \in \runtree_n(\gamma, \blank)} \sigma(\tree)\right)(t)\right)
    \end{aligned}
    \]
    For each $t \in \pi_1(\gamma(x))$ and $x_j \in \supp_X(t) = \{ x_0, \dots, x_{k-1}\}$, we take a run tree $\tree_j = (\tree_j, \ell_j) \in \runtree_n(\gamma, x_j)$ such that
    \[ \tree_j = \bigsqcap_{\tree \in \runtree_n(\gamma, x_j)} \sigma(\tree), \]
    and define $\phi \colon \supp_X(t) \to \Omega^{\sigma}$ by $\phi(x_j) = \sigma(\tree_j)$ for each $x_j \in \supp_X(t)$.
    We define a run tree $\tree_{x,t} = (\tree, \ell)$ by
    $\tree = \{ \epsilon \} \cup \{ j \tau \mid j \in \{ 0, \dots, k - 1\},\ \tau \in \tree_j \}$,
    $\ell(\epsilon) = (x, t)$ and $\ell(j \tau) = \ell(x_\tau, t_\tau)$ for each $j \in \{ 0, \dots, k - 1\}$ and $\tau \in \tree_j$.
    Then, we have
    \[
        \sigma\left(G\left( \bigsqcap_{\tree \in \runtree_n(\gamma, \blank)} \sigma(\tree)\right)(t)\right)
        = \sigma\left(G\phi(t)\right)
        = \sigma(\tree_{x,t})
        \sqsupseteq \bigsqcap_{\tree \in \runtree_{n + 1}(\gamma, x)} \sigma(\tree).
    \]
    We have
    \[
    \begin{aligned}
        \Bellman{G}{\sigma}{\gamma}\left(\bigsqcap_{n = 0}^{\infty} \bigsqcap_{\tree \in \runtree_n(\gamma, \blank)} \sigma(\tree)\right)(x)
        & = \bigsqcap_{n=0}^{\infty} \bigsqcap_{t \in \pi_1(\gamma(x))} \sigma\left(G\left( \bigsqcap_{\tree \in \runtree_n(\gamma, \blank)} \sigma(\tree)\right)(t)\right)
        \\
        & \sqsupseteq \bigsqcap_{n=0}^{\infty} \bigsqcap_{\tree \in \runtree_{n + 1}(\gamma, x)} \sigma(\tree)
        \\
        & \sqsupseteq \bigsqcap_{n=0}^{\infty} \bigsqcap_{\tree \in \runtree_{n}(\gamma, x)} \sigma(\tree),
    \end{aligned}
    \]
    which proves \eqref{eq:coinduction}.
\end{proof}

\begin{definition}[contraction coalgebra]\label{def:contraction-coalg}
    Let $\sigma \colon G\Omega \to \Omega$ be a transition modality.
    For $t \in F \Omega^\sigma$ and a construction tree $\tree$ of $\sigma(\tree)$ with $b_0 \sqsupset \sigma(\ell(\epsilon))$
    where $b_0 = \begin{cases}
        \ell(0) & \text{if $0$ is a leaf} \\
        \sigma(\ell(0)) & \text{otherwise}
    \end{cases}$,
    we define the \emph{contraction coalgebra} $\gamma \colon X \to \BB \times \powsetfin(GX) $ on $X = \tree \setminus \{ \epsilon \}$ by
    \[
    \begin{aligned}
        \gamma( 0 )  & = (\beta_0, A_0 \cup \{ \ell(\epsilon) \}) \\
        \gamma( \tau ) & = (\beta_\tau, A_{\tau}) & \text{if $\tau \ne 0$} \\
    \end{aligned}
    \]
    where
    $\beta_\tau = \begin{cases}
        \btt & \text{if $\tau$ is a leaf and $\ell(\tau) = \finalweight$} \\
        \bff & \text{otherwise}
    \end{cases}$
    and
    $A_{\tau} = \begin{cases}
        \emptyset & \text{if $\tau$ is a leaf} \\
        \{ \ell(\tau) \} & \text{otherwise}
    \end{cases}$
    for each $\tau \in X$.
\end{definition}

\subsection{Proof of main theorem}\label{subsec:pfMainThm}
\begin{proof}[Proof for Thm.~\ref{thm:correctness-G}]
    (\ref{item:expansive-G} $\implies$ \ref{item:correctness-G})
    Suppose that $\sigma$ is expansive.
    By induction on $n$, we prove that
    \[
    d_n(y)
    = \bigsqcap_{m=0}^{\infty}\bigsqcap_{\tree \in \runtree_m(\gamma, y)} \sigma(\tree)
    \]
    for any $y \in Y_{n}$.
    For a run tree $\tree = (\tree, \ell)$ and $\tau \in \tree$, we write $x_\tau$ and $t_\tau$ for $\pi_0(\ell(\tau))$ and $\pi_1(\ell(\tau))$, respectively.

    \textbf{Base case ($n = 1$).}
    For any $y \in Y_1 = \{ x \in X \mid \pi_0(\gamma(x)) = \btt\}$, we have
    \[
        d_1(y)
        = \finalweight
        = \sigma(\{ \epsilon \}, \lambda \tau. \finalweight).
    \]
    Thus, by Lem.~\ref{lem:omega-min}, we have
    $d_1(y) = \bigsqcap_{m = 0}^{\infty} \bigsqcap_{\tree \in \runtree_m(\gamma, y)} \sigma(\tree)$ for any $y \in Y_1$.

    \textbf{Inductive case.}
    Suppose that $d_{n + 1}(y) \sqsupset \bigsqcap_{m = 0}^{\infty} \bigsqcap_{\tree \in \runtree_m(\gamma, y)} \sigma(\tree)$ holds for some $y \in Y_{n + 1}$.
    Then, there exists a run tree $\tree = (\tree, \ell) \in \runtree_m(\gamma, y)$ for some $m \in \NN$ such that 
    \begin{equation}\label{eq:assumption-G}
        \sigma(\tree) \sqsubset d_{n + 1}(y).
    \end{equation}
    Take an internal node $\kappa \in \tree$ such that $x_{\kappa} \in X \setminus S_n$ and $x_{\kappa \tau} \in S_n$ for any string such that $\kappa \tau \in \tree$.
    Let $k$ be the number of children of $\kappa$.
    For each $j \in \{ 0, \dots, k - 1\}$, we take the minimum index $i_j \in \{1, \dots, n \}$ such that $x_{\kappa j} \in S_{i_j}$ holds.
    These $i_j$ exists because $x_{\kappa j} \in S_n$.
    For each $j$, by the minimality of $i_j$, we have $x_{\kappa j} \in Y_{i_j}$.
    For each $j$, we have $d_{i_j + 1}(x_{\kappa}) = \Phi(d_{i_j})(x_{\kappa})$ because $x_{\kappa} \in P_{i_j + 1} = \predecessor(Y_{i_j})$ by Lem.~\ref{lem:pred-supp}.
    We define $i = \max \{ i_j \mid j = 0, \dots, k - 1 \}$.
    Then, we have $i_j \le i$ and $x_{\kappa j} \in S_{i_j} \subseteq S_i$ for each $j$.
    Hence, by the induction hypothesis, for each $j$, it holds that
    \[ d_i(x_{\kappa j}) = \bigsqcap_{m = 0}^{\infty} \bigsqcap_{\tree' \in \runtree_m(\gamma, x_{\kappa j})} \sigma(\tree').\]
    We define a map $\phi \colon \{ x_{\kappa 0} , \dots, x_{\kappa(k-1)} \} \to \Omega^\sigma$ by $\phi(x_{\kappa j}) = \sigma(\tree_{\kappa j})$ for each $j$.
    We have $d_i \sqsubseteq \phi$ since we have
    \[
    d_i(x_{\kappa j})
    = \bigsqcap_{m = 0}^{\infty} \bigsqcap_{\tree' \in \runtree_m(\gamma, x_{\kappa j})} \sigma(\tree')
    \sqsubseteq \sigma(\tree_{\kappa j})
    = \phi(x_{\kappa j})
    \]
    for each $j$.
    Therefore, we have
    \[
    \begin{aligned}
        & d_{n + 1}(y)
        \\
        & = \bigsqcap_{z \in X \setminus S_n} d_{n + 1}(z)
        & \text{by $y \in Y_{n + 1}$}
        \\
        & \sqsubseteq d_{n + 1}(x_{\kappa})
        & \text{by $x_{\kappa} \in X \setminus S_n$}
        \\
        & \sqsubseteq d_{i + 1}(x_{\kappa})
        & \text{by $d_{n + 1} \sqsubseteq d_{i + 1}$}
        \\
        & = \Phi(d_i)(x_{\kappa})
        \\
        & = \bigsqcap_{t \in \pi_1(\gamma(x_{\kappa}))} \sigma(G d_i(t))
        \\
        & \sqsubseteq \sigma(G d_i(t_{\kappa}))
        & \text{by $t_{\kappa} \in \pi_1(\gamma(x_{\kappa}))$}
        \\
        & \sqsubseteq \sigma(G \phi(t_{\kappa}))
        & \text{by $d_i \sqsubseteq \phi$ and monotonicity of $\sigma \circ G(\blank)$}
        \\
        & = \sigma(\tree_{\kappa})
        & \text{by Definition~\ref{def:sigma-value-of-run-tree}}
        \\
        & \sqsubseteq \sigma(\tree)
        & \text{by applying Lem.~\ref{lem:expansive-run-tree}, inductively}
        \\
        & \sqsubset d_{n + 1}(y)
        & \text{by \eqref{eq:assumption-G}.}
    \end{aligned}
    \]
    This is a contradiction.
    Thus, we have $d_{n + 1}(y) \sqsubseteq \bigsqcap_{m = 0}^{\infty} \bigsqcap_{\tree \in \runtree_m(\gamma, y)} \sigma(\tree)$ for every $y \in Y_{n+1}$.
    On the other hand, for every $y \in Y_{n+1}$, we have
    \[
    d_{n+1}(y)
    \sqsupseteq \Phi^{n+1}(\top_{[X, \Omega]})(y)
    = \bigsqcap_{\tree \in \runtree_n(\gamma, y)} \sigma(\tree)
    \sqsupseteq \bigsqcap_{m = 0}^{\infty} \bigsqcap_{\tree \in \runtree_m(\gamma, y)} \sigma(\tree)
    \]
    by Lem.~\ref{lem:Phi-run-tree}.
    Therefore, we have
    \[
    d_{n + 1}(y) = \bigsqcap_{m = 0}^{\infty} \bigsqcap_{\tree \in \runtree_m(\gamma, y)} \sigma(\tree).
    \]

    When the algorithm terminates, we have $X = S = \bigcup_{n} Y_n$.
    Therefore, for any $x \in X = S$, we have
    \[
    d(x)
    = \bigsqcap_{m = 0}^{\infty} \bigsqcap_{\tree \in \runtree_m(\gamma, x)} \sigma(\tree)
    = \gfp \Phi(x)
    \]
    by Lem.~\ref{lem:gfp-run-tree}.

    (\ref{item:correctness-G} $\implies$ \ref{item:expansive-G})
    Suppose that $\sigma$ is not expansive.
    Then, there exist a construction tree $\tree = (\tree, \ell)$ and a child $j \in \tree$ of $\epsilon$ such that $\sigma(\ell(j)) \sqsupset \sigma(\ell(\epsilon))$ holds.
    Without loss of generality, we can assume that $j = 0$.
    Let $X = \tree \setminus \{ \epsilon \}$ and $\gamma \colon X \to \BB \times \powsetfin(GX)$ be the contraction coalgebra defined in Definition~\ref{def:contraction-coalg}.
    Then, we can show that the coalgebraic Dijkstra algorithm does not return the greatest fixed point $\gfp \Bellman{G}{\sigma}{\gamma}$ of $\Bellman{G}{\sigma}{\gamma}$.
\end{proof}

\subsection{Details of Ex.~\ref{example:incorrectness-of-bardi-lopez-condition}}\label{appendix:example:incorrectness-of-bardi-lopez-condition}

We argue that the corrected condition \eqref{eq:bardi-lopez-condition-correct}  is essentially equivalent to our condition of expansiveness. Assume first that \eqref{eq:bardi-lopez-condition-correct} holds.  We want a uniform condition---one that works for every input---so that we require \eqref{eq:bardi-lopez-condition-correct}   for every $n \in \NN$. This implies that $\ell_0 = L$: $\ell_0 \le L$ holds by definition (Ex.~\ref{example:reachability-game-dynamic-game}); and, by letting  $n \to \infty$, we obtain $L\le \ell_{0}$ from \eqref{eq:bardi-lopez-condition-correct}. (Note that $\ell_0 = L$ forces all stepwise rewards to be the same. It is a strong constraint.)

 By letting $n = 0$ and $\ell_0 = L$ in \eqref{eq:bardi-lopez-condition-correct}, we obtain
\begin{equation}\label{eq:bardi-lopez-condition-correct-unbounded}
        \finalweight \le L + r \finalweight.
    \end{equation}
Now, using
 $\ell_0 = L$, the transition modality $\sigma_r$ is described by
    \begin{equation}\label{eq:transition-modality-dynamic-game}
        \sigma_r(A) = \max_{(a,b) \in A} (a + rb) = L + r \max_{(a,b) \in A} b \qquad \text{for $A \in \powsetfinne([L, L] \times \RRinfpos)$},
    \end{equation}
    and we have $\Omega^{\sigma_r} = \{ \frac{1-r^{n}}{1-r} L + r^{n} \finalweight \mid n \in \NN\} \cup \{ \infty\}$ (cf.\ Def.~\ref{def:Omega-sigma}). 
     The expansiveness of $\sigma_r$ amounts to
    \begin{equation}
        \forall n \in \NN. \quad
        \frac{1-r^{n}}{1-r} L + r^{n} \finalweight \le L + r \left( \frac{1-r^{n}}{1-r} L + r^{n} \finalweight \right),
    \end{equation}
    which is easily derived from  \eqref{eq:bardi-lopez-condition-correct-unbounded}.

    Conversely, let us assume that $\sigma_r$ is expansive.
    We have $\Omega^{\sigma_r} = \{ \infty \} \cup \bigcup_{n=0}^{\infty} [\ell_0 \sum_{j = 0}^{n-1}r^j + r^n \xi, L \sum_{j=0}^{\infty} r^j + r^n \xi]$.
    Expansiveness of $\sigma_r$ implies
    $b \le \sigma_r( \{ \ell_0, b\}) = \ell_0 + r b$---thus  $b \le \ell_0 / (1 - r)$---for every $n \in \NN$ and $b \in [\ell_0 \sum_{j = 0}^{n-1}r^j + r^n \xi, L \sum_{j=0}^{\infty} r^j + r^n \xi]$.
    By taking the maximum such $b$ (namely $b = L \sum_{j=0}^{n-1} r^j + r^n \xi$), we obtain \eqref{eq:bardi-lopez-condition-correct}.

\section{Proofs for \S\ref{subsec:complexity}}

\begin{proof}[Proof for Prop.~\ref{prop:complexity}]
    The number of iterations of the \textbf{while} loop is $\order(V)$.
    In each iteration of the \textbf{while} loop, it takes $\order(V)$ to compute $Y$.
    Hence, the time complexity of computation of $Y$ through the algorithm is $\order(V^2)$.

    Through the algorithm, each $y \in X$ is in $Y$ exactly once.
    Therefore, each $x \in X$ is in $P$ at most $\sizeof{\successor(x)}$ times, and $\Bellman{G}{\sigma}{\gamma}(d)(x)$ is computed at most $\sizeof{\successor(x)}$ times.
    Hence, it takes $\order(T\sum_{x\in X}\sizeof{\successor(x)}) = \order(TE)$ to compute $d$.

    Therefore, the time complexity is $\order(TE + V^2)$.
\end{proof}

\begin{proof}[Proof for Thm.~\ref{thm:complexity-fibonacci}]
    The size of the Fibonacci heap is at most $V$.

    Through the algorithm, each $y \in X$ is in $Y$ exactly once.
    Therefore, each $x \in X$ is in $P$ at most $\sizeof{\successor(x)}$ times, and $\sigma(Gd(a))$ where $a \in \pi_1(\gamma(x))$ is computed at most $\sizeof{\successor(x)}$ times.
    Hence, it takes $\order(T'\sum_{x\in X}\sizeof{\successor(x)}) = \order(T'E)$ to compute $d$.

    It takes $\order(1)$ amortized time to decrease a value of a key,
    and it takes $\order(1)$ time to insert a key.
    Since each $x$ is in $P$ at most $\sizeof{\successor(x)}$ times, the loop of Line~\ref{Line:update-insert} takes $\order(\sum_{x \in X}\sizeof{\successor(x)}) = \order(E)$ time through the algorithm.

    Since each $x \in X$ is in $Y$ exactly once, we get the key of the minimum values from $Q$ (Line~\ref{Line:get-Q}) exactly $V$ times and delete the key (Line~\ref{Line:delete-Q}) exactly $V$ times.
    It takes $\order(1)$ to get the minimum, and $\order(\log V)$ amortized time to delete the key of the minimum value.
    Hence, through the algorithm, it takes $\order(V)$ and $\order(V \log V)$ for Line~\ref{Line:get-Q} and Line~\ref{Line:delete-Q}, respectively.

    The time complexity of the whole algorithm is $\order(T'E + E + V + V \log V) = \order(T'E + V \log V)$.
\end{proof}

\fi
\end{document}